%% file: main.tex
\documentclass[11pt,a4paper]{article}

\usepackage[backend=biber, style=alphabetic, sorting=ynt, maxnames=99]{biblatex}
\usepackage{authblk}
\usepackage[left=1in, right=1in, top=1in, bottom=1in]{geometry}
\usepackage{dsfont}
\usepackage{amsthm}
\usepackage{amsfonts}
\usepackage{thmtools}
\usepackage{thm-restate}
\usepackage[pdftex]{hyperref}
\usepackage{physics}
\usepackage{mymacros}
\usepackage{amssymb}
\usepackage{amsmath}
\usepackage{cleveref}
\usepackage{xcolor}
\usepackage{graphicx}
\usepackage{xparse}
\usepackage{mathtools}
\usepackage{subcaption}
\PassOptionsToPackage{dvipdfmx}{graphicx}
\usepackage{adjustbox}[export]

\usepackage[toc,page]{appendix}
\usepackage{etoolbox}

\declaretheoremstyle[
numberwithin=chapter,
headfont=\normalfont\bfseries,
notefont=\normalfont\scshape, 
notebraces={$\lbrack$}{$\rbrack$},
bodyfont=\normalfont\itshape,
mdframed={
    splittopskip=20pt, 
    skipabove = 16pt, 
    innertopmargin=6pt,
    innerbottommargin=6pt},
postheadspace=\newline
]{thstyle}

\newtheoremstyle{remstyle}
  {12pt} 
  {12pt} 
  {} 
  {} 
  {\normalfont\bfseries\itshape} 
  {.} 
  {.5em} 
  {} 

\newtheoremstyle{defstyle}
  {12pt} 
  {12pt} 
  {} 
  {} 
  {\normalfont\bfseries} 
  {.} 
  {.5em} 
  {} 

\declaretheorem[
    style=thstyle,
    name=Theorem,
    numberwithin=section
]{theorem}
\declaretheorem[sibling = theorem]{lemma}
\declaretheorem[sibling = theorem, style = defstyle]{definition}
\declaretheorem[sibling = theorem]{proposition}
\declaretheorem[sibling = theorem, style = remstyle]{remark}

\blocktheorem{theorem}

\renewbibmacro{in:}{}
\title{Classification of the anyon sectors of Kitaev's quantum double model}
\author[1]{Alex Bols \thanks{email: \href{abols01@phys.ethz.ch}{abols01@phys.ethz.ch}}}
\author[2]{Siddharth Vadnerkar \thanks{email: \href{svadnerkar@ucdavis.edu}{svadnerkar@ucdavis.edu}; site: \href{https://sites.google.com/view/siddharthvadnerkar}{https://sites.google.com/view/siddharthvadnerkar}}}
\affil[1]{Institute for Theoretical Physics, ETH Z{\"u}rich}
\affil[2]{Department of Physics, University of California, Davis}

\begin{document}

\setcounter{tocdepth}{2}
\appto\appendix{\addtocontents{toc}{\protect\setcounter{tocdepth}{1}}}

\appto\listoffigures{\addtocontents{lof}{\protect\setcounter{tocdepth}{1}}}
\appto\listoftables{\addtocontents{lot}{\protect\setcounter{tocdepth}{1}}}

\date{\today}

\maketitle
\begin{abstract}
We give a complete classification of the anyon sectors of Kitaev's quantum double model on the infinite triangular lattice and for finite gauge group $G$, including the non-abelian case. As conjectured, the anyon sectors of the model correspond precisely to equivalence classes of irreducible representations of the quantum double algebra of $G$.
\end{abstract}

\tableofcontents

\subsection*{Acknowledgements}
A.B. was supported by the Simons Foundation for part of this project. S.V. was supported by the NSF grant DMS-2108390. We thank Yoshiko Ogata for pointing out a serious error in the first version of the manuscript. We also thank Gian Michele Graf, Pieter Naaijkens, and Bruno Nachtergaele for useful conversations.

\subsection*{Data availability and conflict of interest}
Data availability is not applicable to this article as no new data were created or analyzed in this study. There are no known conflicts of interest.

\input{intro}

\input{excited_states}

\input{anyon_representations}

\input{completeness}

\input{conclusions}

\nocite{*}

\appendix

\input{ribbonprops}

\input{string_nets}

\printbibliography

\end{document}

%% file: intro.tex
\section{Introduction}

Over the decades since the discovery of the integer quantum Hall effect, the notion of topological phases of matter has come to be a central paradigm in condensed matter physics. In contrast to the conventional Landau-Ginsburg paradigm of spontaneous symmetry breaking, topological phases of matter are not distinguished by any local order parameter. Instead they are characterised by a remarkably wide variety of topological properties, ranging from toplogically non-trivial Bloch bands to topological ground-state degeneracy. What all these topological materials seem to have in common is that they are characterised by robust patterns in the entanglement structure of their ground states \cite{Li2008-le, Fidkowski2010-gj}.

Within this zoo of topological phases, the \emph{topologically ordered} phases in two dimensions have received a great deal of attention. The reason for this is in part because of their possible applications to quantum computation \cite{Kitaev2003-qr, Freedman1998-tz, Nayak2008-ef}. Topologically ordered materials exhibit robust ground state degeneracy depending on the genus of the surface on which they sit, and they support anyonic excitations which have braiding statistics that differs from that of bosons or fermions.

With the ever increasing experimental control of quantum many-body systems in the lab in mind, it is desirable to understand topological order from a microscopic point of view. On the one hand, an important role is played in this endeavor by exactly solvable quantum spin models that exhibit topologcial order, such as Kitaev's quantum double models \cite{Kitaev2003-qr} and, more generally, the Levin-Wen models \cite{levin2005string}. On the other hand, one wants to obtain a good understanding of the mathematical structures involved in characterising topological orders in general models \cite{Kitaev2006-ts, 
Shi2019-tl, Kawagoe2020-vt}. The latter problem has proven to be a rich challenge for mathematical physics \cite{Naaijkens2010-aq, Naaijkens2012-fh, Cha2018-ke, Cha2020-rz, Ogata2022-wp}. These works have yielded a rigorous, albeit still incomplete, description of topological order in gapped quantum spin systems in two dimensions. They provide robust definitions of anyon types, their fusion rules, and their braiding statistics, as well as a rigorous understanding of how these data fit together in a braided $C^*$-tensor category.

In this paper we study Kitaev's quantum double models from this mathematical physics point of view. The quantum double models can be thought of as discrete gauge theories with a finite gauge group $G$. These models are of particular interest because for non-abelian $G$, they are paradigmatic examples of models that support \emph{non-abelian anyons} \cite{Kitaev2003-qr}. We take a first step towards integrating the quantum double models for general $G$ into the mathematical framework referred to above. In particular, we classify all the anyon types of these models.

Roughly, an anyon type corresponds to a superselection sector, \ie a unitary equivalence class of representations of the observable algebra that are unitarily equivalent to the ground state representation when restricted to the complement of any cone-like region of the plane. We call such sectors \emph{anyon sectors}. Intuitively an anyon sector contains states that can be made to look like the ground state locally by moving the anyon somewhere else, but globally, the anyon is always detectable by braiding operations.

In order to completely classify the anyon sectors of the quantum double model we construct states $\omega_{s}^{RC;u}$ labeled by an irreducible representation $RC$ of the quantum double algebra $\caD(G)$, a site $s$, and additional microscopic data $u$. These states look like the ground state when evaluated on any local observable whose support does not contain or encircle the site $s$. We characterise these states by showing that they are the unique states that satisfy certain local constraints depending on the site $s$ and the data $RC$ and $u$. In particular, the states $\omega_s^{RC;u}$ are pure. In the particular case where $RC$ corresponds to the trivial representation of the quantum double algebra, the state $\omega_s^{RC;u}$ is the frustration free ground state, so we get existence and uniqueness of the frustration free ground state as a corollary, a result which was first proven in \cite{naaijkens2012anyons}.

We continue by showing that the pure states $\omega_{s}^{RC;u}$ belong to different superselection sectors if and only if they differ in their $RC$ label. It follows that the GNS representations of the states $\omega_{s}^{RC;u}$ give us a collection of pairwise disjoint irreducible representations $\pi^{RC}$ labeled by irreducible representations of the quantum double algebra. By relating the representations $\pi^{RC}$ to so-called amplimorphism representations \cite{Naaijkens2015-xj, vecsernyes1994quantum, fuchs1994quantum, nill1997quantum}, we show that these representations do in fact belong to anyon sectors. Finally, we show that any anyon sector must contain one of the states $\omega_s^{RC;u}$, thus showing that all anyon sectors contain one of the $\pi^{RC}$.

The paper is structured as follows. In Section \ref{sec:results} we set up the problem and state our main results. In Section \ref{sec:excited states} we construct the states $\omega_{s}^{RC;u}$ that `contain an anyon' at site $s$ and prove that these states are pure. Section \ref{sec:anyon representations} is devoted to constructing for each irreducible representation $RC$ of the quantum double algebra a representation $\pi^{RC}$ of the observable algebra that contains the states $\{\omega_s^{RC;u}\}_{s, u}$, and proving that these representations are disjoint. Finally, in section \ref{sec:completeness} we show that any anyon sector contains one of the $\pi^{RC}$, thus showing that the $\pi^{RC}$ exhaust all anyon sectors of the model.

\section{Setup and main results}
\label{sec:results}

\subsection{Algebra of observables}

Let $\Gamma$ be the regular triangular lattice (see Figure \ref{fig:oriented lattice snapshot}) whose set of vertices $\latticevert$ we regard as a subset of the plane $\R^2$ such that nearest neighbouring vertices are separated by unit distance.

The set of oriented edges of $\Gamma$ is identified with the set of ordered pairs of neighbouring vertices:
$$\orientededges := \{ (v_0, v_1) \in \latticevert \times \latticevert \, : \, v_0 \, \text{and} \, v_1 \, \text{are nearest neighbours}  \}.$$
We let $\latticeedge \subset \orientededges$ consist of the oriented edges pointing from left to right as in Figure \ref{fig:oriented lattice snapshot}. Note that $\latticeedge$ contains exactly one oriented edge for every edge of $\Gamma$. We denote the set of faces of $\Gamma$ by $\latticeface$.

\begin{figure}
    \centering
    \includegraphics[ width=0.4\textwidth]{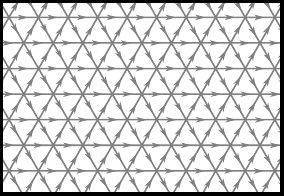}
    \caption{A snapshot of the triangular lattice with all edges oriented towards the right.}
    \label{fig:oriented lattice snapshot}
\end{figure}

Any oriented edge $e = (v_0, v_1)$ has an initial vertex $\partial_0 e = v_0$, a final vertex $\partial_1 e = v_1$, and an opposite oriented edge $\bar e = (v_1, v_0)$. The vertices $\latticevert$ are equipped with the graph distance $\mathrm{dist}( \cdot, \cdot )$ and similarly for the faces (regarded as elements of the dual graph). \\

We fix a finite group $G$ and associate to each edge $e \in \latticeedge$ a Hilbert space $\hilb_e = \mathbb{C}^{|G|}$ and a matrix algebra $\cstar[e] = \mathrm{End}(\hilb_e)$. For any finite $S \subset \latticeedge$  we have a Hilbert space $\hilb_S = \otimes_{e \in S} \hilb_e$ and the algebra of operators $\cstar[S] = \mathrm{End}(\hilb_S)$ on this space.

We employ the following graphical representation of states $\ket{\alpha}$. For any edge $e$, the basis state $\ket{g}_{e}$ of $\hilb_e$ is represented by the edge $e$ being crossed from right to left by an oriented \emph{string} labeled $g$. An equivalent representation of $\ket{g}_e$ is the edge $e$ being crossed from left to right by a string labeled $\bar g$, see Figure \ref{fig:graphical_representation}. The basis element $\ket{1}_e$ is represented by the edge $e$ not being crossed by any string at all. A tensor product of several of such basis states is represented by a figure where each participating edge is crossed by a labeled oriented string by the rules just described. See Figure \ref{fig:graphical_representation} for an example.

\begin{figure}
    \centering
    \includegraphics[width=0.8\textwidth]{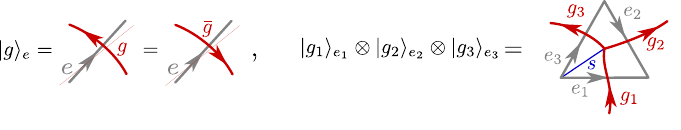}
    \caption{A graphical representation of the basis vector $\ket{g}_e \in \hilb_e$ and of the tensor product vector $\ket{g_1}_{e_1} \otimes \ket{g_2}_{e_2} \otimes \ket{g_3}_{e_3} \in \hilb_{\{e_1, e_2, e_3\}}$ for edges $e_1, e_2, e_3$ belonging to a single face $f$. We have also indicated a site $s$ with $f(s) = f$ in blue (see Section \ref{sec:excited states}).}
    \label{fig:graphical_representation}
\end{figure}

Let $S_1, S_2 \subset \latticeedge$ be finite sets of edges such that $S_1 \subset S_2$, then there is a natural embedding $\iota_{S_1, S_2}: \cstar[S_1] \hookrightarrow \cstar[S_2]$ given by tensoring with the identity, i.e.
$$\iota_{S_1, S_2}(O) = O \otimes \mathds{1}_{\cstar[{S_2 \setminus S_1}]}$$
for all $O \in \cstar[S_1]$. With these embeddings the algebras $\cstar[S]$ for finite $S \subset \latticeedge$ form a directed system of matrix algebras. Their direct limit is called the \emph{local algebra}, and is denoted by $\cstar[loc]$. The norm closure of the local algebra
$$\cstar = \overline{\cstar[loc]}^{||\cdot||}$$
is called the \emph{quasi-local algebra} or \emph{observable algebra}.

Similarly, for any (possibly infinite) $S \subset \latticeedge$ we have the algebra $\cstar[S] \subset \cstar$ of quasi-local observables supported on $S$.\\

A \emph{state} on $\cstar$ is a positive linear functional $\omega: \cstar \rightarrow \mathbb{C}$ with $\omega(\mathds{1}) =1$. Given a state $\omega$ on $\cstar$ there is a representation $\pi_{\omega} : \cstar \rightarrow \mathcal{B}({\hilb_{\omega}})$ for some separable Hilbert space $\hilb_{\omega}$ containing a unit vector $\ket{\Omega}$ that is cyclic for the representation $\pi_{\omega}$ and such that $\omega(O) = \langle \Omega, \pi_{\omega} (O) \Omega\rangle$ for all $O \in \cstar$. The triple $(\pi_{\omega}, \hilb_{\omega}, \ket{\Omega})$ satisfying these properties is unique up to unitary equivalence, and is called the GNS triple of the state $\omega$.\\

Throughout this paper, we will use the word `projector' to mean a self-adjoint operator that squares to itself. A collection of projectors is called \emph{orthogonal} if the product of each pair of projectors in the set vanishes. A set of projectors is called \emph{commuting} if each pair of projectors in the set commutes with each other.

\subsection{The quantum double Hamiltonian and its frustration free ground state}
\label{sec:model}

We say an edge $e$ belongs to a face $f$ and write $e \in f$ when $e$ is an edge on the boundary of $f$. Similarly, we say a vertex $v$ belongs to $f$ and write $v \in f$ if $v$ neighbours $f$, and we say a vertex $v$ belongs to an edge $e$ and write $v \in e$ if $v$ is the origin or endpoint of $e$.\\

We fix for each edge $e$ an orthonormal basis $\{ \ket{g} \}_{g \in G}$ for $\hilb_e$ labeled by elements of the group $G$. For $g \in G$ we denote its inverse by $\overline{g}$, and we define the left group action $L_e^h := \sum_{g \in G} \ket{hg}\bra{g}$, the right group action $R^h_e := \sum_{g \in G} \ket{g \bar h} \bra{g}$, and the projectors $T_e^g := | g \rangle \langle g |$.

For each vertex $v$ and edge $e$ such that $v$ belongs to $e$ we set $L^h(e, v) = L^h_e$ if $\partial_0 e = v$ and $L^h(e, v) = R^{h}_e$ if $\partial_1 e = v$. For each $h \in G$ we define a unitary $A_v^h := \prod_{e \in v} L^h(e, v)$. These are called the \emph{gauge transformations at $v$}. Graphically,
\begin{center}
    \includegraphics[width=0.6\textwidth]{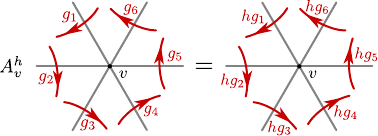}
\end{center}
We then define the \emph{gauge constraint} $A_v := \frac{1}{\abs{G}} \sum_{h \in G} A_v^h$, which is the projector enforcing gauge invariance at the site $v$. Similarly, for each face $f$ and edge $e \in f$ we set $T^h(e, f) = T_e^h$ if $f$ is to the left of $e$, and $T^h(e, f) = T_e^{\bar h}$ if $f$ is to the right of $e$. If the face $f$ has bounding edges $(e_1, e_2, e_2)$ ordered counterclockwise (with arbitrary initial edge), then we define the \emph{flat gauge projector} $B_f := \sum_{ \substack{ g_1, g_2, g_3 \in G : \\ g_1 g_2 g_3 = 1 }} T^{g_1}(e_1, f) \, T^{g_2}(e_2, f) \, T^{g_3}(e_3, f)$ which is also a projector. Note that this expression does not depend on which edge goes first in the triple $(e_1, e_2, e_3)$. Graphically,
\begin{center}
    \includegraphics[width=0.6\textwidth]{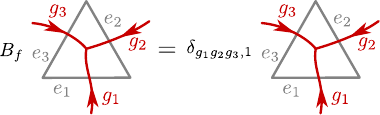}
\end{center}
 The set of projectors $\{ A_v \}_{v \in \latticevert} \cup \{ B_f \}_{f \in \latticeface}$ all commute with each other.\\

The quantum double Hamiltonian is the formal sum (interaction) of commuting projectors 
\begin{equation*}
    H = \sum_{v \in \latticevert} ( \mathds{1} - A_v ) + \sum_{f \in \latticeface} ( \mathds{1} - B_f ).
\end{equation*}

\begin{definition} \label{def:ffgs}
    A state $\omega$ on $\cstar$ is a \ffgs{} of $H$ if $\omega(A_v) = \omega(B_f) = 1$ for all $v \in \latticevert$ and all $f \in \latticeface$.
\end{definition}

If a state $\omega$ satisfies $\omega(A_v) = 1$ then we say it is \emph{gauge-invariant} at $v$, and if $\omega(B_f) = 1$ then we say it is \emph{flat} at $f$. A \ffgs{} is a state that is gauge invariant and flat everywhere.

The following Proposition was first proven for the Toric code (the case $G = \Z_2$) in \cite{alicki2007statistical}, and for general $G$ in \cite{naaijkens2012anyons}. See also \cite{chuah2024boundary} which uses general results on commuting projector Hamiltonians from \cite{jones2023local}. We obtain a new proof of this Proposition as a Corollary to Proposition \ref{prop:characterisation of S^RCu}.
\begin{proposition}[\cite{naaijkens2012anyons}]
\label{prop:ffgsunique}
    The quantum double Hamiltonian has a unique frustration free ground state.
\end{proposition}
We will denote the unique \ffgs{} by $\omega_0$, and let $(\pi_0, \hilb_0, \ket{\Omega_0})$ be its GNS triple. Note that since $\omega_0$ is the \emph{unique} \ffgs{} of the quantum double model it is a pure state, and therefore $\pi_0$ is an irreducible representation.

\subsection{Classification of anyon sectors}

In the context of infinite volume quantum spin systems or field theories, types of anyonic excitations over a ground state $\omega_0$ have a very nice mathematical characterisation. They correspond to the irreducible representations of the observable algebra that satisfy a certain superselection criterion w.r.t. the GNS representation $\pi_0$ of the ground state (\cite{Doplicher1971-jd}, \cite{Doplicher1974-hb}, \cite{fredenhagen1989superselection}, \cite{fredenhagen1992superselection}, \cite{frohlich1990braid}). In our setting of quantum spin systems, the appropriate superselection criterion was first formulated in \cite{Naaijkens2010-aq} in the special case of the Toric code.\\

The cone with apex at $a \in \R^2$, axis $\hat v \in \R^2$ of unit length, and opening angle $\theta \in (0, 2\pi)$ is the open subset of $\R^2$ given by
\begin{equation*}
 	\Lambda_{a, \hat v, \theta} := \{ x \in \R^2 \, : \, (x - a) \cdot \hat v > \norm{x-a} \cos (\theta/2)   \}. \quad\quad\quad\quad \adjincludegraphics[width=2.5cm,valign=c]{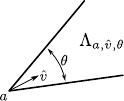}
\end{equation*}
Any subset $\Lambda \subset \R^2$ of this form will be called a cone.

For any subset $S \subset \R^2$ of the plane we denote by $\overline S \subset \latticeedge$ the set of edges whose midpoints lie in $S$, and write $\cstar[S] := \cstar[\overline S]$ for the algebra of observables supported on $\overline S$. With this definition we have $\overline S \cup \overline{ S^c } = \latticeedge$ for any $S \subset \R^2$.

\begin{definition}
\label{def:anyonsector}
    An irreducible representation $\pi: \cstar \rightarrow \mathcal{B}(\hilb)$ is said to satisfy the superselection criterion w.r.t. $\pi_0$ if for any cone $\Lambda$, there is a unitary $U_\Lambda : \hilb_0 \rightarrow \hilb$ such that
    $$\pi(A) = U_\Lambda \pi_0 (A) U^*_\Lambda$$
    for all $A\in \cstar[\Lambda]$. We will call such a representation an \emph{anyon representation} for $\omega_0$. A unitary equivalence class of anyon representations we call an \emph{anyon sector}.
\end{definition}

Let us denote by $\caD(G)$ the quantum double algebra of $G$. The irreducible representations of $\caD(G)$ are, up to isomorphism, uniquely labeled by a conjugacy class $C$ of $G$ together with an irreducible representation $R$ of the centralizer $N_C$ of $C$ (see for example \cite{Gould1993-bt}). We denote the irreducible representation of $\caD(G)$ corresponding to conjugacy class $C$ and irreducible representation $R$ by $RC$.

Our main result is the complete classification of the anyon sectors of $\omega_0$ in terms of the irreducible representations of $\caD(G)$. This result was first obtained for the Toric code in \cite{naaijkens2013kosaki} using very different methods. See also \cite{Fiedler2015-na} where it is indicated how the methods of \cite{naaijkens2013kosaki} can be applied to quantum double models for abelian $G$.
\begin{theorem} \label{thm:main theorem}
    For each irreducible representation $RC$ of $\caD(G)$ there is an anyon representation $\pi^{RC}$. The representations $\{ \pi^{RC} \}_{RC}$ are pairwise disjoint, and any anyon representation is unitarily equivalent to one of them.
\end{theorem}
This Theorem is restated and proven in Section \ref{sec:completeness}, Theorem \ref{thm:classification theorem}.

%% file: excited_states.tex
\section{Anyon states}
\label{sec:excited states}

In this section we construct states containing a single anyonic excitation of arbitrary type, and show that these states are pure. These are states that satisfy the frustration free ground state constraints everywhere except at a fixed site $\site$, where instead they are constrained by some Wigner projector onto an irreducible representation for the quantum double action on that site (see Remark \ref{rem:quantum double at s}). We completely classify the states satisfying such constraints by first classifying all states that satisfy appropriate local versions of these constraints. We note that the methods presented in this section are sufficient to establish that anyon types of the quantum double model are in one-to-one correspondence with the irreducible representations of the quantum double algebra $\caD(G)$ in the context of the entanglement bootstrap program \cite{Shi2019-tl}.\\

\subsection{Ribbon operators, gauge configurations, and gauge transformations} \label{subsec:preliminary notions}

In this subsection we introduce sites, triangles, and ribbons in order to then describe the ribbon operators introduced by \cite{Kitaev2003-qr}. These ribbon operators will play a crucial role from Section \ref{sec:anyon representations} onward. We then introduce the notion of local gauge transformations.

Let $\Gamma^*$ be the dual lattice to $\Gamma$. To each oriented edge $e \in \orientededges$ we associate a unique oriented dual edge $\de \in (\vec{\Gamma^*})^E$ with orientation such that along $e$, the dual edge $\de$ passes from right to left. Here
$$(\vec{\Gamma^*})^E = \{ (f_0, f_1) \in \latticeface \times \latticeface \, : \, f_0 \, \text{and} \, f_1 \, \text{are neighbouring faces}  \}$$
is the set of oriented dual edges of $\Gamma$.

\subsubsection{Sites and triangles}

A \emph{site} $s$ is a pair $s = (v, f)$ of a vertex and a face such that $v$ is on the boundary of $f$. We write $v(s) = v$ for the vertex of $s$ and $f(s) = f$ for the face of $s$. We represent a site graphically by a line from the site's vertex to the center of its face.

A direct triangle $\tau = (s_0, s_1, e)$ consists of a pair of sites $s_0, s_1$ that share the same face, and the edge $e \in \latticeedge$ that connects the vertices of $s_0$ and $s_1$. We write $\partial_0 \tau = s_0$ and $\partial_1 \tau = s_1$ for the initial and final sites of the direct triangle $\tau$, and $e_{\tau} = (v(s_0), v(s_1))$ for the oriented edge associated to $\tau$, see Figure \ref{fig:direct and dual triangles}. Note that $e$ and $e_{\tau}$ need not be the same, $e$ always has the left to right orientation used in the definition of $\latticeedge$ while $e_{\tau}$ is oriented in the direction of the direct triangle. The opposite triangle to $\tau$ is the direct triangle $\bar \tau = (s_1, s_0, e)$. A direct triangle $\tau = (s_0, s_1, e)$ is \emph{positive} if the face $f = f(s_0) = f(s_1)$ lies to the left of $e_{\tau}$ and \emph{negative} otherwise.

Similarly, a dual triangle $\tau = (s_0, s_1, e)$ consists of a pair of sites $s_0, s_1$ that share the same vertex, and the edge $e \in \latticeedge$ whose associated dual edge $\de$ connects the faces of $s_0$ and $s_1$. We write again $\partial_0 \tau = s_0$ and $\partial_1 \tau = s_1$ and write $e^*_{\tau} = (f(s_0), f(s_1))$ for the oriented dual edge associated to $\tau$. We also write $e_{\tau}$ for the oriented edge whose dual is $e_{\tau}^*$. Note again that $e^*$ and $e_{\tau}^*$ need not be the same. The orientation of $e^*$ is determined by the left to right orientation of $e \in \latticeedge$ while $e^*_{\tau}$ is oriented in the direction of the dual triangle. We define the opposite dual triangle by $\bar \tau = (s_1, s_0, e)$. A dual triangle $\tau = (s_0, s_1, e)$ is called \emph{positive} if the vertex $v = v(s_0) = v(s_1)$ lies to the right of $e^*_{\tau}$ and $\emph{negative}$ otherwise.

\begin{figure}
    \centering
    \includegraphics[ width=0.4\textwidth]{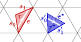}
    \caption{A direct triangle $\tau = (s_0, s_1, e)$ in red, and a dual triangle $\tau^* = (s'_0, s'_1, e')$ in blue. The dual edge $e^*$ associated to $e'$ is indicated as a dotted blue line. Note  that for this particular dual triangle, $e^*$ is oriented opposite to the white arrow, which instead follows the orientation of $(e_{\tau^*})^*$.}
    \label{fig:direct and dual triangles}
\end{figure}

To each dual triangle $\tau = (s_0, s_1, e)$ we associate unitaries $L^h_{\tau}$ supported on the edge $e$. The way $L^h_{\tau}$ acts depends on whether the edge $e^*$ dual to $e$ satisfies $e^* = (f(s_0), f(s_1))$ or $e^* = (f(s_1), f(s_0))$, and on whether $v(s_0) = \partial_0 e$ or $v(s_0) = \partial_1 e$ as follows. If $e^* = (f(s_0), f(s_1))$ and $v(s_0) = \partial_0 e$ then $L^h_{\tau} := L_e^h$. If $e^* = (f(s_0), f(s_1))$ and $v(s_0) = \partial_1 e$ then $L^h_{\tau} := R_e^{\bar h}$. If $e^* = (f(s_1), f(s_0))$ and $v(s_0) = \partial_0 e$ then $L^h_{\tau} := L_e^{\bar h}$ ($L^h_e, R^h_e$ were defined in section \ref{sec:model}). Finally, If $e^* = (f(s_1), f(s_0))$ and $v(s_0) = \partial_1 e$ then $L^h_{\tau} := R_e^h$. Similarly, to each direct triangle $\tau = (s_0, s_1, e)$ we associate projectors $T^g_\tau := T^g_{e}$ if $e = (v(s_0), v(s_1))$ and $T^g_{\tau} := T_e^{\bar g}$ if $e = (v(s_1), v(s_0))$.

\subsubsection{Ribbons}

We define a finite ribbon $\rho:= \{\tau_i\}_{i = 1}^l$ to be an ordered tuple of triangles such that $\partial_1 \tau_i = \partial_0 \tau_{i+1}$ for all $i = 1, \cdots, l-1$, and such that for each edge $e \in \latticeedge$ there is at most one triangle $\tau_i$ for which $\tau_i = (\partial_0 \tau_i, \partial_1 \tau_i, e)$. The empty ribbon is denoted by $\ep$. For non-empty ribbons $\rho$ we write $\partial_0 \rho := \partial_0 \tau_1$ for the initial site of the ribbon and $\partial_1 \rho := \partial_1 \tau_n$ for the final site of the ribbon. See Figure \ref{fig:finite ribbon}. If all triangles belonging to a ribbon $\rho$ are direct, we say that $\rho$ is a direct ribbon, and if all triangles belonging to a ribbon $\rho$ are dual, we say that $\rho$ is a dual ribbon.

A ribbon is said to be positive if all of its triangles are positive, and negative if all of its triangles are negative. All non-empty ribbons are either positive or negative.

\begin{figure}
    \centering
    \includegraphics[ width=0.6\textwidth]{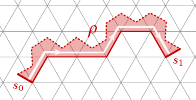}
    \caption{A positive finite ribbon $\rho$ with $\partial_0 \rho = s_0$ and $\partial_1 \rho = s_1$. The direct path of $\rho$ consists of the solid red edges.}
    \label{fig:finite ribbon}
\end{figure}

If we have two tuples $\rho_1 = \{ \tau_i \}_{i=1}^{l_1}$ and $\rho_2 = \{ \tau_i \}_{i = l_1 + 1}^{l_2}$ then we can concatenate them to form a tuple $\rho = \{ \tau_i \}_{i = 1}^{l_1 + l_2}$. We denote this concatenated tuple by $\rho = \rho_1 \rho_2$. Note that if $\rho$ is a finite ribbon and $\rho = \rho_1 \rho_2$, then $\rho_1$ and $\rho_2$ are automatically finite ribbons and (if $\rho_1$ and $\rho_2$ are non-empty), $\partial_0 \rho = \partial_0 \rho_1$, $\partial_1 \rho = \partial_1 \rho_2$ and $\partial_1 \rho_1 = \partial_0 \rho_2$.

The orientation reversal of a ribbon $\rho = \{ \tau_i \}_{i = 1}^l$ is the ribbon $\bar \rho = \bar\tau_l \cdots \bar \tau_1$. We say a finite ribbon $\rho$ is \emph{closed} if $\partial_0 \rho = \partial_1 \rho$.

The support of a ribbon $\rho = \{ \tau_i = (s_0^{(i)}, s_1^{(i)}, e_i)  \}_{i=1}^l$ is $\supp(\rho) := \{ e_i \}_{i = 1}^l.$ If $\supp(\rho) \subseteq S \subseteq \latticeedge$ we say $\rho$ is supported on $S$.

\subsubsection{Direct paths}
\label{sec:direct paths}
A direct path $\gamma = \{ e_i \}_{i=1}^l$ is an ordered tuple  of oriented edges $e_i \in \orientededges$ such that $\partial_1 e_i = \partial_0 e_{i+1}$ for $i = 1, \cdots, l-1$. We write $\partial_0 \gamma = \partial_0 e_1$ for the initial vertex, and $\partial_1 \gamma = \partial_1 e_l$ for the final vertex of $\gamma$. Given two direct paths $\gamma_1 = \{ e_i \}_{i = 1}^{l_1}$ and $\gamma_2 = \{ e_i \}_{i= l_1 + 1}^{l_1 + l_2}$ such that $\partial_1 \gamma_1 = \partial_0 \gamma_2$ we can concatenate them to form a new direct path $\gamma = \{ e_i \}_{i = 1}^{l_1+l_2}$. We denote the concatenated path by $\gamma = \gamma_1 \gamma_2$. The orientation reversal of a direct path $\gamma = \{ e_i \}_{i=1}^l$ is the direct path $\bar \gamma = \bar e_l \cdots \bar e_1$. The support of a direct path $\gamma = \{ e_i \}_{i=1}^l$ is
$$\supp(\gamma) := \{ e \in \latticeedge \, : \, e = e_i \,\, \text{or} \,\, \bar e = e_i \,\, \text{for some} \,\, i = 1, \cdots, l  \}.$$
If $\supp(\gamma) \subseteq S \subseteq \latticeedge$ we say $\gamma$ is supported on $S$.

To each ribbon $\rho$ we can associate a direct path as follows. Let $\rho = \{ \tau_1, \cdots, \tau_l  \}$ be a finite ribbon, and let $J = \{ j_1, \cdots, j_{l'} \} \subset \{1, \cdots, l\}$ be the ordered subset such that $\tau_{j}$ is a direct triangle if and only if $j \in J$. Then $\rho^{dir} := \{ e_{\tau_j} \,: \, j \in J  \}$ is the direct path of $\rho$. To see that this is indeed a direct path, take indices $j_{\nu}, j_{\nu+1} \in J$ and suppose $j_{\nu+1} = j_{\nu} + m$. Then we want to show that $\partial_1 e_{\tau_{j_{\nu}}} = \partial_0 e_{\tau_{j_{\nu+1}}}$. By construction all triangles $\tau_{j_{\nu} + n}$ for $n = 1, \cdots, m-1$ are dual and therefore $v = v(\partial_0 \tau_{j_{\nu}+n}) = v(\partial_1 \tau_{j_{\nu}+n})$ are equal for all these $n$. We therefore have $\partial_1 e_{\tau_{j_{\nu}}} = v(\partial_1 \tau_{j_{\nu}}) = v$ and $\partial_0 e_{\tau_{j_{\nu}+m}} = v(\partial_0 \tau_{j_{\nu+1}}) = v$ as required. We have $\bar \rho^{dir} = \overline{ \rho^{dir} }$, and if $\rho$ is supported in $S \subseteq \latticeedge$ then $\rho^{dir}$ is also supported in $S$.

\subsubsection{Ribbon operators}

Here we describe the ribbon operators introduced by \cite{Kitaev2003-qr}, and state some of their elementary properties. For proofs and many more properties, see Appendix \ref{app:ribbonprops} of this paper or appendices B and C of \cite{Bombin2007-uw}. To each ribbon $\rho$ we associate a ribbon operator $F^{h,g}_\rho$ as follows. If $\epsilon$ is the trivial ribbon, then we set $F_{\epsilon}^{h, g} = \delta_{1, g} \mathds{1}$. For ribbons composed of a single direct triangle $\tau$ we put $F^{h,g}_\tau =  T_\tau^g$. For ribbons composed of a single dual triangle $\tau$, we put $F^{h,g}_{\tau} = \delta_{g,1} L_{\tau}^h$. For longer ribbons the ribbon operators are defined inductively as
\begin{equation} \label{eq:F inductive def}
    F^{h,g}_\rho = \sum_{k \in G} F^{h, k}_{\rho_1} F^{\overline{k}hk,\overline{k}g}_{\rho_2}
\end{equation}
for $\rho = \rho_1 \rho_2$. It follows from the discussion at the beginning of appendix \ref{app:ribbonprops} that this definition is independent of the way $\rho$ is split into two smaller ribbons. By construction, the ribbon operator $F_{\rho}^{h, g}$ is supported on $\supp(\rho)$. Let us define
$$T^g_\rho := F^{1,g}_\rho, \qquad L^h_\rho := \sum_{g \in G} F^{h,g}_\rho$$
so that $F^{h,g}_\rho = L^h_\rho T^g_\rho = T^g_\rho L^h_\rho$ (Lemma \ref{eq:F breaks into LT}).\\

We define gauge transformations $A_s^h$ and flux projectors $B_s^g$ at site $s$ in terms of the ribbon operators as follows:
$$A_s^h := F^{h,1}_{\rho_{\star}(s)}, \qquad B_s^g := F^{1, g}_{\rho_{\triangle}(s)}$$
where $\rho_{\triangle}(s)$ (resp. $\rho_{\star}(s)$) is the unique counterclockwise closed direct (dual) ribbon with end sites at $s$, see Figure \ref{fig:elementary direct and dual ribbons}. It is easily verified that $A_s^{h_1} A_s^{h_2} = A_s^{h_1 h_2}$ for all $h_1, h_2 \in G$, so the gauge transformations at $s$ form a representation of $G$. Similarly, one verifies that $B_s^{g_1} B_s^{g_2} = \delta_{g_1, g_2} B_s^{g_1}$ for all $g_1, g_2 \in G$. We further note that the gauge transformations $A_s^h$ depend only on the vertex $v(s)$, so we may put $A_v^h := A_s^h$ for any site $s$ such that $v = v(s)$ and speak of the gauge transformations at the vertex $v$. Similarly, the projectors $B_s^1$ onto trivial flux depend only on the face $f(s)$ so we may put $B_f^1 = B_{s}^1$ for any site $s$ such that $f = f(s)$.

\begin{remark} \label{rem:quantum double at s}
    For each site $s$ the operators $A_s^h$ and $B_s^g$ generate a realisation of the \emph{quantum double algebra} of $G$. This fact justifies the name of the model, and will be central to our analysis.
\end{remark}

The projectors $A_v, B_f$ appearing in the quantum double Hamiltonian can now be written as follows:
$$ A_v = \frac{1}{\abs{G}} \sum_h A_v^{h}, \quad B_f = B_f^1.$$
They are the projectors onto states that are gauge invariant at $v$, and that have trivial flux at $f$, respectively.

\begin{figure}
    \centering
    \includegraphics[ width=0.4\textwidth]{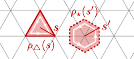}
    \caption{The elementary direct ribbon $\rho_{\triangle}(s)$ associated to the site $s$, and the elementary dual ribbon $\rho_{\star}(s')$ associated to the site $s'$.}
    \label{fig:elementary direct and dual ribbons}
\end{figure}

\subsubsection{Gauge configurations and gauge transformations}

It is very helpful to think of the frustration free ground state of the quantum double model as a string-net condensate, see \cite{levin2005string}. In what follows we establish the language of string-nets, which in this case correspond to gauge configurations. For any $S \subseteq \latticeedge$ we denote by $\gc[S]$ the set of maps $\alpha : S \rightarrow G$. We will denote by $\alpha_e$ the evaluation of $\alpha$ on an edge $e \in S$. We call such maps \emph{gauge configurations} on $S$. Let us write $\orientededges[S] = \{ e \in \orientededges \, : \, e \in S \,\, \text{or} \,\, \bar e \in S \}$ for the set of oriented edges corresponding to $S$. Any gauge configuration $\alpha$ on $S \subseteq \latticeedge$ extends to a function $\alpha : \orientededges[S] \rightarrow G$ on oriented edges by setting $\alpha_{\bar e} = \bar \alpha_{e}$. The meaning of $\alpha_e$ is the parallel transport of a discrete gauge field as one traverses the edge $e$.\\ 

For any finite $V \subset \latticevert$ we define $\gauge[V]$ to be the group of unitaries generated by $\{ A_{v}^{g_v} \, : \, v \in V , \,\,  g_v \in G  \}$. Since $A_v^g$ and $A_{v'}^{g'}$ commute whenever $v \neq v'$, any element $U \in \gauge[V]$ is uniquely determined by an assignment $V \ni v \mapsto g_v \in G$ of a group element to each vertex in $V$ so that
$$U = U[\{ g_v \}] = \prod_{v \in V} \, A_v^{g_v}.$$
We call $\gauge[V]$ the group of gauge transformations on $V$.\\

If each $U \in \gauge[V]$ is supported on a set $S \subseteq \latticeedge$ then the gauge group $\gauge[V]$ acts on the gauge configurations $\gc[S]$ as follows. The gauge transformation $U = U[\{g_v\}] \in \gauge[V]$ acts on a gauge configuration $\al \in \gc[S]$, yielding a new gauge configuration $\al' := U(\al) \in \gc[S]$ given by $\al'_e = g_{\partial_0 e} \, \al_e \, \bar g_{\partial_1 e}$, where we set $g_v = 1$ whenever $v \not\in V$.\\

If $S \subset \latticeedge$ is finite then we let $\caH_{S} := \bigotimes_{e \in S} \caH_e$. The set of gauge configurations $\gc[S]$ then labels an orthonormal basis of $\caH_{S}$ given by $\ket{\al} := \bigotimes_{e \in S} \ket{\al_e}$. If the gauge transformations $\gauge[V]$ for some finite $V \subset \latticevert$ are supported in $S$ then these gauge transformations act on the Hilbert space $\caH_{S}$ as $U \ket{\al} = \ket{ U(\al) }$, i.e. Gauge transformations map basis states to basis states.\\

\subsection{Local gauge configurations and boundary conditions} \label{subsec:local gauge configurations and boundary conditions}

Recall that $\mathrm{dist}(\cdot, \cdot)$ is the graph distance on $\latticevert$. We fix an arbitrary site $\site = (v_0, f_0)$ and define (see figure \ref{fig:regions}):
\begin{align*}
\vregion(\site) &:= \{ v \in \latticevert \, : \,  \mathrm{dist}(v,v_0) \leq n\},\\
\fregion(\site) &:= \{f \in \latticeface \, : \, \exists v \in f \text{ such that } v \in \vregion \}, \\  
\eregion(\site) &:= \{e \in \latticeedge  \, : \, \exists f \in \fregion \text{ such that } e \in f \},\\
\partial \eregion(\site) &:= \{e \in \eregion \, : \,  \exists! f \in \fregion  \text{ with } e \in f \}, \\
\partial \vregion(\site) &:=  \vregion[n+1] \setminus \vregion.
\end{align*}
Note that these regions depend on the choice of an origin $\site$. Throughout this paper, we will want to consider different sites as the origin. In order to unburden the notation we will nevertheless drop $\site$ from the notation and simply write $\vregion, \fregion, \eregion$ and $\partial \eregion$ whenever it is clear from context which site is to serve as the origin.\\

\begin{figure}
    \centering
    \includegraphics[ width=0.4\textwidth]{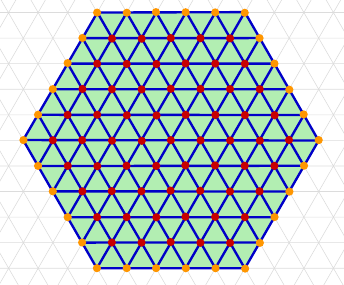}
    \caption{The sets $\vregion(\site)$ (red), $\fregion(\site)$ (green), $\eregion(\site)$ (blue), and $\partial \vregion(\site)$ (orange) are depicted for a site $\site = (v_0, f_0)$. The vertex $v_0$ sits in the center of the figure. The set $\partial \eregion$ consists of the blue edges on the boundary of the figure.}
    \label{fig:regions}
\end{figure}

For the remainder of this section, we fix a site $\site$ as our origin. We write $\gc := \gc[\eregion]$ for the gauge configurations on $\eregion$ and let
$$ \hilb_n := \hilb_{\eregion} = \bigotimes_{e \in \eregion} \hilb_e $$
be the Hilbert space associated to the region $\eregion$. The set of gauge configurations $\gc$ then labels an orthonormal basis of $\hilb_n$, given by $\ket{\alpha} = \bigotimes_{e \in \eregion} \ket{\alpha_e}$ for all $\alpha \in \gc$.\\

For any $\alpha \in \gc$, the corresponding basis state $\ket{\alpha} \in \hilb_n$ has a graphical representation, see Figure \ref{fig:string_net_example} for a schematic example.

\begin{figure}
    \centering
    \includegraphics[width = 0.6 \textwidth]{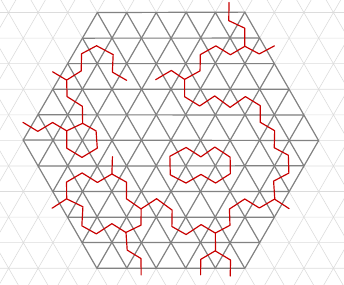}
    \caption{An example of a gauge configuration on $\eregion[5]$. The dark grey edges belong to $\eregion[5]$, they are crossed by red oriented strings that carry group labels. For edges that are not crossed by any red string, the gauge degree of freedom takes the value $1 \in G$. The orientations and group labels of the red strings are not shown. This picture corresponds to a definite basis state $\ket{\alpha} \in \hilb_n$ for some $\alpha \in \gc$.}
    \label{fig:string_net_example}
\end{figure}

\begin{definition} \label{def:flux thourgh ribbon}
    For a gauge configuration $\alpha$, the \emph{flux of $\alpha$ through a ribbon $\rho$} is defined as
    $$\phi_\rho(\alpha) := \prod_{e \in \rho^{dir}} \alpha_e$$
    where the product is ordered by the order of $\rho^{dir}$, the direct path of the ribbon $\rho$ as defined in section \ref{sec:direct paths}. 
\end{definition}

We will be interested in gauge configurations that satisfy certain constraints. Recall that to any site $s$ we can associate the elementary closed direct ribbon $\rho_{\triangle}(s)$ that starts and ends at $s$ and circles $f(s)$ in a counterclockwise direction. Let $\alpha$ be a gauge configuration on a region that contains all edges of $f(s)$ and define
$$\phi_s(\alpha) := \phi_{\rho_{\triangle}(s)}(\alpha)$$
to be the flux of $\alpha$ at $s$. By construction, we have $B_s^{g} \ket{\alpha} = \delta_{g, \phi_s(\alpha)} \ket{\alpha}$. For example, the flux at $s$ for the gauge configuration $\alpha$ depicted in Figure \ref{fig:graphical_representation} is $\phi_s(\alpha) = g_1 \bar g_2 \bar g_3$.\\

Let $\bc := \gc[\partial \eregion]$ be the set of gauge configurations on $\partial \eregion$. We call its elements $b : \partial \eregion \rightarrow G$ \emph{boundary conditions}. For any gauge configuration $\alpha \in \gc$ we denote by $b(\alpha) = \alpha|_{\partial \eregion}$ the \emph{boundary condition of $\alpha$} given by restriction of $\alpha$ to the boundary $\partial \eregion$. We write $b = \emptyset$ for the trivial boundary condition $\emptyset_e = 1 \in G$ for all $e \in \partial \eregion$.\\

Having fixed a site $\site = (v_0, f_0)$ we can regard $v_0$ as the origin of the plane and define unit vectors in $\R^2$ as follows. We let $\hat y$ be the unit vector with base at $v_0$ pointing towards the center of the face $f_0$, and we let $\hat x$ be the unit vector with base at $v_0$, perpendicular to $\hat y$ and such that $\hat x \times \hat y = 1$, \ie $(\hat x, \hat y)$ is a positive basis for $\R^2$. Let us now set $\hat l_1 = \hat x$ and $\hat l_2 = \cos(\pi/3) \hat x + \sin(\pi/3) \hat y$. Then each vertex $v \in \latticevert$ can be identified with its coordinate $(n_1, n_2)$ relative to $v_0$, \ie $v = v_0 + n_1 \hat l_1 + n_2 \hat l_2$. Using these coordinates, let $v_i = (i, 0)$ for $i \in \Z$ and consider the direct path $\nu_n^{dir} = ( (v_0, v_1), (v_1, v_2), \cdots, (v_{n-1}, v_n) )$.

We define the \emph{fiducial ribbon} $\fidu$ to be the unique positive ribbon such that $\partial_0 \fidu = \site$, such that $\nu_n^{dir}$ is the direct path of $\fidu$, and such that the final triangle of $\fidu$ is direct. We let $s_n = \partial_1 \fidu$ denote the final site of $\fidu$. See Figure \ref{fig:fiducial_and_boundary_ribbons}.

We define the \emph{boundary ribbon} $\bdy$ to be the unique closed positive ribbon starting and ending at $s_n$ such that its direct path consists of the edges in $\partial \eregion$, oriented counterclockwise around $\eregion$. See Figure \ref{fig:fiducial_and_boundary_ribbons}.

\begin{definition} \label{def:boundary condition projector}
    For any boundary condition $b \in \bc$ we define a projector $P_b \in \cstar[\eregion]$ given by
    $$P_b = \prod_{\{\tau_e \in \bdy| \tau_e \text{ direct}\}} T^{b_e}_{\tau_e}.$$
\end{definition}

\begin{definition}
    We call $\phi_{\bdy}(\alpha)$ the \emph{boundary flux} of the gauge configuration $\alpha \in \gc$.
\end{definition}

\begin{definition}
    For any boundary condition $b \in \bc$ we write $\phi_{\bdy}(b)$ for the associated \emph{boundary flux} as measured through the boundary ribbon $\bdy$.
\end{definition}

\begin{figure}
    \centering
    \includegraphics[width=0.8\textwidth]{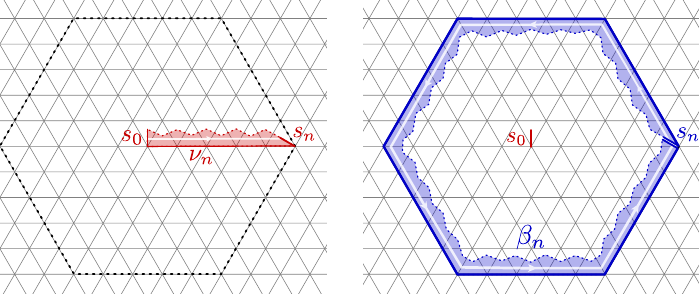}
    \caption{The fiducial ribbon $\fidu$ in red (left) and the boundary ribbon $\bdy$ in blue (right) for a given site $\site$ and $n = 5$. For the fiducial ribbon $\fidu$, we have $\start[ \fidu] = \site$ and $\en[ \fidu] = s_n$. For the boundary ribbon $\bdy$, we have $\start[\bdy] = \en[\bdy] = s_n$, and a counterclockwise orientation around $\partial \eregion$. Any site $s$ is related to the site $\site$ by lattice rotation and translation. We define the fiducial ribbons and boundary ribbons for arbitrary $s$ by the corresponding rotations and translations of the fiducial and boundary ribbons of $\site$.}
    \label{fig:fiducial_and_boundary_ribbons}
\end{figure}

\subsection{Irreducible representations of \texorpdfstring{$\caD(G)$}{G(G)}, Wigner projectors, and local constraints} \label{subsection:irreps and Wigner projections}

As mentioned above, we have at each site $s$ a realisation of the quantum double algebra $\caD(G)$ generated by the gauge transformations and flux projectors at $s$.

Let us introduce some terminology and conventions that will allow us to analyse representations of the quantum double algebra. Denote by $(G)_{cj}$ the set of conjugacy classes of $G$. For each conjugacy class $C \in (G)_{cj}$ let $C = \{c_i\}_{i = 1}^{|C|}$ be a labeling of its elements. Any $g \in C$ has $g = c_i$ for a definite label $i$, and we define the label function $i := i(g)$. Pick an arbitrary \emph{representative element} $r_C \in C$. All elements of $C$ are conjugate to the chosen representative $r_C$ so we can fix group elements $q_i$ such that for all $c_i \in C$ we have $c_i = q_i r_C \dash{q}_i$. We let $Q_C:= \{q_i\}_{i = 1}^{|C|}$ be the \emph{iterator set} of $C$. Let $N_C:=\{n\in G| n r_C = r_C n \}$ be the commutant of $r_C$ in $G$. Note that the group structure of $N_C$ does not depend on the choice of $r_C$, it is a realization of the centralizer of $C$. Denote by $(N_C)_{irr}$ the collection of irreducible representations of $N_C$.

As mentioned in the setup, the irreducible representations of the quantum double algebra of $G$ are in one-to-one correspondence with pairs $RC$ where $C \in (G)_{cj}$ is a conjugacy class and $R \in (N_C)_{irr}$ is an irreducible representation of the group $N_C$.

For each $R \in (N_C)_{irr}$ we fix a concrete unitary matrix representation $N_C \ni m \mapsto R(m) \in \caM_{\dimR}(\C)$ with components $R^{j j'}(m)$.

In what follows we will often consider a label $i \in \{1, \cdots, \abs{C}\}$ for $C$ together with a label $j \in \{ 1, \cdots, \dimR \}$. We define $I_{RC} := \{1, \cdots, \abs{C}\} \times \{ 1, \cdots, \dimR \}$ so that $(i, j) \in I_{RC}$.

\begin{definition} \label{def:Wigner projectors}
    Let us define the Wigner projector to $RC$ at site $s$ by
    $$\qd := \frac{\dimR}{|N_C|} \sum_{m\in N_C} \chi_R(m)^* \sum_{q \in Q_C} A_{s} ^{q  m \dash{q}} B_{s}^{q r_C \dash{q}}.$$
    $\qd$ decomposes as a sum of commuting projectors $\{ \qd[RC;u] \}_{u \in I_{RC}}$ (Lemma \ref{lem:DRC decomposes into DRCu}). For $u = (i, j)$ these are given by
    $$\qd[RC;u] := \frac{\dimR}{|N_C|} \sum_{m \in N_C} R^{jj}(m)^* A_{s} ^{q_i m\dash{q}_i} B_{s}^{c_i}.$$ 
\end{definition}

Fix a site $\site = (v_0, f_0)$ and introduce the following notations
$$\dlatticeface := \latticeface \setminus \{f_0\} \qquad \dlatticevert := \latticevert \setminus \{v_0\}.$$
The site $s_0$ will be fixed throughout this section, and will therefore often not be made explicit in the notation. \\

Let us define the following sets of states.

\begin{definition} \label{def:state spaces}
    Let $\overline{\S}_{s_0}$ be the set of states $\omega$ on $\cstar$ that satisfy
    \begin{equation}
        \label{eq:cons1}
        \omega(A_v) = \omega(B_f) = 1 \quad  \text{ for all } \,\,\,  v \in \dlatticevert, \quad f \in \dlatticeface.
    \end{equation}
    Similarly, we denote by $\S_{s_0}^{RC}$ the set of states $\omega$ on $\cstar$ that in addition to \eqref{eq:cons1} also satisfy
    \begin{equation}
        \label{eq:cons2}
        \omega(D^{RC}_{s_0}) = 1,
    \end{equation}
    and by $\S_{s_0}^{RC;u}$ the set of states that in addition to \eqref{eq:cons1} also satisfy
    \begin{equation}
        \label{eq:cons3}
        \omega(\qdsu) = 1.
    \end{equation}
\end{definition}

In this section we prove that the set $\S_{\site}^{RC;u}$ contains a single pure state. Considering the case where $C = C_1 := \{1\}$ is the trivial conjugacy class and $R$ is the trivial representation of $N_{\{1\}} = G$, we see that $\S_{\site}^{RC_1}$ is precisely the set of frustration free ground states, so we get in particular a new proof of Proposition \ref{prop:ffgsunique}.

\subsection{Local constraints}

We will characterise the state spaces $\overline{\S}_{\site}, \S_{\site}^{RC}$ and $\S_{\site}^{RC;u}$ by investigating the restrictions of states belonging to these spaces to finite volumes $\eregion$. These restrictions correspond to density matrices acting on $\caH_n$ that are supported on subspaces of $\caH_n$ defined by local versions of the constraints \eqref{eq:cons1}, \eqref{eq:cons2} and \eqref{eq:cons3}. Here we introduce these subspaces.

Let us write
$$\dfregion := \fregion \setminus \{f_0\} \qquad \dvregion := \vregion \setminus \{v_0\}.$$

\begin{definition} \label{def:local constraints}
    Let $\overline{\V} \subset \hilb_n$ be the subspace consisting of vectors $\ket{\psi} \in \hilb_n$ that satisfy
    \begin{equation}
        \label{eq:lcons1}
        A_v \ket{\psi} = B_f \ket{\psi} = \ket{\psi} \quad \text{for all} \,\, v \in \dvregion, f \in \dfregion.
    \end{equation}
    Let $\VRC \subset \overline{\V}$ be the subspace consisting  of vectors $\ket{\psi} \in \overline{\V}$ that in addition to \eqref{eq:lcons1} also satisfy
    \begin{equation}
        \label{eq:lcons2}
        D_{\site}^{RC} \ket{\psi} = \ket{\psi},
    \end{equation}
    and let $\VRCu \subset \VRC$ be the subspace consisting of vectors $\ket{\psi} \in \VRC$ that in addition to \eqref{eq:lcons1} also satisfy
    \begin{equation}
        \label{eq:lcons3}
        D_{\site}^{RC;u} \ket{\psi} = \ket{\psi}.
    \end{equation}
\end{definition}

Note that since $\mathds{1} = \sum_{RC} D_{\site}^{RC}$ (Lemma \ref{lem:DRCprops}) and  $D_{\site}^{RC} = \sum_{u} D_{\site}^{RC;u}$ (Lemma \ref{lem:DRC decomposes into DRCu}) we have orthogonal decompositions
$$\overline{\V} = \bigoplus_{RC} \VRC \quad \text{and} \quad \VRC = \bigoplus_u \VRCu$$

\subsection{Imposing flux constraints and boundary conditions}

For each $RC$ and $u = (i, j) \in I_{RC}$, and each site $s$, we define
\begin{equation*}
    A_s^{RC;u} := \frac{\dimR}{|N_C|} \, \sum_{m \in N_C} \, R^{jj}(m)^* \, A_s^{q_i m \bar q_i},
\end{equation*}
so $D_s^{RC;u} = A_s^{RC;u} B_s^{c_i}$. From Lemma \ref{lem:a and B commute with A_v and B_f} we have that the $A_{s}^{RC;u}$ are projectors that commute with $B_{s}^{c_i}$. In other words, the projector $D_{s_0}^{RC;u}$ really imposes two independent constraints, namely a flux constraint $B_{\site}^{c_i}$ and a gauge constraint $A_{\site}^{RC;u}$.\\

Throughout this section we will find the following Lemma useful. Recall the projectors $P_b$ from Definition \ref{def:boundary condition projector} that project on the boundary condition $b \in \bc$.
\begin{lemma} \label{lem:commuting local constraints}
    For any $C \in (G)_{cj}$, any $R \in (N_C)_{irr}$, any $u = (i, j) \in I_{RC}$, and any boundary condition $b \in \bc$, the set
    $$\{ B_f  \}_{f \in \dfregion} \cup \{ A_v \}_{v \in \dvregion} \cup \{ B_{\site}^{c_i}, \, A_{\site}^{RC;u}, \, P_b \}$$
    is a set of commuting projectors.
\end{lemma}

\begin{proof}
    The set
    $$\{ B_f  \}_{f \in \dfregion} \cup \{ A_v \}_{v \in \dvregion} \cup \{ B_{\site}^{c_i}, \, A_{\site}^{RC;u}\}$$
    is a set of commuting projectors by Eq.  \eqref{eq:ABcommproj} and Lemma \ref{lem:a and B commute with A_v and B_f}. The projectors $\{ A_v \}_{v \in \dvregion} \cup \{ B_{\site}^{c_i}, \, A_{\site}^{RC;u}\}$ are all supported on $\eregion \setminus \partial \eregion$ while $P_b$ is supported on $\partial \eregion$, so these projectors commute with $P_b$. The projectors $\{ B_f  \}_{f \in \dfregion}$ and $P_b$ are all diagonal in the basis of gauge configurations, so they also commute.
\end{proof}

We first investigate the space of vectors in $\caH_n$ that satisfy flux constraints.

\begin{definition} \label{def:WCi}
    Let $\WCi \subset \caH_n$ be the subspace consisting of vectors $| \psi \rangle \in \caH_n$ such that
    \begin{equation*}
        \ket{\psi} = B_f \ket{\psi} = B_{s_0}^{c_i} \ket{\psi}
    \end{equation*}
    for all $f \in \dfregion$.
\end{definition}

The space $\WCi$ is spanned by vectors $| \al \rangle$ for certain $\al \in \gc$ that satisfy these constraints.

\begin{definition} \label{def:string nets Ci}
    For any conjugacy class $C$ and any $i = 1, \cdots, \abs{C}$ we define
    $$\packi :=\{\alpha \in \gc \, : \, \ket{\al} = B_{\site}^{c_i} \ket{\al} = B_f \ket{\al} \quad \text{for all} \,\,\, f \in \dfregion \}. $$
\end{definition}

See Figure \ref{fig:defect string net examlpe} for an example of a string net $\al \in \packi$.

\begin{figure}
    \centering
    \includegraphics[width=0.6\textwidth]{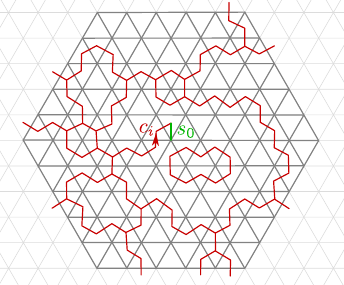}
    \caption{The graphical representation of a string net $\alpha \in \packC[C;i]$. The orientations and labels of the red strings are not shown, except for the piece that ensures that the constraint $B_{\site}^{c_i} \ket{\al} = \ket{\al}$ is satisfied.}
    \label{fig:defect string net examlpe}
\end{figure}

\begin{lemma} \label{lem:WCi spanned by string nets}
    We have
    $$\WCi = \Span{\ket{\al} \, : \, \al \in \packi }.$$
\end{lemma}

\begin{proof}
    That $\ket{\al} \in \WCi$ if $\al \in \packi$ is immediate from the definitions. Conversely, since $\{ B_f \}_{f \in \dfregion} \cup \{ B_{s_0}^{c_i} \}$ is a set of commuting projectors we have
    $$\WCi = \left( \prod_{f \in \dfregion}  B_f \right) \, B_{s_0}^{c_i} \, \caH_n.$$
    Now note that the vectors $\ket{\al}$ for $\al \in \gc$ from an orthonormal basis for $\caH_n$ and that
    $$\left( \prod_{f \in \dfregion}  B_f \right) \, B_{s_0}^{c_i} \ket{\al} = \begin{cases} \ket{\al} \quad &\text{if } \,\,\, \al \in \packi \\
    0 &\text{otherwise}.\end{cases}$$
    The claim follows.
\end{proof}

We can further refine the spaces $\WCi$ by specifying boundary conditions.

\begin{definition} \label{def:compatible boundary conditions}
    We say a boundary condition $b \in \bc$ is compatible with the conjugacy class $C$ if $\phi_{\bdy}(b) \in C$. We denote by $\bc[C]$ the set of boundary conditions compatible with $C$. For $b \in \bc[C]$ we have $\phi_{\bdy}(b) = c_i \in C$ for a definite index $i \in \{1, \cdots, \abs{C}\}$. We write $i = i(b)$.
\end{definition}

\begin{definition} \label{def:WCib}
    For any conjugacy class $C \in (G)_{cj}$, any $i = 1, \cdots, \abs{C}$, and any boundary condition $b \in \bc$ we let $\WCib \subset \caH_n$ be the space of vectors $\ket{\psi} \in \caH_n$ that satisfy
    $$\ket{\psi} = B_f \ket{\psi} = B_{\site}^{c_i} \ket{\psi} = P_b \ket{\psi}$$
    for all $f \in \dfregion$. Here $P_b$ is the projector on the boundary condition $b$ from Definition \ref{def:boundary condition projector}.
\end{definition}

\begin{definition} \label{def:string nets Cib}
    For any boundary condition $b \in \bc$ we let
    $$\packib := \{  \al \in \packi \, : \, b(\al) = b \}.$$
\end{definition}

\begin{lemma} \label{lem:WCib spanned by string nets}
    We have
    $$\WCib := \Span{ \ket{\al} \, : \, \al \in \packib }.$$
\end{lemma}

\begin{proof}
    This follows from $\WCi = \Span{ \ket{\al} \, : \, \al \in \packi  }$ (Lemma \ref{lem:WCi spanned by string nets}), the fact that $\al \in \packib$ if and only if $\al \in \packi$ and $P_b \ket{\al} = \ket{\al}$, and the fact (Lemma \ref{lem:commuting local constraints}) that $\{B_f\}_{f \in \dfregion} \cup \{ B_{\site}^{c_i}, \, P_b \}$ is a set of commuting projectors.
\end{proof}

\begin{lemma} \label{lem:decomposition of WCi according to boundary conditions}
    We have a disjoint union
    $$\packi = \bigsqcup_{b \in \bc[C]} \packib$$
    and an orthogonal decomposition
    $$\WCi = \bigoplus_{b \in \bc[C]}  \WCib.$$

    In particular, $\packib$ is empty if $b$ is not compatible with $C$.
\end{lemma}

\begin{proof}
    To show the first claim, it is sufficient to show that $\packib$ is empty if $b$ is not compatible with $C$. This follows from Lemma \ref{lem:bcinC}.
    The second claim then follows immediately from Lemma \ref{lem:WCib spanned by string nets} and Lemma \ref{lem:WCi spanned by string nets}.
\end{proof}

\subsection{Fiducial flux}

The fiducial flux, which is measured by the projectors $T_{\fidu}^g$, remains unconstrained by the projectors defining the spaces $\WCib$, we can therefore further decompose the spaces $\WCib$ according to the fiducial flux.

\begin{lemma} \label{lem:commuting constraints with fiducial flux}
    For any $C \in (G)_{cj}$, any $i = 1, \cdots, \abs{C}$, any $g \in G$, and any boundary condition $b \in \bc$, the set
    $$\{ B_f  \}_{f \in \dfregion} \cup \{ A_v \}_{v \in \dvregion} \cup \{ B_{\site}^{c_i}, \, P_b, \, T_{\fidu}^g \}$$
    is a set of commuting projectors.
\end{lemma}

\begin{proof}
    That 
    $$\{ B_f  \}_{f \in \dfregion} \cup \{ A_v \}_{v \in \dvregion} \cup \{ B_{\site}^{c_i}, \, P_b\}$$
    is a set of commuting projectors follows from Lemma \ref{lem:commuting local constraints}. To prove the lemma, we must show that $T_{\fidu}^g$ commutes with all the projectors in this set. From Lemma \ref{lem:flux change} we get that $T_{\fidu}^g$ commutes with all $A_v$ for $v \in \dvregion$, all $B_f$ for $f \in \dfregion$ and with $B_{\site}^{c_i}$. To see that $T_{\fidu}^g$ commutes with $P_b$, simply note that $P_b$ is supported on $\partial \eregion$ while $T_{\fidu}^{g}$ is supported on the ribbon $\fidu$, whose support does not contain any edges of $\partial \eregion$.
\end{proof}

Note that by Lemma \ref{lem:inNC} we have for any $\al \in \packib$ that $\bar q_i \phi_{\fidu}(\al) q_{i(b)} \in N_C$. This motivates the following Definition.
\begin{definition} \label{def:WCib(m)}
    For any conjugacy class $C \in (G)_{cj}$, any $i = 1, \cdots, \abs{C}$, any boundary condition $b \in \bc$, and any $m \in N_C$ we let $\WCib(m) \subset \caH_n$ be the space of vectors $\ket{\psi} \in \caH_n$ that satisfy
    $$\ket{\psi} = B_f \ket{\psi} = B_{\site}^{c_i} \ket{\psi} = P_b \ket{\psi} = T_{\fidu}^{q_i m \bar q_{i(b)}} \ket{\psi}$$
    for all $f \in \dfregion$.
\end{definition}

These spaces are again spanned by certain string-net states that have a definite fiducial flux.
\begin{definition} \label{def:string nets Cib(m)}
    For any $m \in N_C$ we define
    $$ \packib(m) := \{  \al \in \packib \, : \, \bar q_i \phi_{\fidu}(\alpha) q_{i(b)} = m  \}.$$
\end{definition}

\begin{lemma} \label{lem:WCib(m) spanned by string nets}
    We have
    $$ \WCib(m) = \Span{ \ket{\al} \, : \, \al \in \packib(m) }. $$
\end{lemma}

\begin{proof}
    This follows from Lemma \ref{lem:commuting constraints with fiducial flux} and the fact that
    $$ T_{\fidu}^{g} \ket{\al} = \delta_{\phi_{\fidu}(\al), g} \ket{\al}, $$
    see Lemma \ref{lem:flux projector}.
\end{proof}

\subsection{Imposing gauge invariance on $\dvregion$}

\begin{definition} \label{def:VCib(m)}
    For any conjugacy class $C$, any $i = 1, \cdots, \abs{C}$, any boundary condition $b \in \bc[C]$, and any $m \in N_C$ we define $\VCib(m) \subset \caH_n$ to be the space of vectors $\ket{\psi} \in \caH_n$ that satisfy
    $$ \ket{\psi} = B_f \ket{\psi} = B_{\site}^{c_i} \ket{\psi} = P_b \ket{\psi} = T_{\fidu}^{q_i m \bar q_{i(b)}} \ket{\psi} = A_v \ket{\psi} $$
    for all $f \in \dfregion$ and all $v \in \dvregion$.
\end{definition}

We will show that the spaces $\VCib(m)$ are one-dimensional. To this end we introduce the following group of gauge transformations.
\begin{definition} \label{def:gauge group on dvregion}
    We let $\dgauge := \gauge[\dvregion]$ be the group of gauge transformations on $\dvregion$. All of its elements are unitaries of the form
    $$ U[\{ g_v\}] = \prod_{v \in \dvregion}  \, A_v^{g_v}$$
    for some $\{g_v\} \in G^{|\dvregion|}$.
\end{definition}

Just like for the average over the gauge group $\dgauge$ we have

$$ \frac{1}{\abs{\dgauge}} \, \sum_{U \in \dgauge} U = \frac{1}{\abs{G}^{\abs{\dvregion}}} \sum_{\{ g_v \} \in G^{|\dvregion|} } \, \prod_{v \in \dvregion} A_v^{g_v} = \prod_{v \in \dvregion} \,\left( \frac{1}{\abs{G}} \, \sum_{g_v \in G} A_v^{g_v} \right) = \prod_{v \in \dvregion} A_v. $$

\begin{lemma} \label{lem:VCib(m) is dgauge projection of WCib(m)}
    We have
    $$ \VCib(m) = \left( \frac{1}{\abs{\dgauge}} \sum_{U \in \dgauge} U \right) \, \WCib(m).$$
\end{lemma}

\begin{proof}
    Since $\frac{1}{\abs{\dgauge}} \sum_{U \in \dgauge} U = \prod_{v \in \dvregion} A_v$ it is sufficient to show that
    $$ \VCib(m) = \bigg( \prod_{v \in \dvregion} A_v \bigg) \, \WCib(m). $$
    This follows immediately from the Definition \ref{def:WCib(m)} and Lemma \ref{lem:commuting constraints with fiducial flux}.
\end{proof}

We now define unit vectors which, as we shall see, span the spaces $\VCib(m)$.
\begin{definition} \label{def:eta^Cib(m)}
    For all $n > 1$, all conjugacy classes $C \in (G)_{cj}$, all labels $i = 1, \cdots, \abs{C}$, all boundary conditions $b \in \bc[C]$, and all $m \in N_C$ we define the unit vector
    $$ \ket{\eta_n^{C;ib}(m)} := \frac{1}{\abs{\packib(m)}^{1/2}} \, \sum_{\al \in \packib(m)} \, \ket{\al}. $$
\end{definition}

We can now use the fact that $\WCib(m)$ is spanned by vectors $\ket{\al}$ with $\al \in \packib(m)$ and that the gauge group $\dgauge$ acts freely and transitively on $\packib(m)$ to show
\begin{lemma} \label{lem:orthonormality of the etaCib(m)}
    The space $\VCib(m) \subset \caH_n$ is one-dimensional and spanned by the vector $\ket{ \eta_n^{C;ib}(m) }$. In particular, the vectors $\biggl\{\ket{ \eta_n^{C;ib}(m) }\bigg\}_{C, i, b, m}$ form an orthonormal family.
\end{lemma}

\begin{proof}
    By Lemma \ref{lem:WCib(m) spanned by string nets}, we have $\WCib(m) = \Span{ \ket{\al} \, : \, \al \in \packib(m) }$. By Lemma \ref{lem:VCib(m) is dgauge projection of WCib(m)} it is sufficient to show that
    $$  \sum_{U \in \dgauge} U \ket{\al} \propto \ket{\eta_n^{C;ib}(m)}$$
    for all $\al \in \packib(m)$. This follows immediately from Lemma \ref{lem:transitive bulk action} which states that $\dgauge$ acts freely and transitively on $\packib(m)$:
    $$ \sum_{U \in \dgauge} U \ket{\al} = \sum_{\al' \in \packib(m)} \ket{\al'} = \abs{\packib(m)}^{1/2} \ket{\eta_n^{C;ib(m)}}. $$
\end{proof}

\subsection{Action of $N_C$ on fiducial flux and irreducible subspaces}

The gauge transformations $A_{v_0}^{q_i m \bar q_i}$ realise a left group action of $N_C$ on the vectors $\ket{\eta_n^{C;ib}(m)}$.
\begin{lemma} \label{lem:N_C action on VCib}
    For any $m_1, m_2 \in N_C$ we have
    $$ A_{\site}^{q_i m_1 \bar q_i} \ket{ \eta_{n}^{C;ib}(m_2) } = \ket{\eta_n^{C;ib}(m_1 m_2)}. $$
\end{lemma}

\begin{proof}
    If $\al \in \packib(m_2)$, then by definition $\phi_{\fidu}(\al) = q_i m_2 \bar q_{i(b)}$. The gauge transformation $A_{\site}^{q_i m_1 \bar q_i}$ acts on such a string-net to yield $\ket{\al'} = A_{\site}^{q_i m_1 \bar q_i} \ket{\al}$ for a new string net $\al'$. Since $A_{\site}^{q_i m_1 \bar q_i}$ commutes with the projectors $\{ B_{\site}^{c_i}, P_b \} \cup \{ B_f \}_{f \in \dfregion}$ (Lemma \ref{lem:commuting local constraints}) and we have $\phi_{\fidu}(\al') = q_i m_1 m_2 \bar q_{i(b)}$ by Lemma \ref{lem:flux change}, we find that $\al' \in \packib(m_1 m_2)$. Since $A_{\site}^{q_i m_2 \bar q_i}$ acts invertibly on string nets, we see that it yields a bijection from $\packib(m_1)$ to $\packib(m_1 m_2)$. In particular, these sets have the same cardinality and
    \begin{align*}
        A_{\site}^{q_i m_1 \bar q_i} \ket{\eta_n^{C;ib}(m_2)} &= \frac{1}{\abs{\packib(m_2)}^{1/2}}  \sum_{\al \in \packib(m_2)} \, {q_i m_1 \bar q_i} \ket{\al} \\
                &= \frac{1}{\abs{\packib(m_1 m_2)}^{1/2}}  \sum_{\al \in \packib(m_1 m_2)} \, \ket{\al} \\  
                &= \ket{\eta_n^{C;ib}(m_1 m_2)}.
    \end{align*}
\end{proof}

The space spanned by the vectors $ \Bigl\{ \ket{\eta_n^{C;ib}(m)}  \Bigl\}_{m \in N_C}$ therefore carries the regular representation of $N_C$, with a left group action provided by the gauge transformations $A_{\site}^{q_i m \bar q_i}$ for $m \in N_C$, which are supported near the site $\site$. It turns out that this space also carries a natural right action of $N_C$ provided by unitaries supported near the boundary of $\eregion$, see Lemma \ref{lem:right action on VCib}.

We can characterise this space as follows. Let us define
\begin{definition} \label{def:VCib}
    $\VCib \subset \caH_n$ is the subspace consisting of vectors $\ket{\psi} \in \caH_n$ such that
    $$\ket{\psi} = A_v \ket{\psi} = B_f \ket{\psi} = B_{\site}^{c_i} \ket{\psi} = P_b \ket{\psi}$$
    for all $v \in \dvregion$ and all $f \in \dfregion$.
\end{definition}
Then
\begin{lemma} \label{lem:VCib as direct sum}
    $$\VCib = \bigoplus_{m \in N_C} \, \VCib(m) = \Span{  \ket{\eta_n^{C;ib}(m)} \, : \, m \in N_C }.$$
\end{lemma}

\begin{proof}
    The spaces $\VCib(m)$ are defined by the same constraints as the space $\VCib$, plus the constraint $T_{\fidu}^{q_i m \bar q_{i(b)}}$ on the fiducial flux (Definitions \ref{def:VCib(m)} and \ref{def:VCib}).
    
    Since $\VCib \subset \WCib$ (cf. Definition \ref{def:WCib}) is spanned by vectors $\ket{\al}$ for $\al \in \packib$ (Lemma \ref{lem:WCib spanned by string nets}) which satisfy $q_i \phi_{\fidu}(\al) \bar q_{i(b)} \in N_C$ (Lemma \ref{lem:inNC}), we have (using Lemma \ref{lem:flux projector}) $\sum_{m \in N_C} T_{\fidu}^{q_i m \bar q_{i(b)}} \, \ket{\phi} = \ket{\phi}$ for all $\ket{\phi} \in \VCib$.

    Using Lemma \ref{lem:commuting constraints with fiducial flux}, it follows that
    $$ \VCib = \sum_{m \in N_C} \, T_{\fidu}^{q_i m \bar q_{{i(b)}}} \, \VCib = \bigoplus_{m \in N_C} \, \VCib(m)$$
    as required.

    The second equality in the claim now follows immediately from Lemma \ref{lem:orthonormality of the etaCib(m)}.
\end{proof}

Since $\VCib$ carries the regular representation of $N_C$, we can construct an orthonormal basis of $\VCib$ that respects the irreducible subspaces of $\VCib$ for both the left and the right action of $N_C$.
\begin{definition} \label{def:representation basis}
    For any conjugacy class $C \in (G)_{cj}$, any irreducible representation $R \in (N_C)_{irr}$, any label $u = (i, j) \in I_{RC}$, any boundary condition $b \in \bc[C]$, and any label $j' = 1, \cdots, \dimR$, and writing $v = (b, j')$, we define a vector
    $$ \qdpure := \left(\frac{\dimR}{\abs{N_C}} \right)^{1/2} \, \sum_{m \in N_C} \, R^{j j'}(m)^*  \ket{\eta_n^{C;ib}(m)}.$$
\end{definition}

We will write $I'_{RC} := \bc[C] \times \{ 1, \cdots, \dimR \}$ for the possible values of the label $v = (b, j')$.

\begin{lemma} \label{lem:etas form ONB}
    The vectors $\{ \qdpure \}$ form an orthonormal family, i.e.
    $$ \inner{\eta_n^{R_1C_1;u_1 v_1}}{ \eta_n^{R_2 C_2 ; u_2 v_2} } = \delta_{R_1 C_1, R_2 C_2} \delta_{u_1, u_2} \delta_{v_1, v_2}. $$
\end{lemma}

\begin{proof}
    Let $u_1 = (i_1, j_1)$, $u_2 = (i_2, j_2)$, $v_1 = (b_1, j'_1)$ and $v_2 = (b_2, j'_2)$. Then
    $$ \inner{\eta_n^{R_1C_1;u_1 v_1}}{ \eta_n^{R_2 C_2 ; u_2 v_2} } = \frac{(\dim(R_1) \dim(R_2))^{1/2}}{\abs{N_{C_1}}} \, \delta_{C_1, C_2} \, \delta_{i_1, i_2} \delta_{b_1, b_2} \, \sum_{m \in N_{C_1}} \, R_1^{j_1 j'_1}(m) \, R_2^{j_2 j'_2}(m)^* $$
    where we used that the vectors $\ket{\eta_{n}^{C;ib}(m)}$ are orthonormal (Lemma \ref{lem:orthonormality of the etaCib(m)}). Now using the Schur orthogonality relation \eqref{eq:Schur} gives us the required result.
\end{proof}

The $\qdpure$ were in fact obtained by a unitary rotation of the states $\ket{\eta_n^{C;ib}(m)}$, and this rotation can be reversed.
\begin{lemma} \label{lem:inverse rotation}
    We have for all $C \in (G)_{cj}$, all $i = 1, \cdots, \abs{C}$, all $b \in \bc[C]$, and all $m \in N_C$ that
    $$ \ket{\eta_n^{C;ib}(m)} = \sum_{R \in (N_C)_{irr}} \, \left( \frac{\dimR}{\abs{N_C}} \right)^{1/2} \, \sum_{j, j'} \, R^{j j'}(m) \, \ket{ \eta_n^{RC;(i, j)(b, j')} }. $$
\end{lemma}

\begin{proof}
    We have
    \begin{align*}
        \sum_{R \in (N_C)_{irr}} \, & \left( \frac{\dimR}{\abs{N_C}} \right)^{1/2} \, \sum_{j, j'} \, R^{j j'}(m) \, \ket{ \eta_n^{RC;(i, j)(b, j')} } \\
        &= \sum_{R \in (N_C)_{irr}} \, \frac{\dimR}{\abs{N_C}}  \, \sum_{j, j'} \, \sum_{m' \in N_C} \, R^{j j'}(m) R^{j j'}(m')^* \, \ket{\eta_n^{C;ib}(m')} \\
        &=  \, \sum_{R \in (N_C)_{irr}} \, \frac{1}{\abs{N_C}} \sum_{m' \in N_C} \, \chi_R(m \bar m') \chi_R(1)^* \, \ket{\eta_n^{C;ib}(m')} \\
        &= \ket{\eta_n^{C;ib}(m)}
    \end{align*}
    where we used $\chi_R(1) = \dimR$ and the Schur orthogonality relation \eqref{eq:Schur2} for irreducible characters.
\end{proof}

\subsection{Characterisation of the spaces $\overline{\V}$, $\VRC$, and $\VRCu$}

We can now describe the spaces $\overline{\V}, \VRC$ and $\VRCu$ from Definition \ref{def:local constraints} in terms of the vectors $\qdpure$.

\begin{proposition} \label{prop:local spaces spanned by etas}
    We have
    \begin{align*}
        \overline{\V} &= \Span{  \qdpure  \, : \, C \in (G)_{cj}, R \in (N_C)_{irr}, u \in I_{RC}, v \in I'_{RC} }, \\
        \VRC &= \Span{  \qdpure \, : \, u \in I_{RC}, v \in I'_{RC} }, \\
        \VRCu &= \Span{  \qdpure \,: \, v \in I'_{RC} }.
    \end{align*}
\end{proposition}

\begin{proof}
    We first note that it follows from Lemma \ref{lem:decomposition of WCi according to boundary conditions}, Definition \ref{def:WCib}, and Lemma \ref{lem:commuting local constraints} that
    $$ B_{\site}^{c_i} \bigg( \prod_{f \in \dfregion}  B_f \bigg)  \, P_b  = 0$$
    whenever $b \in \bc$ is not compatible with $C$. Using this, the fact that $\sum_{C, i} B_{\site}^{c_i} = \sum_{b \in \bc} P_b = \I$, and Lemma \ref{lem:commuting local constraints} we find
    \begin{align*}
        \overline{\V} &= \bigg( \prod_{v \in \dvregion} A_v \prod_{f \in \dfregion} B_f  \bigg)  \, \caH_n = \sum_{C \in (G)_{cj}} \sum_{i = 1}^{\abs{C}} \sum_{b \in \bc}  \bigg( \prod_{v \in \dvregion} A_v \bigg) \, B_{\site}^{c_i} \bigg( \prod_{f \in \dfregion} B_f  \bigg)  \, P_b \, \caH_n \\
        &= \sum_{C \in (G)_{cj}} \sum_{i = 1}^{\abs{C}} \sum_{b \in \bc[C]} \, \bigg( \prod_{v \in \dvregion} A_v \bigg) \, B_{\site}^{c_i} \bigg( \prod_{f \in \dfregion} B_f  \bigg)  \, P_b \, \caH_n \\
        &= \bigoplus_{C \in (G)_{cj}} \bigoplus_{i = 1}^{\abs{C}} \bigoplus_{b \in \bc[C]} \, \VCib,
    \end{align*}
    where we used the Definition \ref{def:VCib} of the spaces $\VCib$. From Lemma \ref{lem:VCib as direct sum} it then follows that
    $$ \overline{\V} = \Span{ \ket{\eta_n^{C;ib}(m)} \, : \, C \in (G)_{cj}, i = 1, \cdots, \abs{C}, b \in \bc[C], m \in N_C  }. $$
    Together with Definition \ref{def:representation basis} and Lemma \ref{lem:inverse rotation} this yields the first claim.
    
    To show the second claim we note that $\VRC = D_{\site}^{RC} \overline{\V}$, and for any $C_1, C_2 \in (G)_{cj}$, $R_1 \in (N_{C_1})_{irr}$, $R_2 \in (N_{C_2})_{irr}$, $u = (i, j) \in I_{R_2 C_2}$, and $v = (b, j') \in I'_{R_2 C_2}$ we have (Lemma \ref{lem:action of Wigner projecitons on string-net condensates})
    $$ D_{\site}^{R_1 C_1} \ket{ \eta_n^{R_2 C_2 ; u v} } = \delta_{R_1 C_1, R_2 C_2}  \, \ket{\eta_n^{R_2 C_2 ; u v}}.$$
    The second claim of the Proposition then follows from the fact that $\overline{\V}$ is spanned by the vectors $\ket{\eta_n^{RC;uv}}$ for arbitrary $RC$ and $u \in I_{RC}$, $v \in I'_{RC}$.

    To show the final claim we note that we have $\VRCu = D_{\site}^{RC;u} \VRC$, and for any $u_1, u_2 \in I_{RC}$ and any $v \in I'_{RC}$ we have (Lemma \ref{lem:action of Wigner sub-projector on string-net condensates})
    $$ D_{\site}^{RC;u_1} \, \ket{ \eta_n^{RC;u_2 v} } = \delta_{u_1, u_2} \, \ket{\eta_n^{RC;u_2 v}}.$$
    The final claim then follows from the fact that $\VRC$ is spanned by the vectors $\ket{\eta_n^{RC;uv}}$ for $u \in I_{RC}$ and $v \in I'_{RC}$.
\end{proof}

\subsection{The bulk is independent of boundary conditions}

Let us define the following operators
\begin{definition} \label{def:label changers}
    For any site $s$, any $n$, any $u_1 = (i_1, j_1), u_2 = (i_2, j_2) \in I_{RC}$, and any $v_1 = (b_1, j_1'), v_2 = (b_2, j_2') \in I'_{RC}$ we define
    \begin{align*}
        A_{s}^{RC; u_2 u_1} &:= \frac{\dimR}{|N_C|}\sum_{m \in N_C} R^{j_2j_1}(m)^* A_{s}^{q_{i_2} m \dash{q}_{i_1}}, \\
        \tilde{A}_n^{RC; v_2 v_1} &:=\frac{\dimR}{|N_C|} \sum_{m \in N_C} R^{j'_2j'_1}(m)  U_{b_2 b_1} L_{\bdy}^{{q_{i(b_1)} \dash{m} \,\dash{q}_{i(b_1)}}}
    \end{align*}
    where $\bdy$ is the boundary ribbon and $U_{b_2 b_1}$ is a boundary unitary provided by Lemma \ref{lem:transitive boundary action} which we choose such that $U_{b_2 b_1} = U_{b_1 b_2}^*$. These boundary unitaries satisfy the following: for any $\alpha \in \packC[C;ib_1]$ we have $U_{b_2 b_1} \ket{\alpha} = \ket{\alpha'}$ where $\alpha' \in \packC[C;i b_2]$, and $\alpha'_e = \alpha_e$ for all $e \in \eregion[n-1]$ and $b(\alpha') = b_2$.
\end{definition}

Note that the $\tilde{A}_n^{RC; v_2 v_1}$ are supported on $\eregion \setminus \eregion[n-1]$ and $A_{\site}^{RC; u_2 u_1}$ is supported on $\eregion[1]$. From Lemma \ref{lem:aconverter} we have for any $u, u_1, u_2 \in I_{RC}$ and any $v, v_1, v_2 \in I'_{RC}$ that
$$A_{\site}^{RC; u_2 u_1} \qdpure[RC;u_1 v] = \qdpure[RC;u_2 v], \quad \quad \tilde{A}_n^{RC; v_2 v_1} \qdpure[RC;u v_1] = \qdpure[RC; u, v_2]$$
as well as
$$(A_{\site}^{RC; u_1 u_2})^* \qdpure[RC;u_1 v] = \qdpure[RC;u_2 v], \quad \quad (\tilde{A}_n^{RC; v_1 v_2})^* \qdpure[RC;u v_1] = \qdpure[RC; u, v_2]$$
\ie these operators change the labels $u$ and $v$ when acting on the states $\qdpure$. We can use these `label changers' to show that expectation values of operators supported on $\eregion[n-1]$ in the state $\qdpure$ are independent of the boundary label $v$.

\begin{lemma}
\label{lem:purerest}
    We have for all $O \in \cstar[{\eregion[n-1]}]$, all $u, u' \in I_{RC}$, and all $v, v', v'' \in I'_{RC}$ that
    $$\inner{\eta_n^{uv}}{O \eta_n^{u'v'}} = \delta_{vv'} \inner{\eta_n^{u'v''}}{O \, \eta_n^{u'v''}}.$$
    In particular, $\inner{\qdstpure}{O \, \qdstpure}$ is independent of $v$. 
\end{lemma}

\begin{proof}
Using Lemma \ref{lem:aconverter} and the fact that any $O \in \cstar[{\eregion[n-1]}]$ commutes with $\tilde{A}_n^{RC;v v'}$ we find that if $v = (b,j)$ and $v' = (b',j')$ then 
\begin{align*}
    \inner{\eta_n^{uv}}{O \eta_n^{u'v'}} &= \inner{\eta_n^{uv}}{P_{b}O P_{b'} \eta_n^{u'v'}} = \delta_{bb'} \inner{\eta_n^{uv}}{O \eta_n^{u'v'}}\\ 
    &= \inner{\qdstpure[RC;u v'']}{O \, (\tilde{A}_n^{RC;v v''})^{*} \tilde{A}_n^{RC; v' v''} \qdstpure[RC;u' v'']} = \delta_{vv'} \inner{\eta_n^{u'v''}}{O \, \eta_n^{u'v''}},
\end{align*}
where in the last equality we have used the fact that if $v,v'$ have the same boundary $b$ then from lemma \ref{lem:action of two different label changers on eta}, $(\tilde{A}_n^{RC;v' v''})^{*} \tilde{A}_n^{RC; v v''} \ket{\qdstpure[RC;u' v'']}= \delta_{vv'}\ket{\qdstpure[RC;u' v'']}$.
\end{proof}

This Lemma shows that the following is well-defined.
\begin{definition} \label{def:uniform string net superpositions}
    For any $n$ we define the states $\qdstrest$ on $\cstar[{\eregion[n]}]$ by
    $$ \qdstrest(O) := \inner{\eta_{n+1}^{RC;uv}}{ O \, \eta_{n+1}^{RC;uv} } $$
    for any $O \in \cstar[{\eregion}]$ and any boundary label $v$. The choice of boundary label does not matter due to Lemma \ref{lem:purerest}.
\end{definition}

\subsection{Construction of the states $\qdstate$ and proof of their purity}

The following basic Lemma will be useful throughout the paper.
\begin{lemma}
    \label{lem:pure components projector Lemma}
    Let $\omega = \sum_{\kappa} \lambda_{\kappa} \omega^{(\kappa)}$ a state on $\cstar[\eregion]$ expressed as a finite convex combination of pure states $\omega^{(\kappa)}$ with positive coefficients $\lambda_{\kappa} > 0$. If $P \in \cstar[\eregion]$ is a projector and $\omega(P) = 1$, then $\omega^{(\kappa)}(P) = 1$ for all $\kappa$.
    
    Moreover, if $\ket{\Omega^{(\kappa)}} \in \hilb_n$ is a unit vector such that $\omega^{(\kappa)}(O) = \inner{\Omega^{(\kappa)}}{ O \, \Omega^{(\kappa)}}$ for all $O \in \cstar[\eregion]$, then $P \ket{\Omega^{(\kappa)}} = \ket{\Omega^{(\kappa)}}$.
\end{lemma}

\begin{proof}
       Since $\omega^{(\kappa)}(P) \leq 1$ and the positive numbers $\lambda_{\kappa} > 0$ sum to one, the equality
       $$1 = \omega(P) = \sum_{\kappa} \lambda_{\kappa} \omega^{(\kappa)}(P)$$ 
       can only be satisfied if $\omega^{(\kappa)}(P) = 1$ for all $\kappa$. If $\omega^{(\kappa)}(\cdot) = \inner{\Omega^{(\kappa)}}{ \cdot \, \Omega^{(\kappa)}}$ for a unit vector $\ket{\Omega^{(\kappa)}} \in \hilb_n$ then in particular
       $$1 = \omega^{(\kappa)}(P) = \inner{\Omega^{(\kappa)}}{\, P \, \Omega^{(\kappa)}}.$$
       Since $P$ is an orthogonal projector, this implies $P | \Omega^{(\kappa)} \rangle = \ket{\Omega^{(\kappa)}}$.
\end{proof}

Let us define the following sets of states on $\cstar[\eregion]$.
\begin{definition} \label{def:local RCu constraints for states}
    The set $\mathcal{S}_n^{RC;u}$ consists of states $\omega$ on $\cstar[\eregion]$ such that
    $$ 1 = \omega( D_{\site}^{RC;u} ) = \omega(A_v) = \omega(B_f) $$
    for all $v \in \dvregion$ and all $f \in \dfregion$.
\end{definition}

\begin{lemma}
\label{lem:restriction yields eta^RCu}
    Let $1 \leq m < n$. If $\omega \in \mathcal{S}_n^{RC;u}$, and $\omega_m$ is its restriction to $\cstar[{\eregion[m]}]$, then $\omega_m = \qdstrest[m]$.
\end{lemma}
\begin{proof}
    Let $\omega_{m+1}$ be the restriction of $\omega$ to $\cstar[{\eregion[m+1]}]$, then $\omega_{m+1} \in \mathcal{S}_{m+1}^{RC;u}$. Let $\omega_{m+1} = \sum_{\kappa} \lambda_{\kappa} \omega_{m+1}^{(\kappa)}$ be the convex decomposition of $\omega_{m+1}$ into finitely many pure components $\omega_{m+1}^{(\kappa)}$. Let $\ket{\Omega_{m+1}^{(\kappa)}} \in \hilb_{m+1}$ be unit vectors corresponding to these pure states. We conclude from Lemma \ref{lem:pure components projector Lemma} that 
    $$A_v \ket{\Omega_{m+1}^{(\kappa)}} = B_f \ket{\Omega_{m+1}^{(\kappa)}} = \qdsu \ket{\Omega_{m+1}^{(\kappa)}} = \ket{\Omega_{m+1}^{(\kappa)}}$$
    for all $\kappa$, all $v \in \dvregion[m+1]$, and all $f \in \dfregion[m+1]$. By Definition \ref{def:local constraints} this means that $\ket{\Omega_{m+1}^{(\kappa)}} \in \mathcal{V}_{m+1}^{RC;u}$ for all $\kappa$. From Proposition \ref{prop:local spaces spanned by etas} it follows that the unit vectors $\ket{\Omega_{m+1}^{(\kappa)}}$ are linear combinations of the $| \eta_{m+1}^{RC;uv} \rangle$ for $v \in I'_{RC}$. Using Lemma \ref{lem:purerest} and Definition \ref{def:uniform string net superpositions} we then have that for any $O \in \eregion[m]$
    \begin{align*}
        \omega_{m+1}^{(\kappa)}(O) &= \inner{ \Omega_{m+1}^{(\kappa)}}{ O \, \Omega_{m+1}^{(\kappa)}} = \eta_{m}^{RC;u}(O)
    \end{align*}
    independently of $\kappa$. The claim follows.
\end{proof}

We define extensions of the states $\qdstrest[n]$ to the whole observable algebra.
\begin{definition}
\label{def:excitation extension}
    We let $\tilde \eta_n^{RC;u}$ be the following extension of $\qdstrest[n]$ to the whole observable algebra. For each $e \in \latticeedge$, let $\zeta_e$ be the pure state on $\cstar[e]$ corresponding to the vector $| 1_e \rangle \in \hilb_e$, and put
    $$\tilde \eta_n^{RC;u} := \qdstrest[n] \otimes \left( \bigotimes_{e \in \latticeedge \setminus \eregion} \zeta_e \right).$$
\end{definition}

Recall the space of states $\S_{s_0}^{RC;u}$ from Definition \ref{def:state spaces}.

\begin{lemma}
\label{lem:qdconvergence}
    The sequence of states $\tilde \eta_n^{RC;u}$ converges in the weak-$^*$ topology to a state $\qdstate \in \S_{s_0}^{RC;u}$.
\end{lemma}

\begin{proof}
    If $O \in \cstar[{\eregion[m]}]$ then $\tilde \eta_n^{RC;u}(O) = \eta_n^{RC;u}(O)$ for all $n > m$ by construction. Since $\eta_n^{RC;u} \in \mathcal{S}_n^{RC;u}$ we have from Lemma \ref{lem:restriction yields eta^RCu} that $\tilde \eta_n^{RC;u}|_m = \qdstrest[m]$. It follows that $\tilde \eta_n^{RC;u}(O)$ is constant for all $n > m$ and hence converges. Since $m$ was chosen arbitrarily, $\tilde \eta_n^{RC;u}$ converges for any local observable $O \in \cstar[loc]$. Since $\cstar[loc]$ is dense in $\cstar$, the states $\omegaRCu$ converge in the weak-$^*$ topology to some state $\qdstate$ that satisfies the constraints \eqref{eq:cons1} and $\eqref{eq:cons3}$, i.e. $\qdstate \in \S_{\site}^{RC;u}$.
\end{proof}

\begin{lemma}
\label{lem:qdunique}
    $\qdstate$ is the unique state in $\S_{\site}^{RC;u}$. It is therefore a pure state.
\end{lemma}

\begin{proof}
    Consider any other state $\omega' \in \S_{\site}^{RC;u}$. Then its restriction $\omega'_n$ to $\cstar[\eregion]$ is a state in $\mathcal{S}^{RC;u}_n$. By Lemma \ref{lem:restriction yields eta^RCu} we have $\omega'(O) = \omega'_n(O) = \qdstrest[m](O) = \omegaRCu(O)$ for all $m < n$ and all $O \in \cstar[{\eregion[m]}]$. It follows that $\omega'$ agrees with $\qdstate$ on all local observables and therefore must be the same state.

    To see that this implies that $\qdstate$ is pure, suppose $\qdstate = \lambda \omega' + (1 - \lambda) \omega''$ can be written as a convex combination of states $\omega'$ and $\omega''$. Then for any $v \in \dlatticevert$ we have $1 = \omega_0(A_v) = \lambda \omega'(A_v) + (1 - \lambda) \omega''(A_v)$. Since $A_v$ is a projector we have $\abs{\omega'(A_v)}, \abs{\omega''(A_v)} \leq 1$, so the previous equality can only be satisfied if $1 = \omega'(A_v) = \omega''(A_v)$. By the same reasoning, $1 = \omega'(B_f) = \omega''(B_f)$ for all $f \in \dlatticeface$, and similarly for the projector $D_{s_0}^{RC;u}$. We conclude that $\omega'$ and $\omega''$ both belong to $\S_{\site}^{RC;u}$ and are therefore equal to $\qdstate$. Thus $\qdstate$ is pure.
\end{proof}

Since the site $\site$ was arbitrary, we have in particular shown
\begin{proposition} \label{prop:characterisation of S^RCu}
    For any site $s_0$, any irreducible representation $RC$ of $\caD(G)$, and any label $u$, the space of states $\S_{s_0}^{RC;u}$ of Definition \ref{def:state spaces} consists of a single pure state $\omega_{s_0}^{RC;u}$.
\end{proposition}

%% file: anyon_representations.tex
\section{Construction of anyon representations} \label{sec:anyon representations}

In this section we show that the pure states $\omega_{s_1}^{R_1C_1;u_1}, \omega_{s_2}^{R_2C_2;u_2}$ constructed in the previous section are equivalent to each other whenever $R_1C_1 = R_2C_2$. The collection of pure states $\{\omega_{s}^{RC;u}\}_{s, u}$ for fixed $RC$ therefore belong to the same irreducible representation $\pi^{RC}$ of the observable algebra. We will show that the irreducible representations $\{\pi^{RC}\}_{RC}$ are pairwise disjoint. In other words, we show that different $RC$ label different superselection sectors. Finally, we will show that the representations $\pi^{RC}$ are anyon representations by relating them to the so-called amplimorphism representations of \cite{Naaijkens2015-xj}.

\subsection{Ribbon operators and their limiting maps} \label{sec:ribbon operators and limiting maps}

From this point onward, the \emph{ribbon operators} introduced in Section \ref{subsec:preliminary notions} will play an increasingly important role in the analysis. By taking certain linear combinations of these ribbon operators, we construct new ribbon operators that can produce, transport, and detect anyonic excitations above the frustration free ground state. \\

Recall from  section \ref{subsec:preliminary notions} that we can associate to any finite ribbon $\rho$ some ribbon operators $F_{\rho}^{h, g}$. The following linear combinations of these ribbon operators are designed so that when acting on the ground state, they produce excitations that sit in irreducible representations for the action of the quantum double algebra at the endpoints of $\rho$.

\begin{definition}[\cite{Bombin2007-uw}] \label{def:RC ribbons}
    For each irreducible representation $RC$ of the quantum double we define
    \begin{align*}
        \frcuv :=\frac{\dimR}{|N_C|} \sum_{m \in N_C} {R}^{jj'}(m)^* F^{\dash{c}_i, q_i m \dash{q}_{i'}}_\rho 
    \end{align*}
    where $u = (i,j) \in I_{RC}$ and $v = (i',j') \in I_{RC}$.    
\end{definition}

\begin{definition}[\cite{naaijkens2012anyons}, \cite{Naaijkens2015-xj}] \label{def:finite mu}
    For any finite ribbon $\rho$, any $RC$, and any $u_1, u_2 \in I_{RC}$ we define a linear map from $\cstar$ to itself by
    $$ \mu^{RC; u_1 u_2}_{\rho}(O) := \bigg(\frac{|N_C|}{\dimR} \bigg)^2 \sum_{v} \, \big( F_{\rho}^{RC; u_1 v} \big)^* \, O \, F_{\rho}^{RC; u_2 v} $$
    for any $O \in \cstar$.
\end{definition}

We define a half-infinite ribbon to be a sequence $\rho = \{ \tau_i \}_{i \in \N}$ of triangles such that $\partial_1 \tau_i = \partial_0 \tau_{i+1}$ for all $i \in \N$, and such that for each edge $e \in \latticeedge$, there is at most one triangle $\tau_i$ for which $\tau_i = (\partial_0 \tau_i, \partial_1 \tau_i, e)$. We denote by $\partial_0 \rho = \partial_0 \tau_1$ the initial site of the half-infinite ribbon, and by $\rho_n = \{\tau_i\}_{i = 1}^n$ the finite ribbon consisting of the first $n$ triangles of $\rho$. A half-infinite ribbon is positive if all of its triangles are positive, and negative if all of its triangles are negative. Any half-infinite ribbon is either positive or negative.\\

The following Proposition due to \cite{Naaijkens2015-xj} says that we can define $\enRC[RC;u_1 u_2]$ as limits of $\mu_{\rho_n}^{RC;u_1 u_2}$, and states some properties of these limiting maps.

\begin{proposition}[Lemma 5.2 in \cite{Naaijkens2015-xj}] \label{prop:ampli properties}
    Let $\rho$ be a half-infinite ribbon. The limit
    $$\enRC[RC;u_1 u_2](O) := \lim_{n \rightarrow \infty} \mu_{\rho_n}^{RC;u_1 u_2} (O)$$
    exists for all $O \in \cstar$ and all $u_1, u_2 \in I_{RC}$, and defines a linear map from $\cstar$ to itself. Moreover, the maps $\mu_{\rho}^{RC;uu}$ are positive, and
    \begin{enumerate}
        \item if $O \in \cstar[loc]$ then there is a finite $n_0$ such that $\mu_{\rho}^{RC;u_1 u_2}(O) = \mu_{\rho_n}^{RC;u_1 u_2}(O)$ for all $n \geq n_0$,
        \item $\mu^{RC; u_1 u_2}_{\rho}(\I) = \delta_{u_1, u_2} \I$.
		\item $\mu^{RC; u_1 u_2}_{\rho}(O) = \delta_{u_1, u_2} O$ if the support of $O$ is disjoint from the support of $\rho$.
		\item $\mu^{RC; u_1 u_2}_{\rho}(OO') = \sum_{u_3 \in I_{RC}} \mu^{RC; u_1 u_3}_{\rho}(O) \mu^{RC; u_3 u_2}_{\rho}(O')$.
		\item $\mu_{\rho}^{RC; u_1 u_2}(O)^* = \mu_{\rho}^{RC; u_2 u_1}(O^*)$.
    \end{enumerate}
\end{proposition}

\begin{proof}
    The only thing that is not coming directly from \cite{Naaijkens2015-xj}'s Lemma 5.2 is the claim that the maps $\mu_{\rho}^{RC;uu}$ are positive. To see this, simply note that for any $O \in \cstar$ and using items 4 and 5 we have
    $$\mu_{\rho}^{RC;uu}(O^* O) = \sum_{v} \, \mu_{\rho}^{vu}(O)^* \, \mu_{\rho}^{vu}(O) \geq 0.$$
\end{proof}

Let $\omega_0$ be the frustration free ground state and $(\pi_0, \caH_0, | \Omega_0 \rangle)$ its GNS triple. We write $\chi_{\rho}^{RC; u v} := \pi_0 \circ \mu_{\rho}^{RC; u v} : \cstar \rightarrow \caB(\caH_0)$.

\begin{lemma} \label{lem:qdstate from mu action}
    Let $\rho$ be a half-infinite ribbon with $\partial_0 \rho = \site$ for any site $\site$. For any $RC$ and any $u \in I_{RC}$ we have
    $$\omega_{\site}^{RC; u} = \omega_0 \circ \mu_{\rho}^{RC;uu}.$$
\end{lemma}

\begin{proof}
    $\omega_0 \circ \mu_{\rho}^{RC;uu}$ is a positive linear functional by Proposition \ref{prop:ampli properties}. Normalisation follows from item 2 of Proposition \ref{prop:ampli properties}.

    Since $\omega_{\site}^{RC;u}$ is completely characterised by
    $$ 1 = \omega_{\site}^{RC;u}(A_v) = \omega_{\site}^{RC;u}(B_f) = \omega_{\site}^{RC;u}(D_{\site}^{RC;u})$$
    for all $v \in \dlatticevert$ and all $f \in \dlatticeface$ (Proposition \ref{prop:characterisation of S^RCu}), it is sufficient to show that $\omega_0 \circ \mu_{\rho}^{RC;uu}$ also satisfies these constraints.
    
    Since for any observable $O \in \cstar$ we have
    $$ (\omega_0 \circ \mu_{\rho}^{RC;uu})(O) = \langle \Omega_0, \chi_{\rho}^{RC;uu}(O) \, \Omega_0 \rangle,$$
    this follows immediately from Lemmas \ref{lem:ampli label projector} and \ref{lem:chi preserves constraints}.
\end{proof}

\begin{lemma} \label{lem:transport of anyons}
    For any two sites $s, s'$, any $RC$, and any two labels $u, u' \in I_{RC}$ there is a local operator $T \in \cstar[loc]$ such that
    $$ \omega_{s'}^{RC;u'}(O) = \omega_s^{RC;u}(T O T^*)  $$
    for all $O \in \cstar.$
\end{lemma}

\begin{proof}
    Let $\rho$ be a half-infinite ribbon having $\partial_0 \rho = s$, and a half-infinite subribbon $\rho'$ with $\partial_0 \rho' = s'$. Then $\rho = \rho_1 \rho'$ for a finite ribbon $\rho_1$. Let $O \in \cstar$. Using Lemma \ref{lem:qdstate from mu action} we now compute
    \begin{align*}
        \omega_{s}^{RC;u}(O) &= (\omega_0 \circ \mu_{\rho}^{RC;uu})(O) \\
        &= \lim_{n \uparrow \infty} \, \bigg(\frac{|N_C|}{\dimR} \bigg)^2 \,  \langle \Omega_0, \, \sum_{v} \, (F_{\rho_1 \rho'_n}^{RC;uv})^* \,O \, F_{\rho_1 \rho'_n}^{RC;uv} \, \Omega_0 \rangle \\
        \intertext{using Lemma \ref{lem:decomposition of F}}
        &= \lim_{n \uparrow \infty} \, \bigg(\frac{|N_C|}{\dimR} \bigg)^4 \sum_{v, w_1, w_2}  \, \langle \Omega_0, \, (F_{\rho'_n}^{RC;w_1v})^* (F_{\rho_1}^{RC;uw_1})^* \, O \, F_{\rho_1}^{RC;uw_2} \, F_{\rho'_n}^{RC;w_2 v} \, \Omega_0 \rangle \\
        \intertext{then using Lemma \ref{lem:change ribbon operator label}}
        &= \lim_{n \uparrow \infty} \, \bigg(\frac{|N_C|}{\dimR} \bigg)^4 \sum_{v, w_1, w_2}  \, \langle \Omega_0, \, (F_{\rho'_n}^{RC;u'v})^*  (A_{s'}^{RC;w_1 u'})^*  (F_{\rho_1}^{RC;uw_1})^* \\
        & \quad\quad\quad\quad\quad\quad\quad\quad\quad\quad\quad\quad \times \, O \, F_{\rho_1}^{RC;uw_2} \, A_{s'}^{RC;w_2 u'} F_{\rho'_n}^{RC;u'v} \, \Omega_0 \rangle \\
        &= (\omega_0 \circ \mu_{\rho'}^{RC;u'u'}) \bigg(  \bigg(\frac{|N_C|}{\dimR} \bigg)^2 \sum_{w_1, w_2} \, (A_{s'}^{RC;w_1 u'})^*  (F_{\rho_1}^{RC;uw_1})^* \, O \, F_{\rho_1}^{RC;uw_2} \, A_{s'}^{RC;w_2 u'} \bigg)
    \end{align*}
    which proves the claim with
    $$T =  \bigg(\frac{|N_C|}{\dimR} \bigg) \, \sum_w \, (A_{s'}^{RC;w u'})^*  (F_{\rho_1}^{RC;uw})^*.$$
\end{proof}

\subsection{Anyon representations labeled by $RC$} \label{sec:construction of piRC}

We define the following GNS representations.
\begin{definition} \label{def:piRC}
    Fix a site $\bsite$. For each $RC$, let $(\pi^{RC}, \hilb^{RC}, \ket{\Omega^{RC; (1, 1)}_{\bsite}})$ be the GNS triple for the pure state $\omega_{\bsite}^{RC;(1, 1)}$.
\end{definition}
Note that $\omega_{\bsite}^{\trivRC;(1, 1)} = \omega_0$ is the \ffgs, so $\pi^{\trivRC} = \pi_0$ is the ground state representation.

In this Section we will show that the representations $\{\pi^{RC}\}_{RC}$ are pairwise disjoint anyon representations with respect to the ground state representation $\pi_0$. In Section \ref{sec:completeness} we will show that any anyon representation is unitarily equivalent to one of the $\pi^{RC}$.

\begin{definition}
\label{def:state belongs to a representation}
    We say a state $\psi$ on $\cstar$ \emph{belongs} to a representation $\pi : \cstar \rightarrow \Bhilb$ of the observable algebra if there is a density operator $\rho \in \Bhilb$ such that
    $$\psi(O) = \Tr \lbrace \rho \pi(O) \rbrace$$
    for all $O \in \cstar$ (This notion is called being \emph{$\pi$-normal} in the operator algebra literature). If $\psi$ is pure and belongs to an irreducible representation $\pi$, then the corresponding density operator is a rank one projector, i.e. $\psi$ has a vector representative in the representation $\pi$. In this case we say $\psi$ is a vector state of $\pi$. Conversely, if $\psi$ belongs to a representation $\pi$, then we say $\pi$ \emph{contains} the state $\psi$.
\end{definition}

We first note that the representation $\pi^{RC}$ contains all the pure states $\{\omega_s^{RC;u}\}_{s, u}$. Since $\pi_{RC}$ is irreducible, it follows that all these states are equivalent to each other.

\begin{lemma} \label{lem:statesbelongingtopiRC}
    The pure states $\omega^{RC;u}_{s}$ are vector states of $\pi^{RC}$ for all sites $s$ and all $u \in I_{RC}$.
\end{lemma}

\begin{proof}
    This follows immediately from Lemma \ref{lem:transport of anyons}.
\end{proof}

We choose representative vectors for the states $\omega_{s}^{RC;(1, 1)}$ as follows.

\begin{definition}
    For all sites $s \neq \bsite$ we choose unit vectors $\ket{\Omega_s^{RC;(1, 1)}} \in \hilb^{RC}$ such that
    $$ \omega_s^{RC;(1, 1)}(O) = \inner{\Omega_{s}^{RC;(1, 1)}}{\pi^{RC}(O) \,\Omega_{s}^{RC;(1, 1)}} $$
    for all $O \in \cstar$. Such vectors exist by Lemma \ref{lem:statesbelongingtopiRC}, (note that the corresponding vector $\ket{\Omega_{\bsite}^{RC;(1, 1)}}$ for the site $\bsite$ was already fixed in Definition \ref{def:piRC}.)
\end{definition}

\subsection{Disjointness of the representations $\pi^{RC}$}

We prove that $\pi^{RC}$ and $\pi^{R'C'}$ are disjoint whenever $RC \neq R'C'$.

Let us first show the following basic Lemma, which is due to  \cite{alicki2007statistical}.
\begin{lemma} \label{lem:absorption of satisfied projectors}
    Let $\omega$ be a state on $\cstar$ and $P \in \cstar$ an orthogonal projector satisfying $\omega(P) = 1$. Then, $\omega(P O) = \omega(O P) = \omega(O)$ for all $O \in \cstar$.
\end{lemma}

\begin{proof}
    Using the Cauchy-Schwarz inequality,
    $$ \abs{\omega(O - PO)}^2 = \abs{\omega(O (\I - P) (\I - P))}^2 \leq \omega( O (\I - P) O^* ) \omega(\I - P) = 0  $$
    which show that $\omega(O) = \omega(PO)$. The equality $\omega(O) = \omega(OP)$ is shown in the same way.
\end{proof}

\begin{definition} (\cite[Eq. (B75)]{Bombin2007-uw}) \label{def:charge detectors}
    For any closed ribbon $\sigma$ we put
    $$\knRC[\sigma] := \dimchbasis \sum_{m \in N_C} \chi_R(m)^* \sum_{q \in Q_C} F_\sigma^{q m \dash{q}, q r_C \dash{q}}.$$
\end{definition}

\begin{lemma} \label{lem:projector Lemma}
    If $\omega \in \overline{\S}_{\site}$ then $\omega( O \qds ) = \omega(O K_{\bdy}^{RC}) = \omega(\qds O) = \omega(K_{\bdy}^{RC} O )$ for all $RC$, all $n>1$, and all $O \in \cstar[{\eregion[n]}]$. In fact, $\omega(O D_{\site}^{RC}) = \omega(D_{\site}^{RC} O)$ holds for all $O \in \cstar$.
\end{lemma}

\begin{proof}
    The restriction $\omega_n$ of $\omega$ to $\cstar[\eregion]$ satisfies
    $$ 1 = \omega_n(A_v) = \omega_n(B_f) $$
    for all $v \in \dvregion$ and all $f \in \dfregion$. Let $\omega_n = \sum_{\kappa} \lambda_{\kappa} \, \omega_n^{(\kappa)}$ be the convex decomposition of $\omega_n$ into its pure components $\omega_n^{(\kappa)}$, and let $\ket{\Omega_n^{(\kappa)}} \in \caH_n$ be unit vectors such that
    $$\omega_n^{(\kappa)}(O') = \inner{\Omega_n^{(\kappa)}}{O' \, \Omega_n^{(\kappa)}}$$
    for all $O' \in \cstar[\eregion]$. From Lemma \ref{lem:pure components projector Lemma} we find that
    $$ 1 = \omega_n^{(\kappa)}(A_v) = \omega_n^{(\kappa)}(B_f)$$
    and
    $$ \ket{\Omega_n^{(\kappa)}} = A_v \ket{\Omega_n^{(\kappa)}} = B_f \ket{\Omega_n^{(\kappa)}}$$
    for all $v \in \dvregion$ and all $f \in \dfregion$.

    Consider for each $RC$ and each $\kappa$ the vector $\ket{\Omega_n^{(RC, \kappa)}} := D_{\site}^{RC} \ket{\Omega_n^{(\kappa)}}$. Since the $A_v, B_f$ commute with $\qds$ for $v \in \dvregion$ and $f \in \dfregion$ (Lemma \ref{lem:qdcommute}), we have
    $$ \ket{\Omega_n^{(RC, \kappa)}} = A_v \ket{\Omega_n^{(RC,\kappa)}} = B_f \ket{\Omega_n^{(RC,\kappa)}} = D_{\site}^{RC} \ket{\Omega_n^{(RC, \kappa)}} $$
    for all $RC$, $\kappa$, $v \in \dvregion$ and $f \in \dfregion$. \ie we have $\ket{\Omega_n^{(RC,\kappa)}} \in \VRC$ (cf. Definition \ref{def:local constraints}). It then follows from Proposition \ref{prop:local spaces spanned by etas} that
    $$ \ket{\Omega_n^{(RC,\kappa)}} = \sum_{uv} c^{RC, \kappa}_{uv} \qdpure $$
    for some coefficients $c^{RC, \kappa}_{uv} \in \C$. Since $\sum_{RC} \qds = \mathds{1}$ (Lemma \ref{lem:DRCprops}) it follows that
    $$ \ket{\Omega_n^{(\kappa)}} =  \sum_{RC}  \ket{\Omega_n^{(RC,\kappa)}} = \sum_{RC} \sum_{uv} c^{RC, \kappa}_{uv} \qdpure. $$

    We now use Lemma \ref{lem:topological charge detector} to obtain
    $$K_{\bdy}^{R' C'} \ket{\Omega_n^{(\kappa)}} = \sum_{RC} \sum_{uv} c_{uv}^{RC, \kappa} \, K_{\bdy}^{R'C'} \qdpure = \sum_{uv} c_{uv}^{R'C', \kappa} | \eta_n^{R'C';uv} \rangle = D_{\site}^{R'C'} \ket{\Omega_n^{(\kappa)}}$$
    from which it follows that
    $$\omega(O D_{\site}^{RC}) = \sum_{\kappa} \, \lambda_{\kappa} \inner{ \Omega_n^{(\kappa)} }{ O D_{\site}^{RC} \, \Omega_n^{(\kappa)} } = \sum_{\kappa} \, \lambda_{\kappa} \inner{ \Omega_n^{(\kappa)} }{ O K_{\bdy}^{RC} \, \Omega_n^{(\kappa)}} = \omega(O K_{\bdy}^{RC})$$
    for all $O \in \cstar[\eregion]$. Using that the $K_{\sigma}^{RC}$ are hermitian (Lemma \ref{lem:basic properties of K}) and the elementary fact that $\omega(A^* B) = \overline{\omega(B^* A)}$ for all $A, B \in \cstar$ we also get
    $$\omega(D_{\site}^{RC} O) = \omega( K_{\bdy}^{RC} O)$$
    for all $O \in \cstar[\eregion]$.
    
    To show the second claim, note that for any $O \in \cstar[loc]$ we can take $n$ large enough so that $O \in \cstar[{\eregion[n-1]}]$. Then $[K_{\bdy}^{RC}, O] = 0$, so $\omega(O K_{\bdy}^{RC}) = \omega(K_{\bdy}^{RC} O)$. Using the results $\omega(K_{\bdy}^{RC} O) = \omega(D_{\site}^{RC} O)$ and $\omega(O K_{\bdy}^{RC}) = \omega(O D_{\site}^{RC})$ obtained above, we get
    $$ \omega(D_{\site}^{RC} O) =  \omega(O D_{\site}^{RC})$$
    for any $O \in \cstar[loc]$. This result extends to all $O \in \cstar$ by continuity.
\end{proof}

\begin{lemma} 
\label{lem:detection Lemma}
    If $\omega \in \S_{\site}^{RC}$ then
    $$ \omega( K^{R'C'}_{\bdy}) = \delta_{RC, R'C'}$$
    for all $n > 1$.
\end{lemma}

\begin{proof}
    By definition $\omega(D_{\site}^{RC}) = 1$ and so by Lemma \ref{lem:absorption of satisfied projectors} we have $\omega(O) = \omega(O D_{\site}^{RC}) = \omega(D_{\site}^{RC} O)$ for all $O \in \cstar$. We take $O = K^{R'C'}_{\bdy}$ and use Lemmas \ref{lem:projector Lemma} and \ref{lem:DRCprops} to obtain
    $$ \omega(K^{R'C'}_{\bdy}) = \omega( D_{\site}^{RC} K^{R'C'}_{\bdy}) = \omega( D_{\site}^{RC} D_{\site}^{R'C'} ) = \delta_{RC, R'C'} \omega(D_{\site}^{RC}) = \delta_{RC, R'C'}$$
    as required.
\end{proof}

\begin{lemma}
\label{lem:piRCaredisjoint}
    The representations $\pi^{RC}$ and $\pi^{R'C'}$ are disjoint whenever $RC \neq R'C'$
\end{lemma}

\begin{proof}
    We have $\omega_{\site}^{RC;(1, 1)} \in \S_{\site}^{RC}$ so Lemma \ref{lem:detection Lemma} says
    $$ \omega_{\bsite}^{RC, (1, 1)}( K_{\bdy}^{R'C'} ) = \delta_{RC,R'C'}$$
    for all $n > 1$. Noting that for any finite region $S \subset \latticeedge$ we can take $n$ large enough such that the projectors $K_{\beta_n}^{RC} \in \cstar[loc]$ are supported outside $S$, the claim follows from Corollary 2.6.11 of \cite{Bratteli2012-gd}.
\end{proof}

\subsection{Construction of amplimorphism representations}

In order to show that the representations $\pi^{RC}$ are anyon representations we first show that they are unitarily equivalent to so-called amplimorphism representations. These are representations which can be obtained from the ground state representation $\pi_0$ by composing $\pi_0 \otimes \id_{\abs{I_{RC}}}$ with an \emph{amplimorphism} $\cstar \rightarrow \mathcal{M}_{\abs{I_{RC}}}(\cstar)$, whose components are given by the maps $\mu_{\rho}^{RC;u_1 u_2}$ for a fixed half-infinite ribbon $\rho$. By the properties listed in Proposition \ref{prop:ampli properties}, this amplimorphism is a homomorphism of $C^*$-algebras, and the composition with $\pi_0$ yields a representation of $\cstar$. The fact that the amplimorphism acts non-trivially only near the ribbon $\rho$ will allow us to establish in Section \ref{sec:pi^RC are anyons representations} that the representations $\pi^{RC}$ are anyons representations.

Amplimorphisms, specifically in the context of non-abelian quantum double models, were first introduced in \cite{Naaijkens2015-xj}. Our presentation here is essentially a completion of the arguments sketched in that work. Amplimorphisms were originally introduced as a tool to investigate generalized symmetries in lattice spin models and quantum field theory, see \cite{szlachanyi1993quantum, vecsernyes1994quantum, fuchs1994quantum, nill1997quantum}.\\

Recall that $(\pi_0, \caH_0, \ket{\Omega_0})$ is the GNS triple of the frustration free ground state $\omega_0$ and we put $\chi_{\rho}^{RC; u v} := \pi_0 \circ \mu_{\rho}^{RC; u v} : \cstar \rightarrow \caB(\caH_0)$. For the remainder of this section we will often write $O$ instead of $\pi_0(O)$ when we are working in the faithful representation $\pi_0$.

We now define the \emph{amplimorphism representation}.
\begin{definition} \label{def:amplimorphism representation}
	We set
	\begin{equation*} 
		\chi_{\rho}^{RC} : \cstar \rightarrow \caB(\caH_0) \otimes M_N(\C) : O \mapsto \begin{bmatrix}
		\chi_{\rho}^{RC; u_1 u_1}(O) & \cdots & \chi_{\rho}^{RC; u_1 u_N}(O) \\
		\vdots & \ddots & \vdots \\
		\chi_{\rho}^{RC; u_N u_1}(O) & \cdots & \chi_{\rho}^{RC; u_N u_N}(O) 
	\end{bmatrix}
	\end{equation*}
	where $N = \abs{I_{RC}} = \abs{C} \dimR$ is the number of distinct values that the label $u \in I_{RC}$ can take.
\end{definition}
Using the properties listed in Proposition \ref{prop:ampli properties}, one can easily check that this is a unital *-representation of the quasi-local algebra.

The amplimorphism representation is carried by the Hilbert space $\caH := \caH_0 \otimes \C^N$. We choose a basis $\{ | u \rangle \}_{u \in I_{RC}}$ of $\C^N$ such that
\begin{equation*}
	\langle \Phi \otimes u, \chi^{RC}_{\rho}(O) \, \Psi \otimes v \rangle = \langle \Phi, \chi_{\rho}^{RC;u v}(O) \, \Psi \rangle
\end{equation*}
for all $| \Phi \rangle, | \Psi \rangle \in \caH_0$.

\subsection{Unitary equivalence of $\chi_{\rho}^{RC}$ and $\pi^{RC}$}

Let us fix a half-infinite ribbon $\rho$ with $\partial_0 \rho = \site$.

\begin{lemma} \label{lem:omega^u is vector state}
	For any $u \in I_{RC}$, the vector $| \Omega_0 \otimes u \rangle$ represents the pure state $\omega_{\site}^{RC; u}$ in the representation $\chi_{\rho}^{RC}$.
\end{lemma}

\begin{proof}
	For any $O \in \cstar$ we have
	\begin{equation*}
		\langle \Omega_0 \otimes u, \chi_{\rho}^{RC}(O) \, \Omega_0 \otimes u \rangle = \langle \Omega_0, \chi_{\rho}^{RC; u u}(O) \, \Omega_0 \rangle = \omega_0 \circ \mu_{\rho}^{RC; u u}(O) = \omega_{s_0}^{RC; u}(O)
	\end{equation*}
	where the last step is by Lemma \ref{lem:qdstate from mu action}.
\end{proof}
\ie the amplimorphism representation $\chi_{\rho}^{RC}$ contains the state $\omega^{RC; u}$, and therefore has a subrepresentation that is unitarily equivalent to $\pi^{RC}$. In fact, we will show that $\chi_{\rho}^{RC}$ is unitarily equivalent to $\pi^{RC}$. To do this, we must show that $| \Omega_0 \otimes u \rangle$ is a cyclic vector for $\chi_{\rho}^{RC}$.

\begin{proposition} \label{prop:Omega^u is cyclic}
	For any $u \in I_{RC}$, the vector $| \Omega_0 \otimes u \rangle$ is cyclic for $\chi_{\rho}^{RC}$.
\end{proposition}

\begin{proof}
    Let
    \begin{equation*}
    \caH_u := \overline{ \chi_{\rho}^{RC}(\cstar) | \Omega_0 \otimes u \rangle } \subseteq \caH.
    \end{equation*}
    We show that actually $\caH_u = \caH$.

    Consider the subspace
    \begin{equation*}
        \caV := \left\lbrace \sum_{v \in I_{RC}} \,  | \Psi_v \otimes v \rangle \, : \,  |\Psi_v\rangle \in \cstar[loc] | \Omega_0 \rangle \right\rbrace \subset \caH.
    \end{equation*}
    This space is dense in $\caH$.
    
    Take any vector $| \Psi \rangle = \sum_v | \Psi_v \otimes v \rangle \in \caV$ such that $| \Psi_v \rangle = O_v | \Omega_0 \rangle$ with $O_v \in \cstar[loc]$ for each $v \in I_{RC}$.
    
    We want to show that if $|\Psi \rangle$ is approximately orthogonal to $\caH_u$, then $\norm{|\Psi\rangle}$ is small. So let $P_u$ be the orthogonal projector onto $\caH_u$ and suppose that $\norm{ P_u | \Psi \rangle }^2 < \ep$.
    
    We now use the maps $t_{\rho_n}^{RC;vu}$ from Definition \ref{def:magic map}. For any $n \in \N$, any $RC$, and any $u, v \in I_{RC}$, these maps are given by
    $$ t_{\rho_n}^{RC;u v}( O ) :=  \bigg(\frac{\dimR}{|N_C|} \bigg)^2 \, \sum_{w, u'} \, F_{\rho_n}^{RC;u w}  \, O \, \big( F_{\rho_n}^{RC;u' w} \big)^* A_{\site}^{RC;u' v} D_{\site}^{RC;v}$$
    for any $O \in \cstar$. Here $a_{\site}^{RC;u' u}$ is the label changer of Definition \ref{def:label changers}, and $D_{\site}^{RC;u}$ is the projector of Definition \ref{def:Wigner projectors}. For any $O \in \cstar[loc]$, Lemma \ref{lem:magic map} says that for $n$ large enough
    $$\chi_{\rho}^{RC;u_1 v_1} \big( t_{\rho_n}^{RC;u_2 v_2}(O) \big) | \Omega_0 \rangle = \delta_{u_1 u_2} \delta_{v_1 v_2} \, O | \Omega_0 \rangle.$$
    We can therefore take $n$ large enough such that
    \begin{align*}
		\chi^{RC}_\rho \big( t^{RC;v u}_{\rho_n}( O_v ) \big) | \Omega_0 \otimes u \rangle &= \sum_{w} \, | \chi^{RC;w u}_\rho( t_{\rho_n}^{RC; v u}(O_v) ) \Omega_0 \otimes w \rangle \\
      &= \sum_w \, \delta_{w, v} \, | O_v \Omega_0 \otimes w \rangle
      = | \Psi_v \otimes v \rangle \in \caH_u
    \end{align*}
    
    It now follows from our assumption $\norm{ P_u | \Psi \rangle }^2 < \ep$ that
    \begin{equation*}
	 \norm{ | \Psi_v \rangle}^2 = \abs{ \langle \Psi, \chi^{RC}_\rho \big( t^{RC;v u}_{\rho_n}( O_v ) \big) \, \Omega_0 \otimes u \rangle } < \ep
    \end{equation*}
    for all $v \in I_{RC}$ and therefore
    \begin{equation*}
		\norm{ |\Psi \rangle }^2 = \sum_v \, \norm{ | \Psi_v \rangle}^2 < N \ep.
    \end{equation*}

    Now take $| \Psi \rangle \in \caH$ and suppose that $| \Psi \rangle$ is orthogonal to $\caH_u$. Since $\caV$ is dense in $\caH$ there is a sequence of vectors $| \Psi_i \rangle \in \caV$ that converges to $| \Psi \rangle$ in norm. For any $\ep > 0$ we can find $i_0$ such that
    \begin{equation*}
		\norm{ P_u | \Psi_i \rangle  }^2 = \norm{P_u( | \Psi_i \rangle - | \Psi \rangle )}^2 < \ep
    \end{equation*}
    for all $i \geq i_0$. From the above, we conclude that
    \begin{equation*}
		\norm{ | \Psi_i \rangle }^2 < N \ep
    \end{equation*}
    for all $i \geq i_0$. We see that the sequence converges to zero, so $| \Psi \rangle = 0$.

    Since any vector in $\caH$ that is orthogonal to $\caH_u$ must vanish, and since $\caH_u \subseteq \caH$, we find that $\caH_u = \caH$. This shows that $| \Omega_0 \otimes u \rangle$ is a cyclic vector for the representation $\chi_{\rho}^{RC}$.
\end{proof}

\begin{proposition} \label{prop:equivalence to pi^RC}
    For any half-infinite ribbon $\rho$ with initial site $s = \partial_0 \rho$, any $RC$, and any $u \in I_{RC}$, the amplimorphism representation $\chi_{\rho}^{RC}$ is unitarily equivalent to the GNS representation of the pure state $\omega_{s}^{RC;u}$. In particular, the amplimorphism representation $\chi_{\rho}^{RC}$ is irreducible and unitarily equivalent to $\pi^{RC}$.
\end{proposition}

\begin{proof}
    Unitary equivalence to the GNS representation of $\omega_{s}^{RC;u}$ follows immediately from Lemma \ref{lem:omega^u is vector state} and Proposition \ref{prop:Omega^u is cyclic}. Since $\omega_{s}^{RC;u}$ is a vector state in the irreducible representation $\pi^{RC}$ (Lemma \ref{lem:statesbelongingtopiRC}) we find that $\chi_{\rho}^{RC}$ is unitarily equivalent to $\pi^{RC}$ and in particular irreducible. 
\end{proof}

\subsection{The representations $\pi^{RC}$ are anyon representations} \label{sec:pi^RC are anyons representations}

For any cone $\Lambda$, let $\caR(\Lambda) := \pi_0(\cstar[\Lambda])'' \subset \caB(\caH_0)$ be the von Neumann algebra generated by $\cstar[\Lambda]$ in the ground state representation.\\

The following Proposition is a slight adaptation of part of Theorem 5.4 of \cite{Naaijkens2015-xj}. 
\begin{proposition}[\cite{Naaijkens2015-xj}] \label{prop:reduction to H_0}
    If $\rho$ is a half-infinite ribbon whose support is contained in a cone $\Lambda$, then there is a representation $\nu_{\Lambda}^{RC} : \cstar \rightarrow \caB(\caH_0)$ which is unitarily equivalent to the amplimorphism representation $\chi_{\rho}^{RC}$ and satisfies
    \begin{equation*}
		\nu_{\Lambda}^{RC}(O) = \pi_0(O)
    \end{equation*}
    for all $O \in \cstar[\Lambda^c]$.
\end{proposition}

The proof is exactly the same as that of Theorem 5.4 of \cite{Naaijkens2015-xj} and we simply point out that the Haag duality assumed in that Theorem is not needed for this part of the statement.

We can now show
\begin{proposition} \label{prop:the pi^RC are anyon representations}
    The representations $\pi^{RC}$ are anyon representations.
\end{proposition}

\begin{proof}
    Fix any cone $\Lambda$. We want to show that
    \begin{equation*}
		\pi^{RC}|_{\cstar[\Lambda]} \simeq \pi_0|_{\cstar[\Lambda]}.
    \end{equation*}
    To this end, let $\rho$ be a half-infinite ribbon supported in a cone $\widetilde \Lambda$ that is disjoint from $\Lambda$. Since $\pi^{RC}$ is unitarily equivalent to the amplimorphism representation $\chi_{\rho}^{RC}$ (Proposition \ref{prop:equivalence to pi^RC}), we get from Proposition \ref{prop:reduction to H_0} a representation $\nu_{\widetilde \Lambda}^{RC} : \cstar \rightarrow \caB(\caH_0)$ that is unitarily equivalent to $\pi^{RC}$ and satisfies $\nu_{\widetilde \Lambda}^{RC}(O) = \pi_0(O)$ for all $O \in \cstar[\Lambda] \subset \cstar[\widetilde \Lambda^c]$. By the unitary equivalence, there is a unitary $U_{\Lambda} \in \caB(\caH_0)$ such that
    \begin{equation*}
		\pi^{RC}(O) = U_{\Lambda} \, \nu_{\widetilde \Lambda}^{RC}(O) \, U_{\Lambda}^*
    \end{equation*}
    for all $O \in \cstar$. For $O \in \cstar[\Lambda]$ we therefore find
    \begin{equation*}
		\pi^{RC}(O) = U_{\Lambda} \, \pi_0(O) \, U_{\Lambda}^*.
    \end{equation*}
    Since the cone $\Lambda$ was arbitrary, this proves the proposition.
\end{proof}

%% file: completeness.tex
\section{Completeness}
\label{sec:completeness}

In order to show that all anyon representations are unitarily equivalent to one of the representations $\pi^{RC}$, we prove that any anyon representation $\pi$ contains a pure state $\qdstate$. This we do as follows. In subsection \ref{sec:equivtofinite} we  show that any anyon representation $\pi$ contains a pure state that is gauge invariant and has trivial flux everywhere outside of a finite region. In subsection \ref{sec:sweep}, we show that such a state is unitarily equivalent to a pure state satisfying \eqref{eq:cons1}. Lastly in subsection \ref{sec:classify} we will show that any pure state satisfying \eqref{eq:cons1} is equivalent to some $\qdstate$ and therefore belongs to a definite anyon representation $\pi^{RC}$. Combining these results with the results of the previous Section, we find that the anyon sectors are in one-to-one correspondence with equivalence classes of irreducible representations of the quantum double of $G$.

\subsection{Any anyon representation contains a state that is gauge invariant and has trivial flux outside of a finite region}
\label{sec:equivtofinite}

Let $\pi : \cstar \rightarrow \mathcal{\hilb}$ be an anyon representation.

For any $S \subset \R^2$ let $S^V := \{ v \in \latticevert \, : \, A_v \in \cstar[S] \}$ and $S^F := \{ f \in \latticeface \, : \, B_f \in \cstar[S] \}$. We also write $S^{VF} = S^V \cup S^F$. If $S \subset \R^2$ is bounded then we define
$$P_{S} := \prod_{v \in S^V} A_v \prod_{f \in S^F} B_f$$
which is a projector in $\cstar[S]$. It projects onto states that are gauge invariant and flat on $S$.

We want to define analogous projectors for infinite regions, but clearly such projectors cannot exists in the quasi-local algebra $\cstar$. Instead, we will construct them in the von Neumann algebra $\pi(\cstar)''$.

\begin{definition}
    A non-decreasing sequence $\{S_n\}$ of sets $S_n \subset \R^2$ is said to converge to $S \subset \R^2$ if $\bigcup_{n} S_n = S$.   
\end{definition}

Let $\{S_n\}$ be a non-decreasing sequence of bounded subsets of $\R^2$ converging to a possibly infinite $S \subset \R^2$. Then the sequence of orthogonal projectors $\{\pi(P_{S_n})\}$ is non-increasing and therefore converges in the strong operator topology to the orthogonal projector $p_S$ onto the intersection of the ranges of the $\pi(P_{S_n})$ \cite[Thm 4.32(a)]{Weidmann}. In particular, the limit $p_S$ does not depend on the particular sequence $\{S_n\}$. If $S$ is finite, then $p_S = \pi(P_S)$.

For any $r \geq 0$, let $B_r = \{ x \in \R^2 \, : \norm{x}_2 \leq r\}$ be the closed ball of radius $r$. 

\begin{proposition} \label{prop:equivtofinite}
    For any anyon representation $\pi : \cstar \rightarrow \mathcal{B}(\hilb)$ there is an $n > 0$ and a pure state $\omega$ belonging to $\pi$ such that
    $$\omega(A_v) = \omega(B_f) = 1$$
    for all $v$ and $f$ such that $A_v, B_f \in \cstar[B_n^c]$.
\end{proposition}

\begin{proof}
    Take two cones $\Lambda_1, \Lambda_2$ such that any $A_v$ and any $B_f$ is supported in (at least) one of them. In particular, $\Lambda_1 \cup \Lambda_2 = \R^2$. Since $\pi$ is an anyon representation, we have unitaries $U_i : \hilb_0 \rightarrow \hilb$ for $i = 1,2$ that satisfy
    $$\pi(O) = U_i \pi_0(O) U_i^* \qquad \forall O \in \cstar[\Lambda_i].$$
    It follows that
    $$\omega_0(O) = \inner{\Omega_0}{\pi_0(O) \Omega_0} = \inner{U_i\Omega_0}{\pi(O) U_i \Omega_0} \quad \forall \, O \in \cstar[\Lambda_i].$$
    Define pure states $\omega_i$ on $\cstar$ given by $\omega_i(O) := \inner{\Omega_i}{\pi(O) \Omega_i}$ where $\ket{\Omega_i} = U_i\ket{\Omega_0} \in \hilb$. The states $\omega_i$ belong to $\pi$ and satisfy $\omega_i(O) = \omega_0(O)$ for all $O \in \cstar[\Lambda_i]$.

    Let $\Lambda_i^{> n} := \Lambda_i \setminus B_n$ and $\Lambda_i^{n, n+m} := \Lambda_i^{>n} \setminus \Lambda_i^{>n+m}$ for all $m, n \in \mathbb{N}$. Then the sequence $m \mapsto \Lambda_i^{n, n+m}$ is a non-decreasing sequence of bounded sets converging to $\Lambda_i^{> n}$. We have
    $$1 = \omega_0 \big( P_{\Lambda_i^{n, n+m}} \big) = \omega_i \big( P_{\Lambda_i^{n, n+m}} \big) = \inner{\Omega_i}{ \pi \big( P_{\Lambda_i^{n, n+m}} \big)  \Omega_i},$$
    where we used that all these projectors are supported in $\Lambda_i$. We now find that
    $$\inner{\Omega_i}{p_{\Lambda_i^{> n}} \Omega_i} = 1$$
    where $p_{\Lambda_i^{> n}}$ is the strong limit of the sequence of projectors $\{ \pi \big( P_{\Lambda_i^{n, n+m}} \big) \}_{m > n}$.

    The pure states $\omega_1$ and $\omega_2$ are unitary equivalent since they are both vector states in the irreducible representation $\pi$. It follows from Corollary 2.6.11 of \cite{Bratteli2012-gd} that for any $\epsilon > 0$ there is an $n(\epsilon) \in \mathbb{N}$ such that
    $$|\omega_1(O) - \omega_2 (O)| \leq \epsilon \norm{O}$$
    for all $O \in \cstar[loc] \cap \cstar[ B_{n(\ep)}^c ]$.

    This gives us
    $$|\inner{\Omega_1}{\pi(P_{\Lambda_i^{n, n+m}}) \Omega_1} - \inner{\Omega_2}{\pi(P_{\Lambda_i^{n, n+m}}) \Omega_2}| < \epsilon$$
    for all $m$ and $n \geq n(\epsilon)$ for some $n(\epsilon) \in \mathbb{N}$. After taking the strong limit we get
    $$\inner{\Omega_1}{p_{\Lambda_2^{>n}} \Omega_1} >1 - \epsilon$$
    where we have used $\inner{\Omega_2}{\pi(P_{\Lambda_2^{n, n+m}}) {\Omega_2}} = 1$.

    Using the fact that $p_{\Lambda_1^{>n}}$ is a projector and $\inner{\Omega_1}{p_{\Lambda_1^{>n}} \Omega_1} = 1$ for all $n \geq n(\epsilon)$, we also have that $p_{\Lambda_1^{>n}} \ket{\Omega_1} = \ket{\Omega_1}$ for all such $n$.

    Now we use the fact that $p_{\Lambda_2^{>n}} p_{\Lambda_1^{>n}}$ projects onto a subspace of the range of $p_{B_{n+1}^c}$ to obtain
    \begin{align*}
        \inner{\Omega_1}{p_{B_{n+1}^c} \Omega_1} \geq  \inner{\Omega_1}{p_{\Lambda_2^{>n}} p_{\Lambda_1^{>n}} \Omega_1} = \inner{\Omega_1}{p_{\Lambda_2^{>n}}  \Omega_1} > 1-\epsilon
    \end{align*}
    for all $n > n(\epsilon)$.

    It follows that for $\epsilon < 1$ we have $p_{B_{n+1}^c} \ket{\Omega_1} \neq 0$ for all $n \geq n(\epsilon)$. Let us therefore fix some $\epsilon < 1$ and $n \geq n(\epsilon)$, and define a normalized vector
    $$\ket{\Omega} := \frac{p_{B_{n+1}^c} \ket{\Omega_1}}{||p_{B_{n+1}^c} \ket{\Omega_1}||} \in \hilb.$$ 

    The vector $\ket{\Omega}$ defines a pure state by $\omega(O):=\inner{\Omega}{\pi(O)\Omega}$ for all $O \in \cstar$. This state belongs to the anyon representation $\pi$.
    
    To finish the proof, we verify that $\omega(A_v) = \omega(B_f) = 1$ whenever $A_v ,B_f \in \cstar[B_{n+1}^c]$. We have
    $$
        \omega(A) = \frac{\inner{\Omega_1}{p_{B_{n+1}^c} \pi(A_v) p_{B_{n+1}^c}\Omega_1}}{||p_{B_{n+1}^c} \ket{\Omega_1}||^2}
        =\frac{\inner{\Omega_1}{p_{B_{n+1}^c} p_{\{v\}} p_{B_{n+1}^c}\Omega_1}}{||p_{B_{n+1}^c} \ket{\Omega_1}||^2} = \frac{\inner{\Omega_1}{p_{B_{n+1}^c} p_{B_{n+1}^c}\Omega_1}}{||p_{B_{n+1}^c} \ket{\Omega_1}||^2} = 1
    $$
    where we used $p_{B_{n+1}^c} \pi(A_v) = p_{B_{n+1}^c}$. The proof that $\omega_\pi(B_f) = 1$ whenever $B_f \in \cstar[B_{n+1}^c]$, is identical. This concludes the proof of the Proposition.
\end{proof}

\subsection{Finite violations of ground state constraints can be swept onto a single site} \label{sec:sweep}

Let $\omega$ be a pure state on $\cstar$ and let $(\pi, \hilb, |\Omega \rangle)$ be its GNS triple. Since $\cstar$ is a simple algebra, the representation $\pi$ is faithful and we can identify $\cstar$ with its image $\pi(\cstar)$. For any $O \in \cstar$ we will write $O$ instead of $\pi(O)$ in the remainder of this Section.

For the remainder of this Section we fix an arbitrary site $s_* = (v_*, f_*)$. We will show that we can move any finite number of violations of the ground state constraints onto the site $s_*$.

The following Lemma will prove useful in achieving this:

\begin{lemma}
\label{lem:sweep identities}
    (\cite{Bombin2007-uw}, Eq. B46, B47) Let $\rho$ be a ribbon such that $v = v(\start)$ or $v = v(\en)$ and $v(\start) \neq v(\en)$, and $ f = f(\start) $ or $f = f(\en)$ and $f(\start) \neq f(\en)$. Then we have the following identities:
    \begin{equation}
    \label{eq:sweep}
        {|G|}\sum_{g \in G} T^g_\rho A_v T^g_\rho = 1, \qquad \qquad \sum_{h \in G}L^{\bar{h}}_\rho B_f L^h_\rho = 1
    \end{equation}
\end{lemma}

For any vector $\ket{\Psi} \in \caH$, let us define
$$ V_{\Psi} := \{ v \in \latticevert \, : \, A_v \ket{\Psi} \neq \ket{\Psi}  \} \setminus \{ v_* \} $$
and
$$ F_{\Psi} := \{ f \in \latticeface \, : \, B_f \ket{\Psi} \neq \ket{\Psi} \} \setminus \{ f_* \}. $$
The set $V_{\Psi}$ consists of all vertices except $v_*$ where the gauge constraint $A_v$ is violated in the state $\ket{\Psi}$, and $F_{\Psi}$ is the set of all faces except $f_*$ where the flux constraint $B_f$ is violated in the state $\ket{\Psi}$.

The following Lemma says that we can remove a single violation of a gauge constraint.

\begin{lemma} \label{lem:remove star violation}
    Let $\ket{\Psi} \in \caH$ be a unit vector and let $v \in \latticevert \setminus \{v_*\}$ be a vertex. Then there is a unit vector $\ket{\Psi'} \in \caH$ such that $V_{\Psi'} \subseteq V_{\Psi} \setminus \{v\}$ and $F_{\Psi'} \subseteq F_{\Psi}$.
\end{lemma}

\begin{proof}
    Let $\rho$ be a ribbon such that $v(\partial_0 \rho) = v_*$ and $v(\partial_1 \rho) = v$. For any $g \in G$ we define the vector
    $$ \ket{\Psi_g} := A_v T_{\rho}^g \ket{\Psi}. $$
    It follows immediately from the definition that $A_v \ket{\Psi_g} = \ket{\Psi_g}$, so $v \not\in V_{\Psi_g}$. We now show that all vertices that were not in $V_{\Psi} \setminus \{v\}$ are also not in $V_{\Psi_g}$, \ie $V_{\Psi_g} \subseteq V_{\Psi} \setminus \{ v \}$. We have just shown this for $v$ itself, and it is true by definition for $v_*$. It remains to show it for any $v' \not\in V_{\Psi}$ such that $v' \neq v, v_*$. Then $A_{v'} \ket{\Psi} = \ket{\Psi}$ and since $A_{v'}$ commutes with $A_v$ (Eq. \eqref{eq:ABcommproj}) and with $T_{\rho}^g$ (Lemma \ref{lem:A,B comm T,L except endpt}), we find $A_{v'} \ket{\Psi_g} = \ket{\Psi_g}$. \ie $v' \not\in V_{\Psi_g}$.

    Similarly, if $f \not\in F_{\Psi}$, \ie $B_f \ket{\Psi} = \ket{\Psi}$, then also $B_f \ket{\Psi_g} = \ket{\Psi_g}$ because $B_f$ commutes with $A_v$ (Eq. \eqref{eq:ABcommproj}) and with $T_{\rho}^g$ (Lemma \ref{lem:A,B comm L,T always}). This shows that $F_{\Psi_{g}} \subseteq F_{\Psi}$.

    It remains to show that at least one of the $\ket{\Psi_{g}}$ is non-zero.
    
    Using the first identity of Lemma \ref{lem:sweep identities} we find
    $$ \ket{\Psi} = \abs{G} \, \sum_{g \in G} T_{\rho}^g A_v T_{\rho}^g \ket{\Psi} = \abs{G} \sum_{g \in G} T_{\rho}^g \ket{\Psi_g}. $$
    Since $\ket{\Psi}$ is a unit vector, there must be at least one $g \in G$ such that $\ket{\Psi_g} \neq 0$. Let $g \in G$ be such that $\ket{\Psi_g} \neq 0$. We can normalise this vector, obtaining a unit vector
    $$ \ket{\Psi'} = \frac{\ket{\Psi_g}}{\norm{ \ket{\Psi_g} }}. $$
    The property that $V_{\Psi'} \subseteq V_{\Psi} \setminus \{v\}$ and $F_{\Psi'} \subseteq F_{\Psi}$ follows immediately from the fact that $\ket{\Psi_g}$ satisfies this property, as shown above. This proves the Lemma.
\end{proof}

Similarly, we can remove a single violation of a flux constraint. The proof is essentially identical to the proof of the previous Lemma and is therefore omitted.

\begin{lemma} \label{lem:remove face violation}
    Let $\ket{\Psi} \in \caH$ be a unit vector and let $f \in \latticeface \setminus \{f_*\}$ be a face. Then there is a unit vector $\ket{\Psi'} \in \caH$ such that $V_{\Psi'} \subseteq V_{\Psi}$ and $F_{\Psi'} \subseteq F_{\Psi} \setminus \{ f \}$.
\end{lemma}

\begin{proposition}
\label{prop:sweeping}
    If $\omega$ is a pure state on $\cstar$ and there is an $n$ such that $\omega(A_v) = \omega(B_f) = 1$ for all $A_v, B_f \in \cstar[B_{n}^c]$, then for any site $s_*$ there is a pure state $\psi \in \overline{\S}_{s_*}$ that is unitarily equivalent to $\omega$.
\end{proposition}

\begin{proof}
    Let $s_* = (v_*, f_*)$ and let $(\pi, \caH, \ket{\Omega})$ be the GNS triple for $\omega$. From Lemma \ref{lem:omega_0(O) = 1} we find that
    $$ A_v \ket{\Omega} = B_f \ket{\Omega} = \ket{\Omega} $$
    for all $A_v, B_f \in \cstar[B_{n}^c]$. In particular, the sets
    $$ V_{\Omega} = \{ v \in \latticevert \, : \, A_v \ket{\Omega} \neq \ket{\Omega} \} \setminus \{ s_* \} $$
    and
    $$ F_{\Omega} = \{ f \in \latticeface \,: \, B_g \ket{\Omega} \neq \ket{\Omega} \} \setminus \{ f_* \} $$
    are finite.

    We can therefore apply Lemmas \ref{lem:remove star violation} and $\ref{lem:remove face violation}$ a finite number of times to obtain a unit vector $\ket{\Psi} \in \caH$ for which $V_{\Psi} = \emptyset$ and $F_{\Psi} = \emptyset$. In other words, the vector $\ket{\Psi}$ satisfies
    $$  A_v \ket{\Psi} = B_f \ket{\Psi} = \ket{\Psi} $$
    for all $v \neq v_*$ and all $f \neq f_*$.

    The vector $\ket{\Psi}$ corresponds to a pure state $\psi$ given by
    $$\psi(O) = \langle \Psi, O \, \Psi \rangle $$
    for all $O \in \cstar$.

    Since $\omega$ and $\psi$ are vector states in the same representation $\pi$ of $\cstar$, these states are unitarily equivalent. Moreover,
    $$ \psi(A_v) = \langle \Psi, A_v \, \Psi \rangle = 1, \quad \text{and} \quad \psi(B_f) = \langle \Psi, B_f \, \Psi \rangle = 1 $$
    for all $v \neq v_*$ and all $f \neq f_*$, so $\psi \in \overline{\S}_{s_*}$ (Definition \ref{def:state spaces}) as required.
\end{proof}

\subsection{Decomposition of states in $\overline{\mathcal{S}}_{\site}$}
\label{sec:classify}

Recall the Definition \ref{def:state spaces} of the set of states $\overline{\S}_{\site}$. We now prove that any state in $\overline{\S}_{\site}$ decomposes into states belonging to the anyon representations $\pi^{RC}$. Let us fix a state $\omega \in \overline{\S}_{\site}$.

\begin{definition} \label{def:omegaRC components defined}
    For each $RC$, define a positive linear functional $\tomegaRC[]$ by
    $$\tomegaRC[](O) := \omega(\qds O \qds)$$
    for all $O \in \cstar$, and a non-negative number
    $$\lamRC := \omega( D^{RC}_{s_0} ) = \tomegaRC[](\mathds{1}) \geq 0.$$
\end{definition}

\begin{lemma} \label{lem:tomega vanishes iff it vanishes on 1}
    We have $\tomegaRC[] \equiv 0$ if and only if $\lamRC = 0$
\end{lemma}

\begin{proof}
    If $\tomegaRC[](\mathds{1}) = 0$ then $\omega(\qds) = 0$ so by Cauchy-Schwarz
    $$\abs{\tomegaRC[](O)}^2 = \abs{\omega( \qds O \qds )}^2 \leq \omega(\qds) \omega(\qds O^* O \qds ) = 0$$
    for any $O \in \cstar$. i.e. $\tomegaRC[] \equiv 0$ if and only if $\tomegaRC[](\mathds{1}) = 0$. But $\tomegaRC[](\mathds{1}) = \omega( D_{s_0}^{RC} ) = \lamRC$, which yields the claim.
\end{proof}

\begin{lemma} \label{lem:omegaRC belongs to S^RC}
    If $\lamRC >0$, we define a linear functional $\omegaRC[]$ by
    $$\omegaRC[](O) := \lamRC^{-1} \,\, \tomegaRC[](O) \qquad \qquad \forall \,\, O \in \cstar.$$
    Then $\omegaRC[](O)$ is a state in $\S_{s_0}^{RC}$.
\end{lemma}

\begin{proof}
    We first show that $\omegaRC[]$ is a state. The linear functional $\tomegaRC[]$ is positive by construction and
    $$\omegaRC[](\mathds{1}) = \lamRC^{-1} \, \tomegaRC[](\mathds{1}) = 1,$$
    so $\omegaRC[]$ is normalized. We conclude that $\omegaRC[]$ is indeed a state.

    Let us now check that $\omega^{RC} \in \S_{s_0}^{RC}$. First we note that
    $$\omegaRC[](D_{s_0}^{RC}) = \lamRC^{-1} \, \tomegaRC[](D_{s_0}^{RC}) = 1$$
    because $D_{s_0}^{RC}$ is an orthogonal projector. Furthermore, since the projectors $A_v$ commute with $D_{s_0}^{RC}$ for all $v \in \dvregion$ (Lemma \ref{lem:qdcommute}), we have
    $$\omegaRC[](A_v) = \lamRC^{-1} \, \tomegaRC[](A_v) = \lamRC^{-1} \, \omega\big( D_{s_0}^{RC} A_v D_{s_0}^{RC} \big) = \lamRC^{-1} \, \omega( D_{s_0}^{RC} A_v ) = \lamRC^{-1} \omega(D_{s_0}^{RC}) = 1$$
    where we used $\omega(A_v) = 1$ and Lemma \ref{lem:absorption of satisfied projectors}. In the same way we can show that $\omegaRC[](B_f) = 1$ for all $f \in \dfregion$. We conclude that $\omegaRC[] \in \S_{s_0}^{RC}$, as required.
\end{proof}

\begin{lemma} \label{lem:decomposition of omega}
    For any state $\omega \in \overline{\mathcal{S}}_{\site}$, we have 
    $$\omega = \quad \sum_{\mathclap{RC \, : \, \lamRC > 0}} \lamRC \, \omegaRC[]$$
    where the states $\omegaRC[]$ are those defined in Lemma \ref{lem:omegaRC belongs to S^RC}.
\end{lemma}

\begin{proof}
    Using the decomposition $\mathds{1} = \sum_{RC} D_{s_0}^{RC}$ (Lemma \ref{lem:DRCprops}) and Lemma \ref{lem:projector Lemma}, we find
    $$
        \omega(O) = \sum_{RC} \omega \big( \, O \, D_{\site}^{RC}  \big) = \sum_{RC} \omega \big( D_{s_0}^{RC} \, O \, D_{s_0}^{RC} \big) = \sum_{RC \, : \, \lamRC > 0} \lamRC \, \omegaRC[](O)
    $$
    where in the last step we used the Definition of $\omegaRC[]$ and Lemma \ref{lem:tomega vanishes iff it vanishes on 1}.
\end{proof}

It remains to show that the states $\omegaRC[]$ belong to $\pi^{RC}$.

\begin{lemma} \label{lem:omega in SRC belongs to piRC}
    Any state $\omega \in \S_{\site}^{RC}$ belongs to $\pi^{RC}$
\end{lemma}
\begin{proof}
    The restriction $\omega_{n}$ of $\omega$ to $\cstar[{\eregion}]$ satisfies
    $$ 1 = \omega_{n}(A_v) = \omega_{n}(B_f) = \omega_{n}(D_{\site}^{RC}) $$ 
    for all $v \in \dregion$ and all $f \in \fregion$. Let $\omega_{n} = \sum_{\kappa} \lambda_{\kappa} \omega_{n}^{(\kappa)}$ be the convex decomposition of $\omega_{n}$ into its pure components $\omega_{n}^{(\kappa)}$, and let $\ket{\Omega_{n}^{(\kappa)}} \in \caH_{n}$ be unit vectors such that
    $$ \omega_{n}^{(\kappa)}(O) = \langle \Omega_{n}^{(\kappa)}, O \, \Omega_{n}^{(\kappa)} \rangle $$
    for all $O \in \cstar[{\eregion}]$. Then Lemma \ref{lem:pure components projector Lemma} yields
    $$ \ket{\Omega_{n}^{(\kappa)}} = A_v \ket{\Omega_{n}^{(\kappa)}} = B_f \ket{\Omega_{n}^{(\kappa)}} = D_{\site}^{RC} \ket{\Omega_{n}^{(\kappa)}} $$
    for all $v \in \dvregion$ and all $f \in \dfregion$. By Definition \ref{def:local constraints} of the space $\VRC$ we then have $\ket{\Omega_{n}^{(\kappa)}} \in \VRC$ for all $\kappa$. By Proposition \ref{prop:local spaces spanned by etas} it follows that
    $$ \ket{\Omega_{n}^{(\kappa)}} = \sum_{u, v} \, c^{(\kappa)}_{uv} \, \ket{ \eta_{n}^{RC;uv} } $$
    for some coefficients $c_{uv}^{(\kappa)} \in \C$. It follows that
    $$ \omega_{n}(O) = \Tr \lbrace \rho O \rbrace $$
    for any $O \in \cstar[\eregion]$, where $\rho$ is the density matrix on $\caH_n$ given by
    $$\rho = \sum_{u_1 v_1 ; u_2 v_2} \rho_{u_1 v_1 ; u_2 v_2} {\ket{\qdstpure[RC;u_1 v_1]}} \bra{\qdstpure[RC;u_2 v_2]}$$
    with $\rho_{u_1 v_1 ; u_2 v_2} = \sum_{\kappa} \, \lambda_{\kappa} \, c_{u_1 v_1}^{(\kappa)} ( c_{u_2 v_2}^{(\kappa)} )^*$.

    Let $O \in \cstar[{\eregion[n-1]}]$. From Lemma \ref{lem:purerest}, we have that
    $$\inner{\qdstpure[RC;u_1 v_1]}{O \qdstpure[RC;u_2 v_2]} = \delta_{v_1, v_2} \inner{\qdstpure[RC;u_1 v_1]}{O \qdstpure[RC;u_2 v_1]}$$
    is independent of $v_1$, so for any choice $v_0$ we have
    $$ \Tr \lbrace \rho O \rbrace = \sum_{u_1 u_2} \, \left( \sum_v  \rho_{u_1 v;u_2 v}\right) \langle \eta_{n}^{RC;u_1 v_0}, O \, \eta_{n}^{RC;u_2 v_0} \rangle.$$
    Now the numbers
    $$\tilde \rho_{u_1 u_2} := \sum_v \rho_{u_1 v;u_2 v} $$
    are the components of a density matrix $\tilde \rho$ which is a partial trace of $\rho_{u_1 v_1;u_2 v_2}$ over the boundary labels $v_1, v_2$. It follows that there is a basis in which the density matrix $\tilde \rho$ is diagonal, i.e. there is a unitary matrix $U$ with components $U_{uu'} \in \mathbb{C}$ such that
    $$\sum_{u_1 u_2}  (U^*)_{u'_1 u_1} \, \tilde \rho_{u_1 u_2}  \, U_{u_2 u_2'}  = \mu_{u_1'}^{(n)} \delta_{u_1' u_2'}$$
    for non-negative numbers $\mu^{(n)}_{u'}$ that sum to one. 

    We find
    $$\omega_n(O) = \Tr \lbrace \rho O \rbrace = \sum_{u'}  \mu^{(n)}_{u'} \, \tilde \eta_{n}^{RC;u' v_0}(O) $$
    where the $\tilde \eta_{n}^{RC;u' v_0}$ are pure states given by
    $$ \tilde \eta_{n}^{RC;u' v_0}(O) = \langle \tilde \eta_{n}^{RC;u' v_0}, O \, \tilde \eta_{n}^{RC;u'v_0} \rangle $$
    with
    $$ | \tilde \eta_{n}^{RC;u' v_0} \rangle :=  \sum_{u} (U^*)_{u' u} \, | \eta_{n}^{RC;u v_0} \rangle.$$

    Using Lemma \ref{lem:restriction yields eta^RCu}, Definition \ref{def:uniform string net superpositions}, and Lemma \ref{lem:aconverter} we find for any $u'' \in I_{RC}$ that
    $$
    \tilde \eta_{n}^{RC;u' v_0} \big( \sum_{w_1, w_2} D_{s_0}^{RC;w_1} \, (U^*)_{u'' \, w_1} \,  U_{w_2 \, u''} \, D_{s_0}^{RC;w_2} \big) = \delta_{u', u''}.
    $$
    Since the $D_{s_0}^{RC;w}$ are supported in $\cstar[{\eregion[n-1]}]$ for all $n \geq 3$ we therefore find that $\mu_{u''}^{(n)} = \mu_{u''}^{(m)}$ for all $u'' \in I_{RC}$ and all $n, m \geq 3$, \ie the numbers $\mu_u^{(n)}$ do not depend on $n$ if $n \geq 3$. Let us write $\mu_u = \mu_u^{(n)}$ for all $n \geq 3$.

    For any $u \in I_{RC}$, let $A^u := A_{\site}^{RC;u (1, 1)}$ be the label changer operator from Definition \ref{def:label changers}, and define vectors
    $$\ket{\Omega_{\site}^{RC;u}} := \pi^{RC} \big( A^u \big) \ket{\Omega_{\site}^{RC;(1, 1)}}.$$
    These are unit vectors representing the states $\omega_{\site}^{RC;u}$ in the representation $\pi^{RC}$. Define pure states $\tilde \omega^{RC;u'}$ on $\cstar$ by
    $$\tilde \omega^{RC;u'}(O) := \langle \widetilde \Omega^{RC;u'}, \pi^{RC}(O) \, \widetilde \Omega^{RC;u'} \rangle$$
    corresponding to the GNS vectors
    $$|\widetilde \Omega^{RC;u'} \rangle := \sum_{u} (U^*)_{u' u} \, | \Omega^{RC;u} \rangle \in \hilb^{RC}.$$
    
    Then we find
    $$\tilde \omega^{RC;u'}(O) = \omega_{\site}^{RC;(1, 1)} \big( \sum_{u_1, u_2}  U_{u_1 u'} (U^*)_{u' u_2} \, (A^{u_1})^* O A^{u_2} \big)$$
    for any $O \in \cstar$. Using Lemma \ref{lem:restriction yields eta^RCu}, Definition \ref{def:uniform string net superpositions} Lemma \ref{lem:aconverter}, and the fact that the $A^u$'s are supported near $\site$, for any $O$ supported on $\eregion[n-1]$ we get:
    $$ \tilde \omega^{RC;u'}(O) = \eta_{n}^{RC;u' v_0} \big( \sum_{u_1, u_2}  U_{u_1 u'} (U^*)_{u' u_2} \, (A^{u_1})^* O A^{u_2} \big) = \tilde \eta_{n}^{RC;u'v_0}(O). $$
    
    We conclude that
    $$\omega(O) = \Tr \lbrace \rho O \rbrace = \sum_{u'} \, \mu_{u'} \, \tilde \omega^{RC;u'}(O)$$
    for all $n \geq 3$ and all $O \in \cstar[{\eregion[n-1]}]$. It follows that we have an equality of states
    $$\omega = \sum_{u'} \mu_{u'} \tilde \omega^{RC;u'},$$
    which expresses $\omega$ as a finite mixture of pure states belonging to $\pi^{RC}$. It follows that $\omega$ also belongs to $\pi^{RC}$, as was to be shown.
\end{proof}

The results obtained above combine to prove the following proposition.
\begin{proposition}
\label{prop:omegadecomp} 
    Let $\omega \in \overline{\S}_{\site}$ for some site $\site$. Then $\omega$ has a convex decomposition
    $$\omega = \sum_{RC} \, \lambda_{RC} \,\, \omegaRC[]$$
    into states $\omegaRC[]$ that belong to the representation $\pi^{RC}$. In particular, if $\omega$ is pure then $\omega$ belongs to $\pi^{RC}$ for some definite $RC$.
\end{proposition}

\begin{proof}
    Lemma \ref{lem:decomposition of omega} provides the decomposition
    $$\omega = \sum_{RC : \lambda_{RC} > 0} \lambda_{RC} \, \omega^{RC}$$
    with positive $\lambda_{RC}$ and states $\omega^{RC}$ defined in Definition \ref{def:omegaRC components defined} and Lemma \ref{lem:omegaRC belongs to S^RC}. Moreover, Lemma \ref{lem:omegaRC belongs to S^RC} states that $\omega^{RC} \in \S_{\site}^{RC}$ for each $RC$ with $\lambda_{RC} > 0$. It then follows from Lemma \ref{lem:omega in SRC belongs to piRC} that the states $\omega^{RC}$ belong to $\pi^{RC}$. This concludes the proof.
\end{proof}

\subsection{Classification of anyon sectors}
We finally put together the results obtained above to prove our main result, Theorem \ref{thm:main theorem}, which we restate here for convenience.

\begin{theorem} \label{thm:classification theorem}
    For each irreducible representation $RC$ of $\caD(G)$ there is an anyon representation $\pi^{RC}$. The representations $\{ \pi^{RC} \}_{RC}$ are pairwise disjoint, and any anyon representation is unitarily equivalent to one of them. 
\end{theorem}

\begin{proof}
    The existence of the pairwise disjoint anyon representations follows immediately from Lemma \ref{lem:piRCaredisjoint} and Proposition \ref{prop:the pi^RC are anyon representations}.

    Let $\pi$ be an anyon representation. By Proposition \ref{prop:equivtofinite}, there is $n \in \N$ such that the anyon representation $\pi$ contains a pure state $\omega$ that satisfies $\omega(A_v) = \omega(B_f) = 1$ for all $A_v, B_f \in \cstar[B_{n}^c]$. For any site $\site$, Proposition \ref{prop:sweeping} then gives us a pure state $\psi \in \overline{\S}_{\site}$ that belongs to $\pi$. Proposition \ref{prop:omegadecomp} shows that this state belongs to an anyon representation $\pi^{RC}$ for some definite $RC$. Since the irreducible representations $\pi$ and $\pi^{RC}$ contain the same pure state, they are unitarily equivalent. This proves the Theorem.
\end{proof}

%% file: conclusions.tex
\section{Discussion and outlook}
\label{sec:conclusion}

In this paper we have fully classified the anyon sectors of Kitaev's quantum double model for an arbitrary finite gauge group $G$ in the infinite volume setting. The proof of the classification contained several ingredients. First, we constructed for each irreducible representation $RC$ of the quantum double algebra a set of pure states $\{ \omega_{s}^{RC;u} \}_{s, u}$ that are all unitarily equivalent to each other. We then showed that the corresponding irreducible GNS representations $\{\pi^{RC}\}_{RC}$ are a collection of disjoint anyon representations. The proof that these representations are anyon representations crucially relied on their identification with `amplimorphism representations'. To show completeness, we proved that any anyon representation of the quantum double models contains one of the states $\omega_s^{RC;u}$, so that any anyon representation is unitarily equivalent to one of the $\pi^{RC}$.

This result is a first step towards integrating the non-abelian quantum double models in the mathematical framework for topological order developed in \cite{Naaijkens2010-aq, Naaijkens2012-fh, Cha2018-ke, Cha2020-rz, Ogata2022-wp}. To complete the story, we would like to work out the fusion rules and braiding statistics of the anyon types we identified. The amplimorphisms of Section \ref{sec:anyon representations} promise to be an ideal tool to carry out this task, see \cite{Naaijkens2015-xj, szlachanyi1993quantum}. We leave this analysis to an upcoming work. A crucial assumption in the general theory of topological order for infinite quantum spin systems is (approximate) Haag duality for cones. This property has been established for abelian quantum double models \cite{Fiedler2015-na, Naaijkens2012-fh}, but not yet for the non-abelian case.

The proofs of purity of the \ffgs{} and of the states $\omega_{s}^{RC;u}$ use the string-net condensation picture \cite{levin2005string}. The same methods can be used to show purity of ground states of other commuting projector models (for example \cite{Bols2023-qg} shows the purity of the double semion \ffgs{} and \cite{Vadnerkar2023-mm} shows purity of the \ffgs{} of the 3d Toric Code). In principle, the same techniques can be used to show that Levin-Wen models \cite{levin2005string} and quantum double models based on (weak) Hopf algebras have a unique frustration free ground state in infinite volume (this has already been done for Levin-Wen models using different techniques in \cite{jones2023local}).

%% file: ribbonprops.tex
\section{Ribbon operators, Wigner projections, and amplimorphisms}
\label{app:ribbonprops}

\subsection{Basic properties of ribbon operators, gauge transformations, and flux projectors}

Recall from Section \ref{subsec:preliminary notions} the definitions of ribbons and ribbon operators, as well as the definitions of the gauge transformations $A_s^h$ and flux projectors $B_s^g$, all originally due to Kitaev \cite{Kitaev2003-qr}.

We have the following basic properties of ribbon operators, which can be easily verified. See also Appendix B of \cite{Bombin2007-uw}.

\begin{align}
\label{eq:F elementary props}
    F_\rho^{h,g} F_\rho^{h',g'} = \begin{cases} \delta_{g,g'} F^{h'h,g}_\rho \quad &\text{if  } \, \rho \, \text{ is positive} \\ \delta_{g,g'} F^{hh',g}_\rho \quad &\text{if } \, \rho \, \text{ is negative}  \end{cases} && (F^{h,g}_\rho )^* = F^{\dash{h},g}_\rho.
\end{align}

Let $\rho$ be such that its end sites $s_i = \partial_i \rho = (v_i, f_i)$ satisfy $v_0 \neq v_1$ and $f_0 \neq f_1$. Then if $\rho$ is positive we have
\begin{equation}
    \begin{aligned}
\label{eq:[A,F] and [B,F]}    
    A_{s_0}^k F^{h,g}_\rho = F^{kh\dash{k}, kg}_\rho A_{s_0}^k && \qquad A_{s_1}^k F^{h,g}_\rho = F^{h, g\dash{k}}_\rho A_{s_1}^k \\
    B_{s_0}^k F^{h,g}_\rho = F^{h,g}_\rho B_{s_0}^{hk} && \qquad  B_{s_1}^k F^{h,g}_\rho = F^{h,g}_\rho B_{s_1}^{k\dash{g}\dash{h}g}
    \end{aligned}
\end{equation}
and if $\rho$ is negative
\begin{equation}
    \begin{aligned}
\label{eq:[A,F] and [B,F] negative case}    
    A_{s_0}^k F^{h,g}_\rho = F^{kh\dash{k}, kg}_\rho A_{s_0}^k && \qquad A_{s_1}^k F^{h,g}_\rho = F^{h, g\dash{k}}_\rho A_{s_1}^k \\
    B_{s_0}^k F^{h,g}_\rho = F^{h,g}_\rho B_{s_0}^{kh} && \qquad  B_{s_1}^k F^{h,g}_\rho = F^{h,g}_\rho B_{s_1}^{\dash{g}\dash{h}gk}.
    \end{aligned}
\end{equation}

Importantly, the ribbon is invisible to the gauge transformations and flux projectors away from its endpoints. That is, if $\rho$ is a finite ribbon with $\partial_0 \rho = s_0$ and $\partial_1 \rho = s_1$, and $s = (v, f)$ is such that $v \neq v(s_0), v(s_1)$ while $s' = (v', f')$ is such that $f' \neq f(s_0), f(s_1)$, then
\begin{align}
\label{eq:A,B commute with F except at endpts}
    [F^{h,g}_\rho, A_{s}^k] = 0 = [F^{h,g}_\rho, B_{s'}^l]
\end{align}
for all $g, h, k, l \in G$.

The following properties of the gauge transformations and flux projectors follow immediately from the properties of the ribbon operators listed above.

\begin{align}
\label{eq:Aelementary}
    &A^1_s = \mathds{1}& (A_s^h)^*= A_s^{\dash{h}} && A_s^h A_s^{h'} = A_s^{hh'}\\
    \label{eq:Belementary}
    &\sum_{g \in G} B_s^g = 1 &(B_s^g)^* = B_s^g  && B_s^g B_s^{g'} = \delta_{g,g'} B_s^g\\
    \label{eq:ABonsamesite}
    && A_s^h B_s^g = B^{hg\dash{h}}_s A_s^h &&
\end{align}

If $s \neq s'$ then for any $h,h',g,g' \in G$,
\begin{align}
\label{eq:ABcommuteondiffsites}
    [A_s^h, B_{s'}^g] = [A_s^h, A_s^{h'}] = [B_s^g, B_{s'}^{g'}] = 0
\end{align}

Recall from Section \ref{subsec:preliminary notions} that the projectors $A_v := \frac{1}{|G|} \sum_{h \in G} A_s^h$ and $B_f := B_s^1$ where $v$ is the vertex of $s$ and $f$ is the face of $f$ are well defined. For any vertices $v, v'$ and any faces $f, f'$, we have (\cite{Kitaev2003-qr})
\begin{equation} \label{eq:ABcommproj}
    [A_v, B_f] = [A_v, A_{v'}] = [B_f, B_{f'}] = 0.
\end{equation}

\subsection{Decomposition of $F^{h,g}_\rho$ into $L^h_\rho, T^g_\rho$}

\subsubsection{Basic properties}

Recall from Section \ref{subsec:preliminary notions} the definitions $T^g_\rho := F^{1,g}_\rho$ and $L^h_\rho := \sum_{g \in G} F^{h,g}_\rho$.

\begin{lemma}
    \label{eq:F breaks into LT}
    $F^{h,g}_\rho = L^h_\rho T^g_\rho =  T^g_\rho L^h_\rho$
\end{lemma}

\begin{proof}
    Using Eq. \eqref{eq:F elementary props} we have
    \begin{align*}
        L^h_\rho T^g_\rho = \sum_{g'} F^{h,g'}_\rho F^{1,g}_\rho = F^{h,g}_\rho =  F^{1,g}_\rho \sum_{g'} F^{h,g'}_\rho = T^g_\rho L^h_\rho
    \end{align*}
\end{proof}

\begin{lemma} \label{lem:flux change}
    Let $\rho$ be a finite ribbon such that $v(\start) \neq v(\en)$. Then if $s_0,s_1$ are sites such that $v(s_0) = \start$ and $ v(s_1) = \en$, we have
    $$ A^{h}_{s_0} \,  T_{\rho}^g = T_{\rho}^{hg} \, A_{s_0}^h \quad \text{and} \quad A_{s_1}^{h} \, T_{\rho}^g = T_{\rho}^{g \bar h} \, A_{s_1}^h $$
    while for sites $s$ such that $v(s) \neq v(s_0), v(s_1)$, we have
    $$ [A_s^h, T_{\rho}^g] = 0$$
    for all $g, h \in G$. Moreover,
    $$ [B_s^{k}, T_{\rho}^g] = 0 $$
    for all $k, g \in G$ and any site $s$.
\end{lemma}

\begin{proof}
    This follows immediately from $T_{\rho}^g = F_{\rho}^{1, g}$ and Eqs. \eqref{eq:[A,F] and [B,F]}, \eqref{eq:[A,F] and [B,F] negative case} and \eqref{eq:A,B commute with F except at endpts}.
\end{proof}

\begin{lemma}
\label{lem:A,B comm T,L except endpt}
    Let $\rho$ be a finite ribbon. For all $v \neq v(\partial_i \rho)$ and $f \neq f(\partial_i \rho)$, we have
    $$[A_v,T_\rho^g] = 0 = [B_f, L^h_\rho].$$
\end{lemma}
\begin{proof}
    This is a trivial consequence of Eq. \eqref{eq:A,B commute with F except at endpts}.
\end{proof}

\begin{lemma}
    \label{lem:A,B comm L,T always}
    Let $\rho$ be a finite ribbon. For all $v$ such that $v\neq v(\start)$ and all $f$ such that $f \neq f(\start)$ we have
    $$[A_v,L_\rho^h] = 0 = [B_f, T^g_\rho].$$
\end{lemma}

\begin{proof}
    If we have $v \neq v(\partial_i \rho)$ or $f \neq f(\partial_i \rho)$ for $i = 0, 1$ then the claim follows immediately from Eq. \eqref{eq:A,B commute with F except at endpts}. Now let $v = v(\partial_1 \rho)$, $f = f(\en)$, and let $s,s'$ be sites such that $v(s) = v$ and $f(s') = f$. We then have,
    \begin{align*}
        A_v L_\rho^h &= \sum_k A_s^k L^{h}_\rho = \sum_{g,k} A_s^k F^{h,g}_\rho = \sum_{g,k} F^{h, g\dash{k}}_\rho A_{s}^k = \sum_k L^h_\rho A_s^k = L^h_\rho A_v\\
        B_f T^g_\rho &= B_{s'}^1 F^{1,g}_\rho = F^{1,g}_\rho B_{s'}^1 = T^g_\rho B_f
    \end{align*}
    Which proves the claim.
\end{proof}

\subsubsection{Alternating decomposition of $L_\rho^h$}

In this section we express the operators $L_\rho^h$ in terms of the decomposition of $\rho$ into its alternating direct and dual sub-ribbons. This result will be useful in Section \ref{sec:action of T and L on string nets}.

\begin{lemma} \label{lem:L with initial direct triangle}
    Let $\rho = \tau \rho'$ be a finite ribbon whose initial triangle $\tau$ is a direct triangle. Then $$L_{\rho}^h = \sum_{k \in G} T_{\tau}^{k} \, L_{\rho'}^{\bar k h k}.$$
\end{lemma}

\begin{proof}
    By definition, $L_{\rho}^h = \sum_g \, F_{\rho}^{h, g}$. Using Eq. \eqref{eq:F inductive def}, this becomes
    $$L_{\rho}^{h} = \sum_g \sum_k F_{\tau}^{h, k} F_{\rho'}^{\bar k h k, g} = \sum_k T_{\tau}^k \, L_{\rho'}^{\bar k h k},$$
    where we used $F_{\tau}^{h, k} = T_{\tau}^k$ because $\tau$ is a direct triangle.
\end{proof}

\begin{lemma} \label{lem:L with initial dual triangle}
    Let $\rho = \tau \rho'$ be a finite ribbon such that its initial triangle $\tau$ is a dual triangle. Then $$L_{\rho}^h = L_{\tau}^{h} \, L_{\rho'}.$$
\end{lemma}

\begin{proof}
    By definition, $L_{\rho}^h = \sum_g \, F_{\rho}^{h, g}$. Using Eq. \eqref{eq:F inductive def}, this becomes
    $$L_{\rho}^{h} = \sum_g \sum_k F_{\tau}^{h, k} F_{\rho'}^{\bar k h k, g} = L_{\tau}^h \, L_{\rho'}^{h},$$
    where we used $F_{\tau}^{h, k} = \delta_{k, 1} L_{\tau}^h$ because $\tau$ is a dual triangle.
\end{proof}

\begin{lemma} \label{lem:unpacking of T for direct ribbon}
    If $\rho = \{\tau_i\}_{i=1}^N$ is a direct ribbon, then
    $$T^g_{\rho} = \sum_{k_1 \cdots k_N = g} \,\,\, \prod_{i = 1}^N \,\,\, T_{\tau_i}^{k_i}.$$
\end{lemma}

\begin{proof}
    By definition, $T_{\rho}^g = F_{\rho}^{1, g}$. Using Eq. \eqref{eq:F inductive def} we find
    $$T_{\rho}^g = F_{\rho}^{1, g} = \sum_{k_1} \, F_{\tau_1}^{1, k_1} F_{\rho \setminus \{\tau_1\}}^{1, \bar k_1 g} = \sum_{k_1} T_{\tau_1}^{k_1}  T_{\rho \setminus \{\tau_1\}}^{\bar k_1 g}.$$
    We can apply this result inductively to find
    $$T_{\rho}^g = \sum_{k_1, \cdots, k_N}  \prod_{i= 1}^N  T_{\tau_i}^{k_i} \, F_{\epsilon}^{1, \bar k_N \cdots \bar k_1 g} = \delta_{k_1 \cdots k_N, g} \, \sum_{k_1, \cdots, k_N} \, \prod_{i=1}^N \, T_{\tau_i}^{k_i}.$$
    This proves the claim.
\end{proof}

\begin{lemma} \label{lem:unpacking of L for dual ribbon}
    If $\rho = \{\tau_i\}_{i=1}^N$ is a dual ribbon, then
    $$L^h_{\rho} = \prod_{i = 1}^N \,\,\, L_{\tau_i}^{h}.$$
\end{lemma}

\begin{proof}
    This follows immediately from a repeated application of Lemma \ref{lem:L with initial dual triangle}.
\end{proof}

Any ribbon decomposes into subribbons that are alternatingly direct and dual.
\begin{definition} \label{def:alternating decomposition}
    Any finite ribbon $\rho$ has a unique decomposition into ribbons $\{I_a, J_a\}_{a = 1, \cdots, n}$ such that the $I_a$ are direct, the $J_a$ are dual and
    $$\rho = I_1 J_1 \cdots I_n J_n.$$
    (possibly, $I_1$ and/or $J_n$) are empty. We call this the alternating decomposition of $\rho$. 
\end{definition}

\begin{lemma} \label{lem:L decomposition}
    Let $\rho$ be a finite ribbon with alternating decomposition $\rho = I_1 J_1 \cdots I_n J_n$. We have
    $$L_{\rho}^h = \sum_{k_1, \cdots, k_n \in G} \, \prod_{i = 1}^n \, T_{I_i}^{k_i}  L_{J_i}^{ \bar K_i h K_i }$$
    where $K_i = k_1 k_2 \cdots k_i$.
\end{lemma}

\begin{proof}
    The first sub ribbon $I_1 = \{ \tau_1, \cdots, \tau_m\}$ consists entirely of direct triangles. A repeated application of Lemma \ref{lem:L with initial direct triangle} yields
    $$L_{\rho}^h = \sum_{l_1, \cdots, l_m} \prod_{i= 1}^m \,T_{\tau_i}^{l_i} \,\, L_{\rho \setminus I_1}^{\bar k_1 h k_1}$$
    where $k_1 = l_1 \cdots l_m$. Using Lemma \ref{lem:unpacking of T for direct ribbon} we can rewrite this as
    $$L_{\rho}^h = \sum_{k_1} \, T_{I_1}^{k_1} \, L_{\rho \setminus I_1}^{\bar k_1 h k_1}.$$
    Let us now write $\rho' = \rho \setminus I_1$ and let $J_1 = \{\sigma_1, \cdots, \sigma_{m'}\}$ be the first sub-ribbon of $\rho'$ that consists entirely of dual triangles. A repeated application of Lemma \ref{lem:L with initial dual triangle} yields
    $$L_{\rho'}^{\bar k_1 h k_1} = \prod_{i = 1}^{m'} \, L_{\sigma_i}^{\bar k_1 h k_1} \, L_{\rho' \setminus J_1}^{\bar k_1 h k_1} = L_{J_1}^{\bar k_1 h k_1} L_{\rho' \setminus J_1}^{\bar k_1 h k_1}$$
    where we used Lemma \ref{lem:unpacking of L for dual ribbon} in the last step.

    Putting the above results together, we obtain
    $$L_{\rho}^h = \sum_{k_1} \, T_{I_1}^{k_1} \, L_{J_1}^{\bar k_1 h k_1} \, L_{\rho \setminus {I_1 J_1}}^{\bar k_1 h k_2}.$$

    Repeating the same argument for the ribbon $\rho \setminus \{ I_1 J_1 \} = I_2 J_2 \cdots I_n J_n$ we get
    $$L_{\rho}^h = \sum_{k_1, k_2} T_{I_1}^{k_1} L_{J_1}^{\bar K_1 h K_1} \, T_{I_2}^{k_2} L_{J_2}^{\bar K_2 h K_2} \,\, L_{I_3 J_3 \cdots I_n J_n}^{\bar K_2 h K_2}.$$

    Repeating the argument $n-2$ more times yields the claim.
\end{proof}

\subsection{Wigner projectors and their decompositions} \label{sec:Wigner projectors}

\subsubsection{Basic tools}

We provide some facts that will be used in calculations involving irreducible representations of $\caD(G)$ throughout the paper.

\begin{lemma}
\label{lem:unique g=qn}
    Let $C \in (G)_{cj}$, then each element $g \in G$ can be written as $g = q n$ with $q \in Q_C$ and $n \in N_C$ in a unique way.
\end{lemma}

\begin{proof}
    We have $g r_C \bar g = q r_C \bar q$ for some $q \in Q_C$. So $\bar q g = n \in N_C$, \ie we have $g = q n$.
    
    As for uniqueness, suppose $q_1 n_1 = q_2 n_2$ with $q_1, q_2 \in Q_C$ and $n_1, n_2 \in N_C$. Then $\bar q_2 q_1 = n_2 \bar n_1 \in N_C$, so $r_C = \bar q_2 q_1 r_C \bar q_1 q_2$ from which it follows that $q_2 r_C \bar q_2 = q_1 r_C \bar q_1$. By construction of the iterator set $Q_C$, this is only possible if $q_1 = q_2$, and therefore also $n_1 = n_2$.
\end{proof}

We will often have to use the Schur orthogonality relations, which we state here for reference. Let $H$ be a finite group and $R_1, R_2 \in (H)_{irr}$ irreducible representations of $H$ with matrix realisations $M_{R_1}$ and $M_{R_2}$ respectively. Then
\begin{equation} 
    \label{eq:Schur}
    \sum_{h \in H} M_{R_1}^{jk}(h) M_{R_2}^{lm}(h)^* = \delta_{R_1,R_2} \delta_{j,l} \delta_{k,m} \frac{|H|}{\dim R_1}.
\end{equation}
If $\chi_R$ is the character of the irreducible representation $R$, then we have,
\begin{equation} 
    \label{eq:Schur2}
    \sum_{R \in (H)_{irr}} \, \chi_R(h_1) \chi_R(h_2)^* = \begin{cases} \abs{ \mathcal{Z}_{h_1} } \,\,\,\, & \text{if} \,\, h_1, h_2 \,\, \text{belong to the same} \,\,  C \in (H)_{cj} \\
    0 \, & \text{otherwise} \end{cases}
\end{equation}
where $\mathcal{Z}_{h_1}$ is the commutant of $h_1$ in $H$.

\subsubsection{Wigner projectors}

Recall the Wigner projectors $D_s^{RC}$ and $D_{s}^{RC;u}$ (Definition \ref{def:Wigner projectors}), and the label changers $A_{s}^{RC;u_2 u_1}$ (Definition \ref{def:label changers}). Note also that $\qd[RC;u] = A_s^{RC;u} B_{s}^{c_i}$ with
\begin{equation*}
    A_s^{RC;u} := \frac{\dimR}{|N_C|} \sum_{m \in N_C} \, R^{jj}(m)^* A_s^{q_i m \bar q_i}.
\end{equation*}

\begin{lemma}
\label{lem:a and B commute with A_v and B_f}
    Let $s_0 = (v_0, f_0)$ be a site. Then $A_\site^{RC;(i,j)}$ and $B_{\site}^{c_i}$ are commuting projectors that also commute with $A_v$ and $B_f$ for all $v \neq v_0$ and all $f \neq f_0$.
\end{lemma}

\begin{proof}
    First we check that the $A_\site^{RC;(i,j)}$ are projectors.
    \begin{align*}
        A_\site^{RC;(i,j)} A_\site^{RC;(i,j)} &= \bigg(\frac{\dimR}{|N_C|} \bigg)^2 \sum_{m, m' \in N_C} R^{jj}(m)^* R^{jj}(m')^*  A_{s_0} ^{q_i m \dash{q}_i} A_{s_0} ^{q_i m' \dash{q}_i}\\
        &= \bigg(\frac{\dimR}{|N_C|} \bigg)^2 \sum_{m, m' \in N_C} R^{jj}(m)^* R^{jj}(m')^* A_{s_0}^{q_i m m' \dash{q}_i}\\
        \intertext{Relabeling $M = m m'$ and using the Schur orthogonality relation Eq. \eqref{eq:Schur} we get}
        &= \bigg(\frac{\dimR}{|N_C|} \bigg) \sum_{M \in N_C} R^{jj}(M)^* \, A_{s_0}^{q_i M \dash{q}_i} = A_\site^{RC;(i,j)}
    \end{align*}
    Showing that $(A_\site^{RC;(i,j)})^* = A_\site^{RC;(i,j)}$ is a straightforward application of Eq. \eqref{eq:Aelementary}.

    $A_\site^{RC;u}$ trivially commutes with $A_v$ for all $v \in \dvregion$ and $B_\site^{c_i}$ trivially commutes with $B_f$ for $f \in \dfregion$ using Eq. \eqref{eq:ABcommuteondiffsites}. Using the same equation, we also have $[A_\site^{RC;u}, B_f] = 0 = [B_\site^{c_i}, A_v]$ for all $f \in \dfregion, v \in \dvregion$.

    It remains to show that $[a_\site^{RC;u}, B_\site^{c_i}] =0$. This follows from Eq. \eqref{eq:ABonsamesite} and the fact that $q_i m \dash{q}_i$ commutes with $c_i$ for all $m \in N_C$. This implies $[A^{q_i m \dash{q}_i}_\site, B^{c_i}_\site] = 0$.
\end{proof}

\begin{lemma}
\label{lem:DRC decomposes into DRCu}
    The $\{  D^{RC;u}_s \}_{u \in I_{RC}}$ are a set of commuting projectors such that $\qd = \sum_u D^{RC;u}_s$. In particular, $\qd$ is a projector.
\end{lemma}
\begin{proof}
    That $D_s^{RC;u}$ is a projector follows from $D_s^{RC;u} = A_s^{RC;u} B_s^{c_i}$ and the fact that $A_s^{RC;u}$ and $B_s^{c_1}$ are commuting projectors (Lemma \ref{lem:a and B commute with A_v and B_f}).
    
    Now to prove commutativity, let $u_1 = (i_1, j_1), u_2 = (i_2, j_2)$. Then,
    \begin{align*}
        D^{RC;u_1}_s D^{RC;u_2}_s &= \bigg(\frac{\dimR}{|N_C|}\bigg)^2 \sum_{m_1, m_2 \in N_C} R^{j_1 j_1}(m_1)^* A_{s} ^{q_{i_1} m_1 \dash{q}_{i_1}} B_{s}^{c_{i_1}} R^{j_2 j_2}(m_2)^* A_{s}^{q_{i_2} m_2 \dash{q}_{i_2}} B_{s}^{c_{i_2}} \\
        \intertext{Now we use Eqs. \eqref{eq:Belementary}, \eqref{eq:ABonsamesite} to get:}
        &= \delta_{i_1, i_2} \bigg(\frac{\dimR}{|N_C|}\bigg)^2 \sum_{m_1, m_2 \in N_C} \, R^{j_1 j_1}(m_1)^* A_{s}^{q_{i_1} m_1 m_2 \dash{q}_{i_1}} R^{j_2 j_2}(m_2)^* B_{s}^{c_{i_1}} \\
        \intertext{relabelling $m = m_1 m_2$ and using the Schur orthogonality relation Eq. \eqref{eq:Schur} this becomes}
        &= \delta_{u_1, u_2} D^{RC;u_1}_s.
    \end{align*}
    Finally, to show they sum up to $\qd$,
    \begin{align*}
        \sum_u D^{RC;u}_s &= \sum_{i,j} \frac{\dimR}{|N_C|} \sum_{m \in N_C} R^{jj}(m)^* A_{s} ^{q_i m \dash{q}_i} B_{s}^{c_i} = \sum_{i} \frac{\dimR}{|N_C|} \sum_{m \in N_C} \chi_R(m)^* A_{s}^{q_i m \dash{q}_i} B_{s}^{c_i}\\
        &= \frac{\dimR}{|N_C|} \sum_{m \in N_C} \chi_R(d)^* \sum_{q_i \in Q_C}  A_{s}^{q_i m \dash{q}_i} B_{s}^{c_i} = \qd.
    \end{align*}
\end{proof}

The projectors $D_s^{RC}$ satisfy the following properties.
\begin{lemma}
\label{lem:DRCprops}
    The $\qd$ are orthogonal projectors and
    $$\qd[RC] \qd[R'C'] = \delta_{RC, R'C'} \qd[RC], \quad \sum_{RC} \, D_s^{RC} = \I.$$
\end{lemma}

\begin{proof}
    This follows immediately from Proposition 21 and Eq. (B77) of \cite{Bombin2007-uw}.
\end{proof}

\begin{lemma} \label{lem:qdcommute}
    Let $s_0 = (v_0, f_0)$ be a site, then $D^{RC}_\site$ commutes with $A_{v}, B_{f}$ for all $v \neq v_0$ and all $f \neq f_0$.
\end{lemma}
\begin{proof}
    Noting that $D_{\site}^{RC} = \sum_u \, D_{\site}^{RC;u}$ (Lemma \ref{lem:DRC decomposes into DRCu}) and $D_{\site}^{RC;(i, j)} = A_{\site}^{RC;(i, j)} B_{\site}^{c_i}$, the claim follows immediately from Lemma \ref{lem:a and B commute with A_v and B_f}.
\end{proof}

\subsection{Representation basis for ribbon operators} \label{sec:RC basis for ribbon operators}

Recall Definition \ref{def:RC ribbons} of the ribbon operators $\frcuv$. These ribbon operators satisfy the following basic properties.
\begin{lemma} (\cite[Lemma~4.11]{Naaijkens2015-xj})
    \label{lem:decomposition of F}
    If $\rho = \rho_1 \rho_2$ then
    $$F^{RC;uw}_{\rho} = \bigg( \frac{\abs{N_C}}{\dimR} \bigg) \sum_{v} F^{RC;uv}_{\rho_1} F^{RC;vw}_{\rho_2}.$$
\end{lemma}

\begin{lemma}(\cite[Eq. (5.1)]{Naaijkens2015-xj})
\label{lem:FdaggerF identity}
    We have
    $$\sum_v (F^{RC;u_1 v}_\rho)^* F^{RC;u_2 v}_\rho = \delta_{u_1,u_2} \bigg(\frac{\dimR}{|N_C|}\bigg)^2 \mathds{1} \quad \text{and} \quad \sum_v F^{RC;u_1 v}_\rho (F^{RC;u_2 v}_\rho)^* = \delta_{u_1,u_2} \bigg(\frac{\dimR}{|N_C|}\bigg)^2 \mathds{1}$$
\end{lemma}

\subsection{Detectors of topological charge} \label{sec:topological charge detectors}

Recall Definition \ref{def:charge detectors} of the `charge detectors' $\knRC[\sigma]$. These satisfy the following basic properties.

\begin{lemma}(\cite[Eq. (B77)]{Bombin2007-uw}) \label{lem:basic properties of K}
    The $\knRC[\sigma]$ are orthogonal projectors and
    \begin{equation*}
        \knRC[\sigma] K^{R'C'}_\sigma = \delta_{RC, R'C'} \knRC[\sigma], \quad 
        \sum_{RC} \knRC[\sigma] = \I.
    \end{equation*}
\end{lemma}

\subsection{Actions on the \ffgs} \label{sec:actions on FFGS}

In this subsection we consider several ways in which the ribbon operators $F_{\rho}^{RC;uv}$, the projectors $D_{s}^{RC;u}$ (Definition \ref{def:Wigner projectors}), and the label changers $A_{s}^{RC;u_2 u_1}$ (Definition \ref{def:label changers}) act on the frustration free ground state. We will work in the GNS representation $(\pi_0, \caH_0, \ket{\Omega_0})$ of the frustration free ground state $\omega_0$, and will in the remainder of this section drop $\pi_0$ from the notation. \ie for any $O \in \cstar$ we simply write $O$ instead of $\pi_0(O)$.

Let us first note the following.
\begin{lemma} \label{lem:omega_0(O) = 1}
    If $O$ is a unitary or a projector and $\omega_0(O) = 1$, then $O \ket{\Omega_0} = \ket{\Omega_0}$.
\end{lemma}

\begin{proof}
    We have
    $$ 1 = \omega_0(O) = \langle \Omega_0, O \, \Omega_0 \rangle.$$
    since $\norm{O} \leq 1$ this is only possible if $O \ket{\Omega_0} = \ket{\Omega_0}$.
\end{proof}

\begin{lemma} \label{lem:change ribbon operator label}
    Let $\rho$ be a finite ribbon with $\partial_0 \rho = \site$. Then
    \begin{equation*}
		A^{RC;u_2 u_1}_{\site} F^{RC;u_1 v}_{\rho} | \Omega_0 \rangle = F^{RC;u_2 v}_{\rho} | \Omega_0 \rangle.
  \end{equation*}
\end{lemma}

\begin{proof}
    Let $u_1 = (i_1, j_2)$, $u_2 = (i_2, j_2)$ and $v = (i', j')$, then
    \begin{align*}
		A^{RC;u_2 u_1}_{\site} F^{RC;u_1 v} &= \left( \frac{\dimR}{|N_C|} \right)^2 \sum_{m, n} \, R^{j_2 j_1}(m)^* R^{j_1 j'}(n)^* A_{s_0}^{q_{i_2} m \bar q_{i_1}} F_{\rho}^{\bar c_{i_2}, q_{i_2} n \bar q_{i'}} \\
		&= \left( \frac{\dimR}{|N_C|}  \right)^2 \sum_{m, n} \,  R^{j_2 j_1}(m)^* R^{j_1 j'}(n)^* \, F_{\rho}^{\bar c_{i_1}, q_{i_1} mn \bar q_{i'}} \, A_{s_0}^{q_{i_2} m \bar q_{i_1}}.\\
 	\intertext{where we used Eq. \eqref{eq:[A,F] and [B,F]}. Since $A_{s_0}^h | \Omega_0 \rangle = | \Omega_0 \rangle$ (Lemma \ref{lem:omega_0(O) = 1}) for all $h \in G$ we then find}
		A^{RC;u_2 u_1}_{\site} F^{RC;u_1 v}_{\rho} | \Omega_0 \rangle &= \left( \frac{\dimR}{|N_C|} \right)^2 \sum_{m, n}  \,  R^{j_2 j_1}(m)^* R^{j_1 j'}(n)^* \, F_{\rho}^{\bar c_{i_1}, q_{i_1} mn \bar q_{i'}} \, | \Omega_0 \rangle \\
      &= \left( \frac{\dimR}{|N_C|} \right)^2 \sum_{m, m'} \sum_l \, R^{j_2 j_1}(m)^* R^{j_1 l}(\bar m)^* R^{l j'}(m')^* \, F_{\rho}^{\bar c_{i_2}, q_{i_2} m' \bar q_{i'}} | \Omega_0 \rangle \\
      &= \left( \frac{\dimR}{|N_C|} \right) \sum_{m} \sum_l R^{j_2 j_1}(m)^* R^{l j_1}(m) \, F_{\rho}^{RC;(i_2, l) v}| \Omega_0 \rangle \\
      &= F_{\rho}^{RC;u_2, v} | \Omega_0 \rangle
    \end{align*}
    where we substituted $m' = mn$ in the second line, and used Eq. \eqref{eq:Schur} in the last step.
\end{proof}

\begin{lemma} \label{lem:D kills the ground state}
	We have
	\begin{equation*}
		D^{RC; u}_s | \Omega_0 \rangle = \delta_{RC, R_1 C_1} | \Omega_0 \rangle
	\end{equation*}
	where $R_1 C_1$ is the trivial representation. \ie $C_1 = \{ 1 \}$, so $N_{C_1} = G$, and $R_1$ is the trivial representation of $G$.
\end{lemma}

\begin{proof}
	 Note that $D_{s}^{R_1 C_1} = A_{v(s)} B_{f(s)}$, so the frustration free ground state satisfies $\omega_0(D_{s}^{R_1 C_1}) = 1$ which proves the claim in the case that $RC = R_1 C_1$. Using Lemma \ref{lem:pure components projector Lemma} and Lemma \ref{lem:DRCprops} we find $\omega_0(D^{RC;u}_s) = \omega_0(D^{RC;u}_s D^{R_1C_1}_s) = 0$. Finally, since $D^{RC;u}_s$ is a projector and $| \Omega_0 \rangle$ is the GNS vector of $\omega_0$, it follows that $D^{RC;u}_s | \Omega_0 \rangle = 0$ (Lemma \ref{lem:omega_0(O) = 1}).
\end{proof}

\begin{lemma} \label{lem:ribbon label projector}
	Let $\rho$ be a finite ribbon with $\partial_0 \rho = \site$. Then
	\begin{equation*}
		D^{RC;u_1}_{\site} F^{RC;u_2 v}_{\rho} | \Omega_0 \rangle = \delta_{u_1, u_2} F^{RC;u_1 v}_{\rho} | \Omega_0 \rangle.
	\end{equation*}
\end{lemma}

\begin{proof}
    The proof is a computation using the basic commutation rules of the $A_{s_0}^h$ and $B_{s_0}^g$ with the ribbon operators (Eq. \eqref{eq:[A,F] and [B,F]}) and the fact that $B_{s_0}^g | \Omega_0 \rangle = \delta_{g, e} | \Omega_0 \rangle$ and $A_{s_0}^h | \Omega_0 \rangle = | \Omega_0 \rangle$ for all $h, g \in G$ (Lemma \ref{lem:omega_0(O) = 1}).
    
    Let $u_1 = (i_1, j_2)$, $u_2 = (i_2, j_2)$ and $v = (i', j')$, then
    \begin{align*}
    D^{RC;u_1}_{\site} F^{RC;u_2 v}_{\rho} | \Omega_0 \rangle  &= \left( \frac{\dimR}{|N_C|} \right)^2 \sum_{m, n \in N_C} \, R^{j_1 j_1}(m)^* R^{j_2 j'}(n)^* A_{s_0}^{q_{i_1} m \bar q_{i_1}} \, B_{s_0}^{c_{i_1}} \, F^{\bar c_{i_2}, q_{i_2} n \bar q_{i'}}_{\rho} \, | \Omega_0 \rangle \\
      &= \left( \frac{\dimR}{|N_C|} \right)^2 \delta_{i_1, i_2} \, \sum_{m, n \in N_C} \, R^{j_1 j_1}(m)^* R^{j_2 j'}(n)^* F^{\bar c_{i_1}, q_{i_1} m n \bar q_{i'}}_{\rho} \, A_{s_0}^{q_{i_1} m \bar q_{i_1}} \, | \Omega_0 \rangle \\
      &= \delta_{i_1, i_2} \, \left( \frac{\dimR}{|N_C|} \right)^2 \sum_{n, n' \in N_C} \, \sum_{l} R^{j_1 l}(n')^* R^{l j_1}(\bar n)^* R^{j_2 j'}(n)^* \, F^{\bar c_{i_1}, q_{i_1} n' \bar q_{i'}}_{\rho} \, |\Omega_0 \rangle \\
      &= \delta_{i_1, i_2} \delta_{j_1, j_2} \, \left( \frac{\dimR}{|N_C|} \right) \sum_{n'} R^{j_1 j'}(n')^* F^{\bar c_{i_1}, q_{i_1} n' \bar q_{i'}}_{\rho} | \Omega_0 \rangle = \delta_{u_1, u_2} \, F^{RC;u_1 v}_{\rho} \, | \Omega_0 \rangle
    \end{align*}
    where we used Schur orthogonality to get the last line.
\end{proof}

\subsection{Properties of $\mu_{\rho}^{RC;uv}$ and $\chi_{\rho}^{RC;uv}$} \label{sec:properties of mu and chi}

\subsubsection{Various actions on the frustration free ground state}

Recall that we defined $\chi_{\rho}^{RC;u_1 u_2} = \pi_0 \circ \mu_{\rho}^{RC;u_1 u_2} : \cstar \rightarrow \caB{(\caH_0)}$ where $(\pi_0, \caH_0, \ket{\Omega_0})$ is the GNS triple of the frustration free ground state $\omega_0$. In the following we will drop $\pi_0$ from the notation. \ie for any $O \in \cstar$ we simply write $O$ instead of $\pi_0(O)$.

\begin{lemma} \label{lem:change ampli label}
	Let $\rho$ be a half-infinite ribbon with $\partial_0 \rho = \site$. For any $O \in \cstar$ we have
	\begin{equation*}
		\chi^{RC;u_1 u_2}_{\rho}( O A^{RC;u_3 u_2}_{\site}  ) | \Omega_0 \rangle = \chi^{RC;u_1 u_3}_{\rho}(O) | \Omega_0 \rangle.
	\end{equation*}
	In particular,
	\begin{equation*}
		\chi^{RC;u_1 u_2}_{\rho}( A^{RC;u_3 u_2}_{\site} ) | \Omega_0 \rangle = \chi^{RC;u_1 u_3}_{\rho}(\I) | \Omega_0 \rangle = \delta_{u_1 u_3} | \Omega_0 \rangle.
	\end{equation*}
\end{lemma}

\begin{proof}
	By definition
	\begin{equation*}
		\chi^{RC;u_1 u_2}_{\rho}( O A_s^{RC;u_3 u_2}) = \lim_{n \uparrow \infty}  \, \bigg(\frac{|N_C|}{\dimR} \bigg)^2 \,  \sum_v  \big( F_{\rho_n}^{RC;u_1 v} \big)^*  \, O A_{\site}^{RC;u_3 u_2} \, F_{\rho_n}^{RC;u_2 v}
	\end{equation*}
	so using Lemma \ref{lem:change ribbon operator label} we get
	\begin{align*}
		\chi^{RC;u_1 u_2}_{\rho}( O A_{\site}^{RC;u_3 u_2}) | \Omega_0 \rangle &= \lim_{n \uparrow \infty} \, \bigg(\frac{|N_C|}{\dimR} \bigg)^2 \,  \sum_v \, \big( F_{\rho_n}^{RC;u_1 v} \big)^*  \, O \,  F_{\rho_n}^{RC;u_3 v} | \Omega_0 \rangle \\
                &= \chi_{\rho}^{RC;u_1 u_3}(O)
	\end{align*}
	as required.

	The last claim follows immediately from the first and item 2 of Lemma \ref{prop:ampli properties}.
\end{proof}

It follows that
\begin{lemma} \label{lem:ampli label projector}
	Let $\rho$ be a half-infinite ribbon with $\partial_0 \rho = \site$. For any $O \in \cstar$ we have
	\begin{equation*}
		\chi^{RC;u_2 u_1}_{\rho}( O D_{\site}^{RC;u_3} ) | \Omega_0 \rangle = \delta_{u_1, u_3} \chi_{\rho}^{RC;u_2 u_1}(O) \, | \Omega_0 \rangle.
	\end{equation*}
\end{lemma}

\begin{proof}
	We have
	\begin{align*}
		\chi^{RC;u_2 u_1}_{\rho}(O D_{\site}^{RC;u_3}) | \Omega_0 \rangle &= \lim_{n \uparrow \infty} \, \bigg(\frac{|N_C|}{\dimR} \bigg)^2 \,  \sum_{v} \big(F_{\rho_n}^{RC;u_2 v} \big)^* \, O D_{\site}^{RC;u_3} \, F_{\rho_n}^{RC;u_1 v} \, | \Omega_0 \rangle \\
		\intertext{Using Lemma \ref{lem:ribbon label projector} this becomes}
							     &= \delta_{u_1, u_3} \, \lim_{n \uparrow \infty} \, \bigg(\frac{|N_C|}{\dimR} \bigg)^2 \, \sum_v \, \big( F_{\rho_n}^{RC;u_2 v} \big) \, O \, F_{\rho_n}^{RC;u_1 v} | \Omega_0 \rangle\\
							     &= \delta_{u_1, u_3} \, \chi_{\rho}^{RC;u_2 u_1}(O) | \Omega_0 \rangle,
	\end{align*}
	where in the last step we used $\chi_{\rho}^{RC;u_1 u_2} = \pi_0 \circ \mu_{\rho}^{RC;u_1 u_2}$ and the definition of $\mu_{\rho}^{RC;u_1 u_2}$ given in Proposition \ref{prop:ampli properties}.
\end{proof}

\begin{lemma} \label{lem:chi preserves constraints}
    Let $\rho$ be a half-infinite ribbon with $\partial_0 \rho = \site$. For any vertex $v \neq v(\site)$ and any face $f \neq v(\site)$ we have
    $$ \chi_{\rho}^{RC;uu}( A_v ) \ket{\Omega_0} = \chi_{\rho}^{RC;uu}(B_f) \ket{\Omega_0} = \ket{\Omega_0}. $$
\end{lemma}

\begin{proof}
    This follows from the definition of $F^{RC;uv}_\rho$, Eq. \eqref{eq:[A,F] and [B,F]}, and the fact that $A_v \ket{\Omega_0} = B_f \ket{\Omega_0} = \ket{\Omega_0}$ for any $v \in \latticevert$ and any $f \in \latticeface$.
\end{proof}

\subsubsection{A tool to prove non-degeneracy of the amplimorphism representation}

We continue to work in the GNS representation $(\pi_0, \caH_0, \ket{\Omega_0})$ of the frustration free ground state $\omega_0$ and again drop $\pi_0$ from the notation.

\begin{definition} \label{def:magic map}
	For any finite ribbon $\rho$ with $\partial_0 \rho = \site$ and any $RC$ and $u, v \in I_{RC}$ define a linear map $t_{\rho}^{RC;u v} : \cstar \rightarrow \cstar$ by
	\begin{equation}
		t_{\rho}^{RC;uv}( O ) :=  \bigg(\frac{\dimR}{|N_C|} \bigg)^2 \, \sum_{w, z} \, F_{\rho}^{RC;u w}  \, O \, \big( F_{\rho}^{RC;z w} \big)^* A_{\site}^{RC;z v} D_{\site}^{RC;v}.
	\end{equation}
\end{definition}

Recall that for a half-infinite ribbon $\rho$ we write $\rho_n$ for the finite ribbons consisting of the first $n$ triangles of $\rho$. We have the following remarkable property:
\begin{lemma} \label{lem:magic map}
	Let $\rho$ be a half-infinite ribbon with $\partial_0 \rho = \site$. For any local $O$ whose support does not intersect $\rho \setminus \rho_n$ we have
	\begin{equation}
		\chi_{\rho}^{RC;u_1 v_1} \big( t_{\rho_n}^{RC;u_2 v_2}(O) \big) | \Omega_0 \rangle = \delta_{u_1 u_2} \delta_{v_1 v_2} \, O | \Omega_0 \rangle.
	\end{equation}
\end{lemma}

\begin{proof}
	We compute
	\begin{align*}
		\chi_{\rho}^{RC;u_1 v_1} \big( t_{\rho_n}^{RC;u_2 v_2}(O) \big) | \Omega_0 \rangle &= \bigg(\frac{\dimR}{|N_C|} \bigg)^2 \, \sum_{w, u}  \chi^{RC;u_1 v_1}_\rho \left(  F_{\rho_n}^{RC;u_2 w} \, O \, \big( F_{\rho_n}^{RC;u w} \big)^* \, A^{RC;u v_2}_\site D^{RC;v_2}_\site  \right) | \Omega_0 \rangle \\
		\intertext{Using Lemma \ref{lem:change ampli label} and Lemma \ref{lem:ampli label projector}, this becomes: }
      &= \delta_{v_1, v_2} \, \bigg(\frac{\dimR}{|N_C|} \bigg)^2 \, \sum_{w, u} \, \chi^{RC;u_1 u}_\rho \left( F_{\rho_n}^{RC;u_2 w} \, O \, \big( F_{\rho_n}^{RC;u w} \big)^* \right) | \Omega_0 \rangle \\
      \intertext{Taking $m > n$ large enough we get}
      &= \delta_{v_1, v_2} \, \sum_{u, v, w}  \, \big( F_{\rho_m}^{RC;u_1 v} \big)^* F_{\rho_n}^{RC;u_2 w} \, O \, \big( F_{\rho_n}^{RC;u w} \big)^* F_{\rho_m}^{RC;u v} | \Omega_0 \rangle \\
      \intertext{Decomposing $\rho_m = \rho_n \rho'$ we get, using Lemma \ref{lem:decomposition of F}}
      &= \delta_{v_1, v_2} \, \bigg(\frac{\abs{N_C}}{\dimR} \bigg)^2 \,  \sum_{u, v, w} \, \sum_{y, z} \, \big(  F_{\rho_n}^{RC;u_1 y} F_{\rho'}^{RC;y v}  \big)^* \, F_{\rho_n}^{RC;u_2 w} \, O  \\
        & \quad\quad\quad\quad\quad \times \, \big( F_{\rho_n}^{RC;u w} \big)^* \, F_{\rho_n}^{RC;u z} F_{\rho'}^{RC;z v} \, | \Omega_0 \rangle \\
        \intertext{Since the support of $O$ does not intersect $\rho \setminus \rho_n \supset \rho'$ we have $[F_{\rho'}^{RC;z v}, O] = 0$. We also have $[F_{\rho'}^{RC;z v},F_{\rho_n}^{RC;u',w'}] = 0$ for all $u',w' \in I_{RC}$ since $\rho'$ and $\rho_n$ are disjoint. We can therefore commute $F_{\rho'}^{RC;z v}$ to the left and get,}
        &= \delta_{v_1, v_2} \, \bigg(\frac{\abs{N_C}}{\dimR} \bigg)^2 \,  \sum_{w} \, \sum_{y, z} \, \bigg(   \sum_{v} \big(F_{\rho'}^{RC;y v}\big)^*\, F_{\rho'}^{RC;z v} \bigg)  \,  \\
        & \quad\quad\quad\quad\quad \times \, \big( F_{\rho_n}^{RC;u_1 y}\big)^* \, F_{\rho_n}^{RC;u_2 w}  O \bigg( \, \sum_u \, \big( F_{\rho_n}^{RC;u w} \big)^* \, F_{\rho_n}^{RC;u z} \bigg) \, | \Omega_0 \rangle \\
      \intertext{Using Lemma \ref{lem:FdaggerF identity}, the sum over $u$ yields a $\delta_{w, z}$ and the sum over $v$ yields a $\delta_{y, z}$ so}
      &= \delta_{v_1, v_2} \,\bigg(\frac{\dimR}{|N_C|} \bigg)^2 \,  \sum_{w} \, \big(  F_{\rho_n}^{RC;u_1 w} \big)^* \, F_{\rho_n}^{RC;u_2 w} \, O \, | \Omega_0 \rangle  \\
      &= \delta_{u_1, u_2} \delta_{v_1, v_2} \, O | \Omega_0 \rangle
	\end{align*}
	as required.
\end{proof}

%% file: String_nets.tex
\section{Properties of string nets}
\label{sec:properties of string nets}

\subsection{Direct paths and flux}

Recall the definitions of direct paths and the direct path of a ribbon from Section \ref{subsec:preliminary notions}.

If a direct path $\gamma$ is supported in a region $S \subset \latticeedge$ and $\al \in \gc[S]$ is a gauge configuration on $S$ then we define the flux of $\al$ through $\gamma$ to be
$$ \phi_{\gamma}(\al) := \prod_{e \in \gamma} \, \al_e, $$
where the product is ordered according to the order of $\gamma$. We have $\phi_{\dash{\gamma}}(\alpha) = \dash{\phi_{{\gamma}}(\alpha)}$ and if $\gamma = \gamma_1 \gamma_2$ then we have $\phi_\gamma(\alpha) = \phi_{\gamma_1} (\alpha) \phi_{\gamma_2} (\alpha)$.

Similarly, we say a finite ribbon $\rho = \{ \tau_i \}_i^l$ is supported in $S \subset \latticeedge$ if for all $i = 1, \cdots, l$ we have $e_{\tau_i} \in S$ or $\bar e_{\tau_i} \in S$. In that case the direct path $\rho^{dir}$ is supported in $S$ and we put $\phi_{\rho}(\al) := \phi_{\rho^{dir}}(\al)$. This is consistent with Definition \ref{def:flux thourgh ribbon}.

\subsection{Fluxes of string-nets}

\begin{definition} \label{def:face-move}
    We say two direct paths $\gamma_1$ and $\gamma_2$ are related by a face-move over $f \in \latticeface$ if $\gamma_1 = \gamma' \gamma^{f}_1 \gamma''$ and $\gamma_2 = \gamma' \gamma^{f}_2 \gamma''$ for direct paths $\gamma', \gamma'', \gamma_1^f$ and $\gamma_2^f$ such that $\gamma_1^f \bar \gamma_2^f$ is a closed direct path consisting of three edges circling the face $f$.
\end{definition}

\begin{lemma} \label{lem:trivial face-move}
    Let $\gamma_1$, $\gamma_2$ be direct paths in $\eregion$ that are related by a face-move over $f \in \fregion$, and let $\al \in \gc$ be such that $B_f \ket{\al} = \ket{\al}$, then
    $$\phi_{\gamma_1}(\al) = \phi_{\gamma_2}(\al).$$
\end{lemma}

\begin{proof}
    From the definition, we have $\gamma_1 = \gamma' \gamma^{f}_1 \gamma''$ and $\gamma_2 = \gamma' \gamma^{f}_2 \gamma''$ for direct paths $\gamma', \gamma'', \gamma_1^f$ and $\gamma_2^f$ such that $\gamma_1^f \bar\gamma_2^f$ is a closed direct path consisting of three edges circling the face $f$.

    It follows from $B_f \ket{\al} = \ket{\al}$ that $\phi_{\gamma_1^f \bar\gamma_2^f }(\al) = 1$, or $\phi_{\gamma^f_1}(\al) = \phi_{\gamma^f_2}(\al)$. It follows that
    $$ \phi_{\gamma_1}(\al) =  \phi_{\gamma'}(\al) \phi_{\gamma^f_1}(\al) \phi_{\gamma''}(\al) = \phi_{\gamma'}(\al) \phi_{\gamma_2^f}(\al) \phi_{\gamma''}(\al) = \phi_{\gamma_2}(\al) $$
    as required.
\end{proof}

Recall the fiducial ribbons $\nu_n$ and boundary ribbons $\beta_n$ defined in Section \ref{subsec:local gauge configurations and boundary conditions}, see Figure \ref{fig:fiducial_and_boundary_ribbons}.
\begin{lemma} \label{lem:bcinC}
    Let $C \in (G)_{cj}$ and $i = 1, \cdots, \abs{C}$. If $\alpha \in \packi$ then we have
    $$\phi_{\bdy}(\al) = \dash{ \phi_{\fidu}(\al)  } \, c_i \, \phi_{\fidu}(\al) \in C.$$
\end{lemma}

\begin{proof}
    Let $\gamma_{\site}$ be the direct path of $\rho_{\triangle}(\site)$ so $\phi_{\gamma_{\site}}(\alpha) = c_i \in C$. Let $\gamma_{\fidu}$ be the direct path of $\fidu$, and $\gamma_{\bdy}$ the direct path of the boundary ribbon $\bdy$. Consider the direct path $\gamma = \gamma_{\site} \gamma_{\fidu} \dash{\gamma_{\bdy}} \, \dash{\gamma_{\fidu}}$.

    Since $\al \in \packi$ we have by definition that $B_f \ket{\al} = \ket{\al}$ for all $f \in \dfregion$. Since $\gamma$ can be shrunk to the empty ribbon through a sequence of face-moves over faces $f \in \dfregion$, it follows from Lemma \ref{lem:trivial face-move} that $\phi_{\gamma}(\al) = \phi_{\emptyset}(\al) = 1$, which is equivalent to
    $$ \phi_{\bdy}(\al) = \phi_{\gamma_{\bdy}}(\al) = \dash{\phi_{\gamma_{\fidu}}(\al)} \, \phi_{\gamma_{\site}}(\al) \,\ \phi_{\gamma_{\fidu}}(\al) = \dash{\phi_{\fidu}(\al)} \, c_i \, \phi_{\fidu}(\al) $$
    as required.
\end{proof}

\begin{lemma} \label{lem:inNC}
    Let $\alpha \in \packib$, then $\bar q_i \phi_{\fidu}(\alpha) q_{i(b)} \in N_C$.
\end{lemma}

\begin{proof}
    From Lemma \ref{lem:bcinC} we have
    $$ \phi_{\bdy}(\al) = \dash{ \phi_{\fidu}(\al)  } \, c_i \, \phi_{\fidu}(\al) = \dash{ \phi_{\fidu}(\al)  } \,q_i \, r_C \, \bar q_i \, \phi_{\fidu}(\al), $$
    in particular, $\phi_{\bdy}(\al) = \phi_{\bdy}(b(\al)) \in C$, so we have a unique label $i(b) \in \{ 1, \cdots, \abs{C} \}$ such that $\phi_{\bdy}(\al) = q_{i(b)} r_C \bar q_{i(b)}$. Usign this we obtain
    $$  r_C = \bar q_{i(b)}\dash{ \phi_{\fidu}(\al)  } \,q_i \, r_C \, \bar q_i \, \phi_{\fidu}(\al) \, q_{i(b)}. $$
    This shows that $\bar q_i \phi_{\fidu}(\al) \, q_{i(b)} \in N_C$, as required.
\end{proof}

\subsection{The action of gauge groups on string nets}
\label{app:Uisfaithful}

Recall the group of gauge transformations $\gauge$ consisting of unitaries of the form $\mathcal{U}(\{g_v\}) = \prod_{v \in \vregion} \, A_v^{g_v}$ with $g_v \in G$ for each $v \in \vregion$. These gauge transformations act in the bulk of $\eregion$, they are all supported on $\eregion \setminus \partial \eregion$.

We define \emph{boundary gauge transformations} acting on $\hilb_n$ in a similar way.

\begin{definition} \label{def:boundary gauge transformations}
    Recall $\partial \vregion = \vregion[n+1] \setminus \vregion$ and let $\partial \gauge$ be the group of unitaries of the form $\mathcal{U}(\{g_v\}) = \prod_{v \in \vertex(\partial \eregion)} \tilde{A}_v^{g_v}$ with $g_v \in G$ for each $v \in \partial \vregion$. Here $\tilde{A}_{v}^g$ is the restriction of $A_v^g$ to $\hilb_n$. We call $\partial \gauge$ the group of boundary gauge transformations.
\end{definition}

Note that the boundary gauge transformations are supported on $\eregion \setminus \eregion[n-1]$.

\begin{figure}
    \centering
    \includegraphics[width=0.5\textwidth]{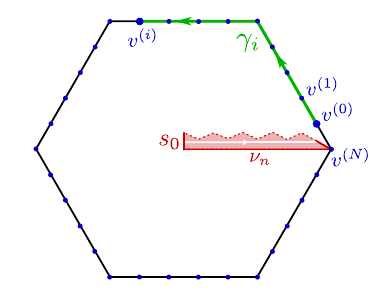}
    \caption{Vertices $v^{(i)}$ and direct paths $\sigma_i$ on the boundary of the region $\eregion$, defined relative to the fiducial ribbon $\fidu$. The region $\eregion$ is indicated by the black hexagon. The individual edges of $\eregion$ are not shown.}
    \label{fig:boundary gauge construction}
\end{figure}

Before proving the Lemma on the free and transitive action of the gauge group $\gauge$, we show the following result, which will help us prove uniqueness statements for gauge transformations.

\begin{lemma} \label{lem:unique gauge tranformations}
    If $\alpha \in \gc$ is any gauge configuration on $\eregion$ and $U \in \gauge$ is such that $U \ket{\alpha} = \ket{\alpha}$, then $U = \mathds{1}$.
\end{lemma}

\begin{proof}
    Since $U \in \gauge$ it is of the form $U = \prod_{v \in \vregion} \, A_v^{g_v}$ for group elements $g_v \in G$. For any edge $e = (v, v')$ with $v \in \partial \vregion$ and $v' \in \vregion$ we have $\alpha_e = \alpha_e \bar g_{v'}$, so $g_v' = 1$. This shows that in fact $U \in \mathcal{G}_{n-1}$. Proceeding inductively, we find that $U \in \mathcal{G}_0$, i.e. $U$ is of the form $U = A_{v_0}^{g_{v_0}}$. Finally, by considering any edge $e = (v_0, v)$ we find that $\alpha_e = g_{v_0} \alpha_e$ so $g_{v_0} = 1$ and $U = \mathds{1}$.
\end{proof}

Recall the boundary gauge transformations $\partial \gauge$ of Definition \ref{def:boundary gauge transformations}.

\begin{lemma} \label{lem:transitive boundary action}
    For any pair of boundary conditions $b, b' \in \bc[C]$ that are compatible with conjugacy class $C$ there is a boundary gauge transformation $U_{b' b} \in \partial \gauge$ such that for any $\alpha \in \packC[C;ib]$ we have $U_{b' b} \ket{\alpha} = \ket{\alpha'}$ for an $\alpha' \in \packC[C;ib']$ that satisfies $\alpha_e = \alpha'_e$ for all $e \in \eregion[n-1]$.
\end{lemma}

\begin{proof}
    Fix a conjugacy class $C$ and a flux $c_i \in C$. We will first prove the claim for the simple boundary condition $b_0$ corresponding to the string net of Figure \ref{fig:simple_string_net}. i.e. let $b$ be an arbitrary boundary condition compatible with $C$. We will construct a boundary gauge transformation $U_{b_0 b} \in \partial \gauge$ such that for any $\alpha \in \packC[C;ib]$ we have $U \ket{\alpha} = \ket{\alpha'}$ with $\alpha' \in \packC[C;ib_0]$ such that $\alpha_e = \alpha'_e$ for all $e \in \eregion[n-1].$

    To that end, let $\partial \vregion = \{v^{(0)}, v^{(1)}, \cdots, v^{(N)}\}$ be a labeling of the vertices in $\partial \vregion$ as in Figure \ref{fig:boundary gauge construction}. For $i = 1, \cdots, N$, let $\sigma_i$ be the direct path proceeding counterclockwise around $\partial \eregion$ from $v^{(0)}$ to $v^{(i)}$ as in the Figure. Let $\sigma_{N+1}$ be the direct path that circles $\partial \eregion$ in a counterclockwise direction starting and ending at $v^{(0)}$.

    Since $\phi_{\sigma_{N+1}}(b) \in C$ it can be written as $\phi_{\sigma_{N+1}}(b) = c_{i'} = q c_i \bar q$ for some $q \in G$. Set $g_{v^{(i)}} := \bar q \, \phi_{\sigma_i}(b)$ and $U_{b_0 b} = \prod_{i = 1}^N \, \tilde A_{v^{(i)}}^{g_{v^{(i)}}} \in \partial \gauge$. Take $\alpha \in \packC[C;ib]$ and let $\alpha'$ be the unique string net such that $U \ket{\alpha} = \ket{\alpha'}$. Then for any $i = 1, \cdots, N-1$ we set $e = (v^{(i)}, v^{(i+1)}) \in \partial \eregion$, and find $$b(\alpha')_e = \alpha'_e = g_{v^{(i)}} \, \alpha_e \, g_{v^{(i+1)}} = \bar q \phi_{\sigma_i}(b)  b_e \bar \phi_{\sigma_{i+1}}(b) q = 1.$$
    Furthermore, for the final boundary edge $e = (v^{(N)}, v^{(0)})$ we have
    $$b(\alpha')_e = \alpha'_e = g_{v^{(N)}} \alpha_e \bar g_{v^{(0)}} = \bar q \phi_{\sigma_{N+1}}(b) q = c_i.$$
    We conclude that $b(\alpha') = b$. Moreover, since $U_{b_0 b}$ is supported on $\eregion \setminus \eregion[n-1]$ we have $\alpha_e = \alpha'_e$ for all $e \in \eregion[n-1]$. This proves the existence part of the claim is this special case.

    Using the same arguments one can show that if $\alpha \in \packC[C;ib_0]$, then $U_{b_0 b}^* \ket{\alpha} = \ket{\alpha'}$ for a string net $\alpha' \in \packC[C;ib_0]$ such that $\alpha_e = \alpha'_e$ for all $e \in \eregion[n-1]$.

    Let us now prove the claim for general boundary conditions $b, b'$ compatible with $C$.
    
    We set $U_{b' b} = U_{b_0 b'}^* U_{b_0 b} \in \partial \gauge$ where the boundary gauge transformations $U_{b_0 b}$ and $U_{b_0 b'}$ are as constructed above. then for any $\alpha \in \packC[C;ib]$ we have $U_{b' b} \ket{\alpha} = \ket{\alpha'}$ for a string net $\alpha' \in \packC[C;ib']$ such that $\alpha_e = \alpha'_e$ for all $e \in \eregion[n-1]$. This proves the general case.
\end{proof}

Recall the Definition \ref{def:string nets Cib(m)} of the collections of string nets $\packC[C;i b](m)$:

Fix a conjugacy class $C$, a boundary condition $b$ compatible with $C$ and a label $i = 1, \cdots, |C|$. For any $m \in N_C$ we have
$$\packib(m) := \{\alpha \in \packib \, : \, \phi_\fidu(\alpha) = q_i m \dash q_{i'}\}$$
where $i' = i(b)$.

\begin{lemma} \label{lem:transitive bulk action}
    For any two $\alpha, \alpha' \in \packib(m)$ there is a unique gauge transformation $U \in \dgauge$ such that $U \ket{\alpha} = \ket{\alpha'}$. Moreover, if $\alpha \in \packib(m)$ and $U \in \dgauge$ then $U \ket{\alpha} = \ket{\alpha'}$ with $\alpha' \in \packib(m)$. i.e. $\dgauge$ acts freely and transitively on $\packib(m)$.
\end{lemma}

\begin{proof}
    Fix a conjugacy class $C$ and a flux $c_i \in C$. We will first prove the claim for the simple boundary condition $b_0$ corresponding to the string net of Figure \ref{fig:simple_string_net}. Denote by $\alpha^{(0)} \in \packC[C;i b_0](1)$ the string net depicted in that figure. It has trivial gauge configuration everywhere except at the red edges. We will first show that for any $\alpha \in \packC[C;ib_0]$ there is a $U \in \dgauge$ such that $U \ket{\alpha} = \ket{\alpha^{(0)}}$.

    Let $v_* \in \partial \vregion$ be the vertex as defined in Figure \ref{fig:simple_string_net}. For any site $v \in \vregion \cup \partial \vregion$, let $\gamma_v$ be a direct path from $v_*$ to $v$ that does not contain any of the red edges (Lemma \ref{lem:trivial face-move}). i.e. $\gamma_v$ is forbidden from crossing the fiducial ribbon. See Figure \ref{fig:simple_string_net} for an example. Define $g_v := \phi_{\gamma_v}(\alpha)$ for all $v \in \vregion \cup \partial \vregion$. Note that since $\alpha$ satisfies the flat gauge condition for all faces except for $f_0$, the group elements $g_v$ are independent of the choice of path $\gamma_v$, as long as we stick to paths that do not include red edges. (This strip acts as a branch cut.) Note further that since the boundary condition is trivial everywhere except on the red boundary edge, we have $g_v = 1$ for all $v \in \partial \vregion$. Moreover, $g_{v_0} = 1$ because we can take $\gamma_{v_0}$ to run along the direct part of the fiducial ribbon as in Figure \ref{fig:paths_I_II}. Since $\alpha \in \packC[C;i b_0](1)$, we have $\phi_{\fidu}(\alpha) = 1$, therefore $g_{v_0} = 1$.

    Let
    $$U = \prod_{v \in \vregion} \, A_v^{g_v} \,\, \prod_{v \in \partial \vregion} \, \tilde A_v^{g_v} = \prod_{v \in \dvregion} \, A_v^{g_b}$$
    where we used that $g_v = 1$ for all $v \in \partial \vregion$ and for $v = v_0$. i.e. we have $U \in \dgauge$.

    We now let $\alpha' \in \packC[C;i b_0]$ be the unique string net such that $U \ket{\alpha} = \ket{\alpha'}$. We will show that $\alpha' = \alpha^{(0)}$.

    Let $e = (v, v')$ be an edge that is not red. Then $\alpha'_e = g_{v} \alpha_e \bar g_{v'} = 1$ because $g_v \alpha_e$ is the flux of $\alpha$ through $\gamma_{v} e$, which is a path from $v^{(N)}$ to $v'$ that does not involve red edges. As noted before, that implies $g_v \alpha_e = g_{v'}$. We see that $\alpha'_e = 1$ for all edges $e$ except possibly the red edges.

    Let us now consider a red edge $e = (v, v')$ which we take to be oriented upwards so that $\alpha^{(0)}_e = c_i$, see Figure \ref{fig:paths_I_II}. Let $I$ be the path from $v_0$ to $v$ and $II$ the path from $v_0$ to $v'$ as shown in Figure \ref{fig:paths_I_II}. Then $\gamma_{v_0} I$ is a path from $v_*$ to $v$ and since $g_{v_0} = 1$ we have $g_{v} = \phi_I(\alpha)$. Similarly, we have $g_{v'} = \phi_{II}$. Let $\gamma_{s_0}$ be the direct path which starts and ends at $v_0$ and circles $f_0$ in a counterclockwise direction. The closed loop $I e \overline{II}$ can be shrunk to $\gamma_{s_0}$ by a sequence of face-moves (Definition \ref{def:face-move}) over faces $f \in \dfregion$. Since $\alpha \in \packC[C;ib_0]$ we have $B_f \ket{\al} = \ket{\al}$ for all $f \in \dfregion$ so it follows from Lemma \ref{lem:trivial face-move} that $c_i = \phi_{\gamma_{s_0}}(\al) = \phi_{I e \overline{II}}(\al) = g_v \al_{e} \bar g_{v'} = \al_e'$.
    
    Let now $\alpha_1, \alpha_2 \in \packC[C;i b_0](1)$ be arbitrary. We have just shown that there are gauge transformations $U_1, U_2 \in \dgauge$ such that $U_1 \ket{\alpha_1} = U_2 \ket{\alpha_2} = \ket{\alpha^{(0)}}$. It follows that the gauge transformation $U = U_2^* U_1 \in \dgauge$ satisfies $U \ket{\alpha} = \ket{\alpha'}$. i.e. we have shown the existence claim in the special case of $\packC[C;i b_0](1)$.

    Let us now generalise to $\alpha_1, \alpha_2 \in \packC[C;i b_0](m)$ for arbitrary $m \in N_C$. i.e. these string nets satisfy $\phi_{\fidu}(\alpha_1) = \phi_{\fidu}(\alpha_2) = q_i m \bar q_{i}$ where we noted that $i(b_0) = i$.

    Acting with the gauge transformation $U_{v_0} = A_{v_0}^{q_i \bar m \bar q_i}$ yields $U_{v_0} \ket{\alpha_1} = \ket{\alpha'_1}$ and $U_{v_0} \ket{\alpha_2} = \ket{\alpha'_2}$ for string nets $\alpha'_1, \alpha'_2 \in \packC[C;i](1)$. Here the flux $c_i$ at $s_0$ was preserved because $q_i \bar m \bar q_i$ commutes with $c_1$, and the action of $A_{v_0}^{q_i \bar m \bar q_i}$ multiplies the flux through $\fidu$ from the left by $q_i \bar m \bar q_i$, thus trivializing it.

    Applying the above result, we have a  gauge transformation $U \in \dgauge$ such that $U \ket{\alpha'_1} = \ket{\alpha'_2}$. Since $U$ commutes with $U^*_{v_0}$ we then find
    $$U \ket{\alpha_1} = U U_{v_0}^* \ket{\alpha'_1} = U_{v_0}^* \ket{\alpha'_2} = \ket{\alpha_2}.$$
    This proves the existence claim in the case of $\packC[C;ib_0](m)$ for arbitrary $m \in N_C$.

    Let us now consider a general boundary condition $b$ that is compatible with $C$. Take $\alpha_1, \alpha_2 \in \packC[C;ib](m) \subset \packC[C;ib]$ for some $m \in N_C$. Lemma \ref{lem:transitive boundary action} provides a boundary gauge transformation $U_{b_0 b} \in \partial \gauge$ such that $U_{b_0 b} \ket{\alpha_1} = \ket{\alpha'_1}$ and $U_{b_0 b} \ket{\alpha_2} = \ket{\alpha'_2}$ for string nets $\alpha'_1, \alpha'_2 \in \packC[C;ib_0](m')$ for some $m' \in N_C$. Here we noted that since $\alpha_1$ and $\alpha_2$ have the same flux through the fiducial ribbon $\fidu$ and both are acted on by the same boundary gauge transformation $U_{b_0 b}$, the resulting string nets $\alpha'_1$ and $\alpha'_2$ also have the same flux through $\fidu$ (though possibly different from the fluxes of $\alpha_1$ and $\alpha_2$).

    Using the result obtained above, we have a gauge transformation $U \in \dgauge$ such that $U \ket{\alpha'_1} = \ket{\alpha'_2}$. since $U$ commutes with $U_{b_0 b}$ we find
    $$U \ket{\alpha_1} = U U_{b_0 b}^* \ket{\alpha'_1} = U_{b_0 b}^* U \ket{\alpha'_1} = U_{b_0 b}^* \ket{\alpha'_2} = \ket{\alpha_2}.$$
    This proves the existence part of the claim in full generality.

    As for uniqueness, take $\alpha_1, \alpha_2 \in \packC[C;ib](m)$ and suppose that $U, U' \in \dgauge$ both satisfy $U \ket{\alpha_1} = U' \ket{\alpha_1} = \ket{\alpha_2}$. Then $U' U^* \ket{\alpha_2} = \ket{\alpha_2}$ and it follows from Lemma \ref{lem:unique gauge tranformations} that $U' = U$.

    It remains to show that if $\alpha \in \packib(m)$ and $U \in \dgauge$, then $U \ket{\alpha} = \ket{\alpha'}$ for an $\alpha' \in \packib(m)$. To see this it is sufficient to note that $U$ is supported on $\eregion \setminus \partial \eregion$ and therefore it cannot change the boundary condition. Further, By Lemma \ref{lem:a and B commute with A_v and B_f} any gauge transformation not supported on $v_0$ commutes with the projectors $B_{\site}^{c_i}$, so $U \in \dgauge$ cannot change the label $i$. Finally, if $v \in \dvregion$ then either no edges incident on $v$ belong to the direct path of the fiducial ribbon, in which case $\phi_{\fidu}(\alpha') = \phi_{\fidu}(\alpha)$ is obvious. Or, precisely two edges incident on $v$ are part of the fiducial ribbon, say $e_i$ and $e_{i+1}$ where we have labeled the direct edges of the fiducial ribbon $\{e_1, \cdots, e_n \}$ along the orientation of $\fidu$. In that case, if $\ket{\alpha'} = A_v^h \ket{\alpha}$, then $\alpha'_{e_i} = \alpha_{e_i} \bar h$ and $\alpha'_{e_{i+1}} = h \alpha_{e_{i+1}}$, and $\alpha$ and $\alpha'$ agree on all other edges. It follows that
    $$\phi_{\fidu}(\alpha') = \prod_{j = 1}^n \alpha'_{e_j} = \prod_{j = 1}^{i-1} \alpha_{e_j}  \times \alpha_{e_i} \bar h \, h \, \alpha_{e_{i+1}} \times \prod_{j = i+2}^n \, \alpha_{e_j} = \phi_{\fidu}(\alpha).$$
    We see that no $A_v^h \in \dgauge$ changes the flux through the fiducial ribbon. Since $\dgauge$ is generated by these on-site gauge transformations, we get the required result. 
\end{proof}

\begin{figure}[t!]
    \centering
    \begin{subfigure}[t]{0.49\textwidth}
    \centering
    \includegraphics[width=\textwidth]{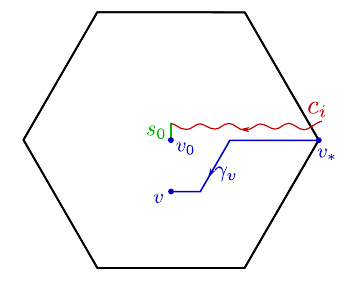}
    \caption{A simple string net $\alpha^{(0)} \in \packC[C;i]$ that is non-trivial only on the dual part of the fiducial ribbon $\fidu$. The corresponding boundary condition $b_0$ is trivial everywhere except at one edge.}
    \label{fig:simple_string_net}      
    \end{subfigure}
    \begin{subfigure}[t]{0.49\textwidth}
    \centering
    \includegraphics[width=\textwidth]{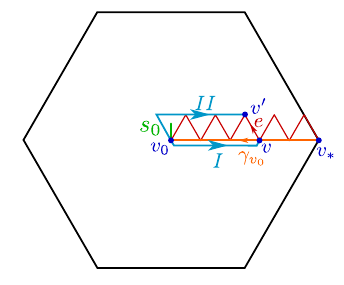}
    \caption{Paths used in the proof of Lemma \ref{lem:transitive bulk action}.}
    \label{fig:paths_I_II}   
    \end{subfigure}
    \caption{}
\end{figure}

\subsection{Action of ribbon operators on string-net states}
\label{sec:action of T and L on string nets}

\begin{lemma} \label{lem:flux projector}
    Suppose $\rho$ is a finite ribbon supported within $S \subset \latticeedge$ and $\al \in \gc[S]$. Then
    $$ T_{\rho}^g \ket{\al} = \delta_{\phi_{\rho}(\al), g} \ket{\al}  $$
    for any $g \in G$. In particular, $[T_{\rho}^g, T_{\rho'}^{g'}] = 0$ for all ribbons $\rho, \rho'$ and any $g, g' \in G$.
\end{lemma}

\begin{proof}
    We prove the Lemma by induction. If $\ep$ is the empty ribbon then $T_{\ep}^g = \delta_{1, g} \I$, which says that the flux through the empty ribbon is always trivial. If $\rho = \{\tau \}$ consists of a single dual triangle then $T_{\rho}^g = F_{\rho}^{1, g} = \delta_{1, g} \I$ which says that the flux through $\rho^{dir} = \emptyset$ is always trivial. If $\rho = \{ \tau \}$ consists of a single direct triangle then $T_{\rho}^g = F_{\rho}^{1, g} = T_{\tau}^g$ which acts on the string net as $T_{\tau}^g \ket{\al} = \delta_{\al_{e_{\tau}}, g} \ket{\al} = \delta_{\phi_{\tau}(\al), g} \ket{\al}$, as required.

    Now suppose $\rho = \rho' \tau$ and suppose the claim is true for the ribbon $\rho'$. Then $T_{\rho}^g = \sum_k \, T_{\rho'}^{k} T_{\tau}^{\bar k g}$. Using the above we get
    $$ T_{\rho}^g \, \ket{\al} = \sum_{k \in G}  \delta_{\phi_{\rho'}(\al), k} \delta_{\phi_{\tau}(\al), \bar k g} \ket{\al} = \delta_{\phi_{\rho}(\al), g} \ket{\al},$$
    as required.

    The commutativity can now be shown as follows. Let $S \subset \latticeedge$ be finite and such that $S$ contains the supports of $T_{\rho}^g$ and $T_{\rho'}^{g'}$. Then for any $\al \in \gc[S]$ we have
    $$ T_{\rho}^g T_{\rho'}^{g'} \ket{\al} = \delta_{\phi_{\rho}(\al), g} \delta_{\phi_{\rho'}(\al), g'} \ket{\al} = T_{\rho'}^{g'} T_{\rho}^g \ket{\al}. $$
    Since $\ket{\al}$ for $\al \in \gc[S]$ is an orthonormal basis for $\caH_S$, the claim follows.
\end{proof}

Let us now consider the boundary ribbon $\bdy$. Its alternating decomposition (cf. Definition \ref{def:alternating decomposition}) $\bdy = I_1J_1 \cdots I_N J_N$ has the direct parts $I_i = \{\tau_i\}$ consisting of a single triangle with $e_{\tau_i} \in \partial \eregion$. The dual parts $J_i$ for $i = 1, \cdots, N-1$ consist of one or two dual triangles each, corresponding to the edges of $\eregion \setminus \partial \eregion$ attached to each boundary vertex in $\partial \vregion$. See Figure \ref{fig:boundary triangles}. For each boundary vertex $v$, let us write $J_v$ for the corresponding dual ribbon. Let us moreover order the boundary vertices $\partial \vregion = \{ v^{(1)}, \cdots, v^{(N)}\}$ counterclockwise as in Figure \ref{fig:counterclockwise boundary labeling}.

\begin{figure}
    \centering
    \includegraphics[width=0.5\textwidth]{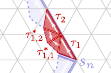}
    \caption{Examples of direct triangles $\tau_i$ and dual triangles $\tau^*_{i, 1}, \tau^*_{i, 2}$ that make up the boundary ribbon $\bdy$. The direct parts $I_1 = \{\tau_1\}$ and $I_2 = \{\tau_2\}$ as well as the first dual part $J_1 = \{ \tau^*_{1, 1}, \tau^*_{1, 2}\}$ of $\bdy$ are depicted.}
    \label{fig:boundary triangles}
\end{figure}

\begin{figure}
    \centering
    \includegraphics[width=0.5\textwidth]{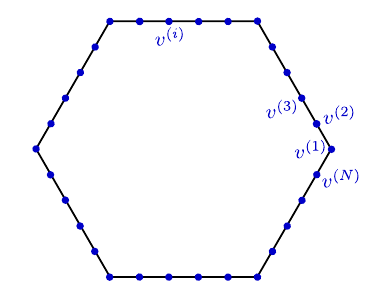}
    \caption{A counterclockwise labeling of the boundary vertices $\partial \vregion$.}
    \label{fig:counterclockwise boundary labeling}
\end{figure}

Let $\alpha \in \gc$ be a gauge configuration on $\eregion$. With the notations just established, it follows from Lemma \ref{lem:L decomposition} that
\begin{equation} \label{eq:boundary L on gauge config}
    L_{\bdy}^h \ket{\alpha} = \prod_{i = 1}^N  L_{J_{v^{(i)}}}^{\bar K_i h K_i}  \ket{\alpha}
\end{equation}
where
$$K_i = \prod_{j = 1}^i \phi_{\tau_i}(\alpha)$$
is the flux of $\alpha$ through $\tau_1 \cdots \tau_i$. Note that $K_N = \phi_{\bdy}(b)$.

We can now prove
\begin{lemma} \label{lem:boundary L on string nets}
    Let $\alpha \in \packC[C;ib]$ and $h \in G$ that commutes with $\phi_{\bdy}(b)$. Then $L_{\bdy}^{h} \ket{\alpha} = \ket{\alpha'}$ for a string net $\alpha' \in \packC[C;ib]$ such that $\phi_{\fidu}(\alpha') = \phi_{\fidu}(\alpha) h$.
\end{lemma}

\begin{proof}
    Using Eq. \eqref{eq:boundary L on gauge config} one easily checks that for each face, except possibly the final one that contains the site $s_n$ (Figure \ref{fig:boundary triangles}), the action of $L_{\bdy}^h$ preserves the trivial flux constraints.

    For that final face, label its edges as in Figure \ref{fig:final face}. From Eq. \eqref{eq:boundary L on gauge config} we see that the operator $L_{\bdy}^h$ acts on the edge degrees of freedom of this triangle as $L_{e_2}^{\bar K_1 \bar h K_1} R_{e_3}^{\bar K_N \bar h K_N}$. On the string net state $\alpha$ this becomes
    $$L_{e_2}^{\bar K_1 \bar h K_1} R_{e_3}^{\bar K_N \bar h K_N} \, \ket{ \alpha_{e_1} } \otimes \ket{ \alpha_{e_2} } \otimes \ket{ \alpha_{e_3} }  = \ket{ \alpha_{e_1} } \otimes \ket{ \bar K_1 \bar h K_1 \alpha_{e_2} } \otimes \ket{ \alpha_{e_3} \bar K_N h K_N }.$$
    
    Noting that $K_N = \phi_{\bdy}(b) $ and $K_1 = \al_{e_1}$ we see that the resulting flux measured at $s_n$ is
    $$ \bar h \al_{e_1} \al_{e_2} \al_{e_3}  \dash{\phi_{\bdy}(b)} \, h \phi_{\bdy}(b) = \bar h \dash{\phi_{\bdy}(b)} \, h \phi_{\bdy}(b)$$
    where we used that $\alpha$ satisfies the trivial flux constraint $\alpha_{e_1} \alpha_{e_2} \alpha_{e_3} = 1$. We now use that $h$ commutes with the boundary flux $\phi_{\bdy}(b)$ to see that the trivial flux condition is also maintained in the final face.

    As already noted, $L_{\bdy}^h$ acts on the degree of freedom at the edge $e_3$ as $R_{e_3}^{\bar K_N \bar h K_N}$, using again that $K_N = \phi_{\bdy}(b)$ and $h$ commutes with $\phi_{\bdy}(b)$ this is the same as $R_{e_3}^{\bar h}$.  since $e_2$ is the final direct edge of the fiducial ribbon $\fidu$, this immediately implies the final claim.
\end{proof}

\begin{figure}
    \centering
    \includegraphics[width=0.2\textwidth]{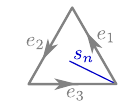}
    \caption{A labeling of the edges of the face $f(s_n)$.}
    \label{fig:final face}
\end{figure}

Recall from Definition \ref{def:eta^Cib(m)} the unit vectors
$$ \ket{\eta_n^{C;ib}(m)} = \frac{1}{\abs{ \packib(m) }^{1/2}} \, \sum_{\al \in \packib(m)} \, \ket{\al}. $$

\begin{lemma} \label{lem:action of T_bdy on eta^C;ib(m)}
    For any $g \in G$, $C \in (G)_{cj}, i = 1, \cdots, \abs{C}, b \in \bc[C]$, and $m \in N_C$ we have
    $$T_{\bdy}^g \ket{\eta_n^{C;ib}(m)} = \delta_{g, q_{i(b)} r_C \, \bar q_{i(b)}} \ket{\eta_n^{C;ib}(m)}. $$
\end{lemma}

\begin{proof}
    By definition, $\ket{\eta_n^{C;ib(m)}}$ is a linear combination of states $\ket{\al}$ with $\al \in \packib(m)$ hence $\phi_{\fidu}(\al) = q_i m \bar q_{i(b)}$ (cf. Definition \ref{def:string nets Cib(m)}) and $\phi_{f_0}(\al) = c_i$. It follows from Lemmas \ref{lem:bcinC} and \ref{lem:inNC} that for all these string nets we have $\phi_{\bdy}(\al) = \bar \phi_{\fidu(\al)} c_i \phi_{\fidu}(\al) = q_{i(b)} r_C \, \bar q_{i(b)}$. The result now follows immediately from Lemma \ref{lem:flux projector}.
\end{proof}

Recall from Definition \ref{def:representation basis} the unit vectors
$$  \qdpure = \left( \frac{\dimR}{\abs{N_C}} \right)^{1/2} \, \sum_{m \in N_C} \, R^{j j'}(m) \, \ket{\eta_n^{C;ib}(m)}$$
where $u = (i, j) \in I_{RC}$ and $v = (b, j') \in I'_{RC}$.

\begin{lemma} \label{lem:topological charge detector}
    We have
    $$ K_{\bdy}^{R_1 C_1} \ket{ \eta_n^{R_2 C_2 ;u v} } = \delta_{R_1 C_1, R_2 C_2}\, \ket{\eta_n^{R_2 C_2; uv}}.  $$
\end{lemma}

\begin{proof}
    Let $u = (i, j) \in I_{R_2 C_2}$ and $v = (b, j') \in I'_{R_2 C_2}$, then
    \begin{align*}
        K_{\bdy}^{R_1 C_1} \, \ket{\eta_n^{R_2 C_2; u v}} &= \sum_{\substack{m_1 \in N_{C_1} \\ m_2 \in N_{C_2}}} \, \chi_{R_1}(m_1)^* \, R_2^{j j'}(m_2)^* \, \sum_{q \in Q_{C_1}} \, L_{\bdy}^{q m_1 \bar q} \, T_{\bdy}^{q r_{C_1} \bar q} \, \ket{\eta_n^{C_2;ib}(m_2)} \\
        \intertext{from Lemma \ref{lem:action of T_bdy on eta^C;ib(m)} and noting that $q r_{C_1} \bar q = q_{i(b)} r_{C_2} \bar q_{i(b)}$ implies $C_1 = C_2$ and $q = q_{i(b)}$ we get}
        &= \delta_{C_1, C_2} \, \sum_{m_1, m_2 \in N_{C_2}} \, \chi_{R_1}(m_1)^* \, R_{2}^{j j'}(m_2)^* \, L_{\bdy}^{q_{b(i)} m_1 \bar q_{i(b)}} \, \ket{ \eta_n^{C_2;ib}(m_2) } \\
        \intertext{noting that $\ket{ \eta_n^{C_2;ib}(m_2) }$ is a linear combination of $\ket{\al}$ for $\al \in \packib(m_2)$, and for each such $\al$ we have $\phi_{\bdy} = q_{i(b)} r_{C_2} \bar q_{i(b)}$, it follows from Lemma \ref{lem:boundary L on string nets} that}
        &= \delta_{C_1, C_2} \, \sum_{m_1, m_2 \in N_{C_2}} \, \chi_{R_1}(m_1)^* \, R_{2}^{j j'}(m_2)^* \, \ket{ \eta_n^{C_2;ib}(m_2 m_1) } \\
        \intertext{changing variables to $M = m_2 m_1$ and $m = m_1$ and writing the character as a trace this becomes}
        &= \delta_{C_1, C_2} \, \sum_{M, m \in N_{C_2}} \, \sum_{l, l'} \, R_1^{l l}(m)^* \, R_2^{j l'}(M)^* \, R_2^{j' l'}(m) \, \ket{\eta_n^{C_2;ib}(M)} \\
        \intertext{Finally, applying Schur orthogonality we get}
        &= \delta_{R_1 C_1, R_2 C_2} \, \ket{\eta_n^{R_2 C_2;uv}},
    \end{align*}
    finishing the proof.
\end{proof}

\subsection{Action of Wigner projectors and label changers on string-net states}

Recall the unit vectors (Definitions \ref{def:eta^Cib(m)}) and \ref{def:representation basis})
$$ \ket{\eta_n^{C;ib}(m)} = \frac{1}{\abs{ \packib(m) }^{1/2}} \, \sum_{\al \in \packib(m)} \, \ket{\al} $$
and
$$  \qdpure = \left( \frac{\dimR}{\abs{N_C}} \right)^{1/2} \, \sum_{m \in N_C} \, R^{j j'}(m)^* \, \ket{\eta_n^{C;ib}(m)}$$
where $u = (i, j) \in I_{RC}$ and $v = (b, j') \in I'_{RC}$.

Recall from Definition \ref{def:Wigner projectors} the Wigner projectors
$$\qd := \frac{\dimR}{|N_C|} \sum_{m\in N_C} \chi_R(m)^* \sum_{q \in Q_C} A_{s} ^{q  m \dash{q}} B_{s}^{q r_C \dash{q}}$$
and for each $u = (i, j) \in I_{RC}$
$$\qd[RC;u] := \frac{\dimR}{|N_C|} \sum_{m \in N_C} R^{jj}(m)^* A_{s} ^{q_i m\dash{q}_i} B_{s}^{c_i}.$$

\begin{lemma} \label{lem:action of Wigner projecitons on string-net condensates}
    We have
    $$ D_{\site}^{R_1C_1} \ket{ \eta_n^{R_2C_2;uv} } = \delta_{R_1C_1, R_2C_2} \ket{\eta_n^{R_2C_2;uv}}. $$
\end{lemma}

\begin{proof}
    Let $u = (i, j) \in I_{RC}$ and $v = (b, j') \in I'_{RC}$. Then
    \begin{align*}
        D_{\site}^{R_1C_1} \ket{\eta_n^{R_2C_2;uv}} &= \bigg( \frac{\dim R_1}{\abs{N_{C_1}}} \bigg)\bigg( \frac{\dim R_2}{ \abs{N_{C_2}} } \bigg)^{1/2} \, \sum_{\substack{ m_1 \in N_{C_1} \\ m_2 \in N_{C_2} }} \, \chi_{R_1}(m_1)^*  R_2^{j j'}(m_2)^* \, \\
        & \quad\quad\quad \times\sum_{q \in Q_{C_1}} \, A_{\site}^{q m_1 \bar q} B_{\site}^{q r_{C_1} \bar q} \, \ket{\eta_n^{C_2;ib}(m_2)} \\
        &= \delta_{C_1, C_2} \, \frac{ \dim R_1 (\dim R_2)^{1/2} }{\abs{N_{C_1}}^{3/2}} \,\sum_{m_1, m_2 \in N_{C_1}} \, \chi_{R_1}(m_1)^*  R_2^{j j'}(m_2)^* \, A_{\site}^{q_i m_1 \bar q_i} \, \ket{\eta_n^{C_1;ib}(m_2)} \\
        \intertext{using Lemma \ref{lem:N_C action on VCib} this becomes}
        &= \delta_{C_1, C_2} \, \frac{ \dim R_1 (\dim R_2)^{1/2} }{\abs{N_{C_1}}^{3/2}} \,\sum_{m_1, m_2 \in N_{C_1}} \, \chi_{R_1}(m_1)^*  R_2^{j j'}(m_2)^*  \ket{\eta_n^{C_1;ib}(m_1m_2)} \\
        \intertext{changing variables to $m = m_1$ and $M = m_1 m_2$, and using using the Schur orthogonality relation \eqref{eq:Schur} we get}
        &= \delta_{R_1C_1, R_2C_2} \, \bigg( \frac{\dim R_1}{\abs{N_{C_1}}} \bigg)^{1/2} \, \sum_{M \in N_{C_1}} \, R_2^{j j'}(M)^*  \ket{\eta_n^{C_1;ib}(M)} \\
        &= \delta_{R_1C_1, R_2C_2} \, \ket{\eta_n^{R_1C_1;uv}}.
    \end{align*}
\end{proof}

\begin{lemma} \label{lem:action of Wigner sub-projector on string-net condensates}
    We have
    $$ D_{\site}^{RC;u_1} \, \ket{\eta_n^{RC;u_2 v}} = \delta_{u_1, u_2} \ket{\eta_n^{RC;u_1 v}}. $$
\end{lemma}

\begin{proof}
    Let $u_1 = (i_1, j_1)$, $u_2 = (i_2, j_2)$ and $v = (b, j')$. Then
    \begin{align*}
        D_{\site}^{RC;u_1} \ket{\eta_n^{RC;u_2 v}} &= \bigg( \frac{\dim R}{\abs{N_C}} \bigg)^{3/2} \, \sum_{m_1, m_2 \in N_C} \, R^{j_1 j_1}(m_1)^* R^{j_2 j'}(m_2)^* \, A_{\site}^{q_{i_1} m_1 \bar q_{i_1}} B_{\site}^{c_{i_1}} \, \ket{\eta_n^{C;i_2 b}(m_2)} \\
        \intertext{noting that $B_{\site}^{c_{i_1}} \, \ket{\eta_n^{C;i_2 b}(m_2)} = \delta_{i_1, i_2}$ and using Lemma \ref{lem:N_C action on VCib} this becomes}
        &= \delta_{i_1, i_2} \, \bigg( \frac{\dim R}{\abs{N_C}} \bigg)^{3/2} \, \sum_{m_1, m_2 \in N_C} \, R^{j_1 j_1}(m_1)^* R^{j_2 j'}(m_2)^* \, \ket{\eta_n^{C;i_2 b}(m_1m_2)} \\
        \intertext{changing variables to $m = m_1$ and $M = m_1 m_2$, and using Schur orthogonality \eqref{eq:Schur} we get}
        &= \delta_{u_1, u_2} \, \bigg( \frac{\dim R}{\abs{N_C}} \bigg)^{1/2} \, \sum_{M \in N_C}  R^{j_2 j'}(M)^* \, \ket{\eta_n^{C;i_2 b}(M)} = \delta_{u_1, u_2} \, \ket{\eta_n^{RC;u_2 v}}.
    \end{align*}
\end{proof}

Recall the operators from Definition \ref{def:label changers}: 
$$  A_{s}^{RC; u_2 u_1} := \frac{\dimR}{|N_C|}\sum_{m \in N_C} R^{j_2j_1}(m)^* A_{s}^{q_{i_2} m \dash{q}_{i_1}} \qquad  u_1 = (i_1,j_1) \quad u_2 = (i_2,j_2)  $$ 
and 
$$  \tilde{A}_n^{RC; v_2 v_1} :=\frac{\dimR}{|N_C|} \sum_{m \in N_C} R^{j'_2j'_1}(m)  U_{b_2 b_1} L_{\bdy}^{{q_{i(b_1)} \dash{m} \,  \dash{q}_{i(b_1)}}} \qquad v_1 = (b_1,j'_1) \quad u_2 = (b_2,j'_2)  $$
where $U_{b_2 b_1}$ is a unitary provided by Lemma \ref{lem:transitive boundary action}, which we choose such that $U_{b_2 b_1} = (U_{b_1 b_2})^*$. It follows from Lemma \ref{lem:transitive boundary action} that the unitary $U_{b_2 b_1}$ yields a bijection between $\packbc[C;ib_1]$ and $\packbc[C;ib_2]$ whenever $b_1, b_2 \in \bc[C]$.

It was shown in Lemma \ref{lem:N_C action on VCib} that the gauge transformations $A_{\site}^{q_i m \bar q_i}$ for $m \in N_C$ yield a left group action of $N_C$ on the vectors $\ket{\eta_n^{C;ib}(m)}$. We show now that the operators $L_{\bdy}^{q_{i(b)} \bar m \bar q_{i(b)}}$ for $m \in N_C$ yield a right action of $N_C$ on these vectors.

\begin{lemma} \label{lem:right action on VCib}
    For any $m_1, m_2 \in N_C$ we have
    $$ L_{\bdy}^{q_{i(b)} \bar m_1 \bar q_{i(b)} } \, \ket{\eta_n^{C;ib}(m_2)} = \ket{\eta_n^{C;ib}(m_2 \bar m_1)}. $$
\end{lemma}

\begin{proof}
    From Lemma \ref{lem:boundary L on string nets} and the fact that $L_{\bdy}^{q_{i(b)} \bar m_1 \bar q_{i(b)}}$ is unitary, we see that this operator yields a bijection from $\packib(m_2)$ to $\packib(m_2 \bar m_1)$. It follows that
    \begin{align*}
        L_{\bdy}^{q_{i(b)} \bar m_1 \bar q_{i(b)} } \, &\ket{\eta_n^{C;ib}(m_2)} = \frac{1}{\abs{\packib(m_2)}^{1/2}} \, \sum_{\al \in \packib(m_2)} \, L_{\bdy}^{q_{i(b)} \bar m_1 \bar q_{i(b)} }\, \ket{\al} \\
        &= \frac{1}{\abs{\packib(m_2 \bar m_1)}^{1/2}} \, \sum_{\al \in \packib(m_2 \bar m_1)} \, \ket{\al} = \ket{\eta_n^{C;ib}(m_2 \bar m_1)}.
    \end{align*}
\end{proof}

We can now show
\begin{lemma} \label{lem:aconverter}
    For any $u, u_1, u_2 \in I_{RC}$ and any $v, v_1, v_2 \in I'_{RC}$ we have
    $$A_{\site}^{RC; u_2 u_1} \qdpure[RC;u_1 v] = \qdpure[RC;u_2 v], \quad \quad \tilde{A}_n^{RC; v_2 v_1} \qdpure[RC;u v_1] = \qdpure[RC; u, v_2]$$
    as well as
    $$(A_{\site}^{RC; u_1 u_2})^* \qdpure[RC;u_1 v] = \qdpure[RC;u_2 v], \quad \quad (\tilde{A}_n^{RC; v_1 v_2})^* \qdpure[RC;u v_1] = \qdpure[RC; u, v_2].$$
\end{lemma}

\begin{proof}
    We prove the claim about the action of $\tilde A_n^{RC;v_2 v_1}$. The claim about $(\tilde{A}_n^{RC; v_1 v_2})^*$ is proven in exactly the same way, and the claims about $A_{\site}^{RC; u_2 u_1}$ and its hermitian conjugate have similar but simpler proofs. Let $u = (i, j)$, $v_1 = (b_1, j'_1)$ and $v_2 = (b_2, j'_2)$, then
    \begin{align*}
        \tilde{A}_n^{RC; v_2 v_1} \qdpure[RC;u v_1] &= \left( \frac{\dimR}{\abs{N_C}} \right)^{3/2} \, \sum_{m_1, m_2 \in N_C} \, R^{j'_2 j'_1}(m_2) \, R^{j j'_1}(m_1)^* \, U_{b_2 b_1} \, L_{\bdy}^{q_{i(b_1)} \bar m_2 \bar q_{i(b_1)}} \, \ket{\eta_n^{C;ib_1}(m_1)} \\
        \intertext{using Lemma \ref{lem:right action on VCib} and the basic properties of $U_{b_2 b_1}$}
        &= \left( \frac{\dimR}{\abs{N_C}} \right)^{3/2} \, \sum_{m_1, m_2 \in N_C} \, R^{j'_2 j'_1}(m_2) \, R^{j j'_1}(m_1)^* \, \ket{\eta_n^{C;ib_2}(m_1 \bar m_2)} \\
        \intertext{letting $M = m_1 \bar m_2$ and $m = m_2$, and using Schur orthogonality, this becomes}
        &= \left( \frac{\dimR}{\abs{N_C}} \right)^{1/2} \, \sum_{M \in N_C}  \, R^{j j'_2}(M)^* \, \ket{\eta_n^{C;ib_2}(M)} = \qdpure[RC;u v_2].
    \end{align*}
\end{proof}

This Lemma tells us that $u$ is a ``bulk" label, as the operator that changes $u_1$ to $u_2$ is $A_{\site}^{RC; u_2 u_1} \in \cstar[{\eregion[1]}]$. We also see that $v$ is a ``boundary" label, as the operator that changes $v_1$ to $v_2$ is $\tilde{A}_n^{RC; v_2 v_1} \in \cstar[{\eregion \setminus \eregion[n-1]}]$.

We can also detect the boundary data by operators supported on $\eregion \setminus \eregion[n-1]$. Recall from Definition \ref{def:boundary condition projector} the projectors $P_b$ supported on $\partial \eregion$ that project onto states with boundary condition $b \in \bc$.

\begin{lemma}
    \label{lem:action of two different label changers on eta}
    For any $v_1,v_2, v \in I'_{RC}$ such that $v_1 = (b_0, j'_1)$ and $v_2 = (b_0, j'_2)$ (\ie they have the same boundary label $b_0$), we have $$(\tilde{A}_n^{RC;v_2 v})^{*} \tilde{A}_n^{RC; v_1 v} \ket{\qdstpure[RC;u v]}= \delta_{v_1v_2}\ket{\qdstpure[RC;u v]}$$
\end{lemma}
\begin{proof}

    Using lemma \ref{lem:aconverter} we get,
    \begin{align*}
        (\tilde{A}_n^{RC;v v_2})^{*} \tilde{A}_n^{RC; v_1 v} \ket{\qdstpure[RC;u v]} &= (\tilde{A}_n^{RC;v_2 v})^{*}  \ket{\qdstpure[RC;u v_1]}\\
        &= \left( \frac{\dimR}{\abs{N_C}} \right)^{3/2} \sum_{m,m' \in N_C} R^{j'_2 j'}(m')^* R^{jj'_1}(m)^*  \ket{\eta_n^{C;ib}(mm')}\\
        \intertext{Now we relabel $mm' = M$ and use Schur orthogonality to get}
        &=  \left( \frac{\dimR}{\abs{N_C}} \right)^{1/2} \sum_{M \in N_C} \sum_{ j_3'} \delta_{j_3' j} \delta_{j'_2 j_1'}  R^{j'_3 j'}(M)^*  \ket{\eta_n^{C;ib}(M)}\\
        &= \delta_{j_1', j'_2}  \ket{\qdstpure[RC;u v]} = \delta_{v_1, v_2}  \ket{\qdstpure[RC;u v]}
    \end{align*}
\end{proof}